\documentclass[11pt]{article}

\usepackage{graphicx, amssymb, amsmath, fullpage, dsfont, bm, epsfig, multirow, amsthm, cite}

\date{}

\newcommand{\captionfonts}{\small}

\makeatletter  
\long\def\@makecaption#1#2{%
  \vskip\abovecaptionskip
  \sbox\@tempboxa{{\captionfonts #1: #2}}%
  \ifdim \wd\@tempboxa >\hsize
    {\captionfonts #1: #2\par}
  \else
    \hbox to\hsize{\hfil\box\@tempboxa\hfil}%
  \fi
  \vskip\belowcaptionskip}
\makeatother   

\DeclareMathOperator{\arccot}{ArcCot}

\begin{document}

\title{\textbf{Quantum Computation of Scattering \\ in Scalar 
Quantum Field Theories}}
\author{{\Large Stephen P.\ Jordan,{\normalsize $^{\dag\S}$} 
		Keith S.\ M.\ Lee,{\normalsize $^{\ddag\S}$} 
        	and John Preskill {\normalsize $^{\S}$} 
\thanks{\texttt{stephen.jordan@nist.gov}, \texttt{ksml@theory.caltech.edu},
\texttt{preskill@theory.caltech.edu}} } \\[20pt]
\textit{ $^{\dag}$ National Institute of Standards and Technology, 
	       Gaithersburg, MD 20899 } \\[8pt] 
\textit{$^{\ddag}$ University of Pittsburgh, Pittsburgh, PA 15260 }\\[8pt] 
\textit{$^{\S}$ California Institute of Technology, Pasadena, CA 91125}}

\maketitle

\bibliographystyle{unsrt}

\newcommand{\ud}{\mathrm{d}}
\newcommand{\braket}[2]{\langle #1|#2\rangle}
\newcommand{\bra}[1]{\langle #1|}
\newcommand{\ket}[1]{|#1\rangle}
\newcommand{\Bra}[1]{\left<#1\right|}
\newcommand{\Ket}[1]{\left|#1\right>}
\newcommand{\Braket}[2]{\left< #1 \right| #2 \right>}
\renewcommand{\th}{^\mathrm{th}}
\newcommand{\tr}{\mathrm{Tr}}
\newcommand{\id}{\mathds{1}}
\newcommand\T{\rule{0pt}{2.6ex}}
\newcommand\B{\rule[-1.2ex]{0pt}{0pt}}
\newcommand{\nn}{\nonumber}
\newcommand{\eq}[1]{(\ref{#1})}
\newcommand{\sect}[1]{\S\ref{#1}}

\newtheorem{theorem}{Theorem}
\newtheorem{proposition}{Proposition}

\begin{abstract}
 Quantum field theory provides the framework for the most fundamental
 physical theories to be confirmed experimentally and has enabled
 predictions of unprecedented precision. However, calculations of
 physical observables often require great computational complexity and
 can generally be performed only when the interaction strength is
 weak. A full understanding of the foundations and rich consequences
 of quantum field theory remains an outstanding challenge.  We develop
 a quantum algorithm to compute relativistic scattering amplitudes in
 massive $\phi^4$ theory in spacetime of four and fewer
 dimensions. The algorithm runs in a time that is polynomial in the
 number of particles, their energy, and the desired precision, and
 applies at both weak and strong coupling.  Thus, it offers
 exponential speedup over existing classical methods at high precision
 or strong coupling. 
\end{abstract}

\newpage

\setcounter{tocdepth}{2}
\tableofcontents

\newpage

\section{Introduction}

Quantum field theory, by which quantum mechanics and special relativity
are reconciled, applies quantum mechanics to functions of space and
time --- fields. In the quantum mechanics of point particles,
dynamical variables become operators obeying suitable commutation
relations, and the wavefunction $\psi(\mathbf{x})$ encodes the
probability that a particle is at position $\mathbf{x}$.
In quantum field theory, fields are promoted to operators satisfying
commutation relations, and the wavefunction specifies the probability
of each field configuration.

Quantum field theory is the basis for the most fundamental physical theory 
to be established by experiment, namely, the Standard Model of particle
physics. This quantum field theory encompasses the strong, weak, and
electromagnetic forces --- all known forces except gravity --- 
and explains all subatomic processes observed so far. 
Quantum field theory has also enabled predictions of unparalleled precision.
For example, the most precise determination of 
the fine-structure constant (a measure of the strength of
the electromagnetic force) uses quantum electrodynamics (QED)
theory and the electron magnetic moment; it gives \cite{Hanneke}
\begin{equation}
\alpha_{\rm em}^{-1} = 137.035\,999\,084\,\,[0.37 \, \textrm{ppb}] \,, 
\end{equation}
in excellent agreement with the next most precise independent 
determinations \cite{Clade, Gerginov}.

To compare the predictions of quantum field theories with experimental
results from particle accelerators, one typically calculates
scattering amplitudes, using a perturbative expansion in powers of the
coupling constant, which quantifies the strength of interactions
between particles. These perturbative calculations are greatly
simplified through the use of Feynman diagrams.  Nevertheless, the
number of Feynman diagrams needed to evaluate a scattering process
involving $n$ particles scales factorially with both the order in the
perturbative expansion and $n$.  Many-particle scattering processes
can be phenomenologically important: for example, in a strongly
coupled quantum field theory, the collision of two highly relativistic
particles (with momentum much larger than mass, $|\mathbf{p}| \gg m$) can
produce a shower of $n_{\mathrm{out}} \sim |\mathbf{p}|/m$ outgoing
particles.

In cases where the coupling constant is strong, such as at low
energies in quantum chromodynamics (QCD), the quantum field theory governing 
nuclear physics, perturbative techniques break down. Strongly coupled 
quantum field theories can be studied on supercomputers via a 
method called lattice field theory. However, lattice field 
theory is useful only for computing static quantities such as mass
ratios and cannot predict dynamical quantities such as scattering
amplitudes.

In this work we develop a quantum algorithm to calculate scattering
amplitudes in a massive quantum field theory with a quartic
self-interaction, called $\phi^4$ theory. The complexity of our
algorithm is polynomial in the time and volume being simulated, the
number of external particles, and the energy scale. The algorithm
enables the calculation of amplitudes to arbitrarily high  precision,
and at strong coupling.
Specifically, the asymptotic scaling of the number of quantum gates,
$G_{\rm weak}$, required to sample from scattering probabilities
in the weakly coupled, $(d+1)$-dimensional theory with 
an error of order $\epsilon$ 
is (\sect{efficiency})
\begin{equation*} 
G_{\mathrm{weak}} = 
\left\{ \begin{array}{ll}
\left( \frac{1}{\epsilon} \right)^{1.5+o(1)} \,, & d=1 \,,\\
\left( \frac{1}{\epsilon} \right)^{2.376+o(1)} \,, & d=2 \,,\\
\left( \frac{1}{\epsilon} \right)^{5.5+o(1)} \,, & d=3 \,.
\end{array} \right.
\end{equation*}
For the strongly coupled theory, the asymptotic scaling is given in
Table~\ref{strongtable} of \sect{efficiency}.
The minimum number of (perfect) qubits required for a non-trivial 
simulation in $D=2$ is estimated --- necessarily crudely --- 
to be on the order of a thousand to ten thousand, corresponding to 
$\epsilon$ ranging from $10\%$ to $1\%$ (see Appendix~F).

There are at least three motivations for this work. First, there is a
familiar motivation: the search for new problems that can be solved
more efficiently by a quantum algorithm than by any existing classical
algorithm. Secondly, one would like to know the physical limits of
computation. What is the computational power of our universe? 
Since our most complete theory yet is the Standard Model, this is a
question about the computational power of quantum field
theories. Thirdly, we wish to learn more about the nature and
foundations of quantum field theory itself. Indeed, history provides a
precedent: Wilson discovered deep insights --- of great conceptual and
technical significance --- about quantum field theory through his work
on simulating it on classical computers \cite{Wilson_Nobel}.

The extensive efforts that are currently being made to build 
a quantum computer are largely motivated by the
discovery of a small number of quantum algorithms that can solve
certain computational problems in exponentially fewer steps (as a
function of input size) than the best existing classical
algorithms. The most famous example is Shor's quantum algorithm for 
prime factorization of an $n$-bit integer in $O(n^2 \log n \log \log n)$
steps \cite{Shor_factoring}. In comparison, the fastest known classical 
algorithm uses $2^{O(n^{1/3} (\log n)^{2/3})}$ steps \cite{Lenstra}. 
There has now been a large amount
of effort devoted to finding additional quantum algorithms
offering exponential speedup. This turns out to be very difficult. At
present, only relatively few problems have been discovered admitting
exponential quantum speedup~\cite{algorithm_zoo}, including simulating 
nonrelativistic many-body quantum systems, estimating certain topological 
invariants such as Jones polynomials, and solving certain number-theory 
problems such as Pell's equation.

The theory of classical computation was put on a firm foundation by
proofs that Turing machines, possibly supplemented by random-number
generation, can efficiently simulate classical systems and, in
particular, all other proposed models of universal classical 
computation. Thus, computational-complexity results obtained within the
Turing machine model, or any of the equivalent models, are universally
applicable, up to polynomial factors. The field of quantum computing
originated with Feynman's observation that quantum-mechanical systems
represent an apparent exception to this efficient mutual simulability:
simulating quantum-mechanical systems with $n$ degrees of freedom on a
classical computer requires time scaling exponentially in $n$
\cite{Feynman}. Feynman conjectured that this exception was an
indication of two classes of computation: classical and
quantum, with all quantum models able to simulate each other efficiently, 
and all classical models able to simulate each other, but not quantum
models, efficiently.

In 1985, Deutsch defined quantum Turing machines and proposed them as a
concrete model of a universal quantum computer able efficiently to 
simulate all others~\cite{Deutsch85}. Subsequently proposed models of quantum
computation, such as the quantum circuit model \cite{Deutsch89},
continuous \cite{Continuous_walk} and discrete \cite{Discrete_walk}
quantum walks, topological quantum computers \cite{Kitaev97}, and
adiabatic quantum computers \cite{Farhi_adiabatic}, have all been
proven to be polynomially equivalent to the quantum Turing machine
model \cite{Yao, Childs_equivalence, FKW, FLS, Aharonov_equivalence, 
  Bernstein_Vazirani}. 
Today, the model
of quantum computation most widely used by theorists is the quantum
circuit model. Feynman's original question about whether a universal
quantum computer can efficiently simulate not only other models of
quantum computation but also naturally occurring quantum systems has
also been addressed: quantum circuits have been designed to simulate
the evolution for time $t$ of systems of $n$ nonrelativistic particles
using $\mathrm{poly}(n,t)$ gates
\cite{Zalka, Abrams_Lloyd, Wiesner, Lloyd_science}.

Whether quantum circuits can efficiently simulate quantum field
theories has remained an open question. 
Interactions between fields are nonlinear (that is, there are 
nonlinearities in the equations of motion), so the dynamics of
an interacting theory is non-trivial to compute. 
Interactions are also local, occurring at a point in spacetime, and 
hence there are infinitely many degrees of freedom even within a finite 
volume. 
Thus, known techniques for simulating the dynamics of
many-body quantum-mechanical systems do not directly apply. If quantum
field theories could not be simulated in polynomial time by quantum
circuits, this would undermine the status of quantum circuits as a
universal model of quantum computation and raise the possibility of
quantum field computers, exponentially faster for certain tasks than
even quantum circuits.

Known classical methods for simulating quantum field theory require, on
worst-case instances, exponential time to compute scattering amplitudes.
Moreover, complexity theory provides reasons to believe that this is a
fundamental barrier for all classical algorithms. It is generally believed 
that, in principle, a quantum computer could be built from known forms of
matter. (The obstacles are mainly related to achieving high precision and good 
isolation from noise.) Thus, the quantum field theory governing known forms
of matter (the Standard Model) should be sufficiently rich to model an 
idealized universal quantum computer. 
Hence, polynomial-time classical simulation of the Standard Model 
should be impossible, unless all quantum computations can
be efficiently performed by classical computers. 
One could restate this claim in the language of complexity theory as a formal 
conjecture: the problem of simulating the Standard Model is 
$\mathsf{BQP}$-hard and therefore cannot be contained in 
$\mathsf{P}$ unless $\mathsf{P}=\mathsf{BQP}$.

When one simulates classical physics, it is standard to discretize space
using a lattice. The same thing can be done in quantum field theory,
and this method underlies the lattice field theory calculations
extensively used in supercomputer studies of quantum
chromodynamics. To simulate a typical process at energy scale $E$,
it is believed to suffice if one chooses a lattice spacing small
compared with $\frac{\hbar c}{E}$. However, discretizing a quantum field
theory involves special difficulties unfamiliar from the classical
context. A discretized system from classical physics converges to a
continuum limit when one straightforwardly takes the lattice spacing to
zero. In contrast, continuum limits of quantum field theories
  are achieved only through careful adjustment of the
parameters of the Hamiltonian (or, equivalently, Lagrangian) as a function
of the lattice spacing. This is the process of renormalization, which
we discuss in \sect{continuum}, and it is an important
consideration in the analysis of quantum (or classical) algorithms for
simulating quantum field theories. In \sect{sec:a} we use methods
of effective field theory to make a detailed analysis of the errors
introduced to our simulation by discretization.

Once discretized, a quantum field theory becomes essentially an ordinary 
many-body quantum-mechanical system, whose evolution can be efficiently 
simulated on quantum computers with established methods \cite{Byrnes}. 
However, a full simulation includes, in addition to time evolution, 
preparation of a physically meaningful initial state, and final measurement 
of physically relevant observables. The state preparation and
measurements depend strongly on the underlying physics and must be analyzed
on a case-by-case basis. In \sect{algorithm} we propose an
adiabatic method for preparing wavepacket states of $\phi^4$ theory,
and in \sect{preparing} we analyze its complexity. 
Adiabatic state preparation is widely used to prepare energy eigenstates
in quantum-simulation algorithms (see, for example, \cite{aspuru2005}).
Our extension of this technique counteracts the dynamical phases
associated with the different energy eigenstates in superposition in a
wavepacket. 
We construct observables for measuring the outcome of scattering events
in \sect{measurement} and \sect{detectors}. 
Our analysis of state preparation and discretization are two of the most 
important technical contributions of this paper.

The issue of gauge symmetries in quantum simulation of lattice
field theories has been addressed in \cite{Byrnes}. There is an extensive 
literature on analog simulation of interacting quantum field theories
using ultracold atoms \cite{Buchler, Zohar, Szirmai:2011pj,
  IgnacioCirac:2010us, Mazza, Kapit, Bermudez, Maraner, Lepori, Maeda,
  Rapp, Weimer}, trapped ions \cite{Casanova, Casanova2}, and 
Josephson-junction arrays \cite{Doucot}. Much work has also been done on
analog simulation of special-relativistic quantum-mechanical effects
such as Zitterbewegung and the Klein paradox, as well as
general-relativistic quantum effects such as Hawking radiation. For
recent reviews, see \cite{Lewenstein, Johanning}. 
Our work, in contrast to these previous studies, addresses digital quantum 
simulation, with explicit consideration of convergence to the
continuum limit and efficient preparation of wavepacket states for the
computation of dynamical quantities such as scattering probabilities.
Our analysis includes error estimates of all parts of our algorithm.

The experimental implementations of Hamiltonians approximating lattice
theories could be viewed as specialized quantum
computers. However, they differ from the quantum circuit model
considered here in that they are a form of analog, rather than digital,
quantum computation. The proofs that digital quantum computers can in
principle perform universal, scalable, fault-tolerant computation
(e.g. \cite{Aharonov_Benor, Knill_Laflamme_Zurek, Kitaev97}) do
not apply to analog quantum computers. The studies of analog 
and digital quantum simulation thus serve complementary
purposes. By studying digital quantum simulation we may understand the
ultimate abilities and limits of future large-scale quantum computers,
whereas analog quantum simulation offers more specialized and less
scalable techniques that are implementable with present-day technology.

Some of the techniques discussed in this paper may be useful for designing 
quantum algorithms to do things other than the simulation of high-energy 
physics.  In particular, suppose one wishes to simulate the
quantum dynamics of a crystal lattice. One could do this directly using
the standard quantum algorithms for quantum simulation, and the
resources would scale polynomially in the number of lattice
sites. However, if the phenomena being studied are characterized by
much longer length scales than the crystal lattice spacing, this
direct approach may be unnecessarily costly. Instead, one could use
the renormalization group to obtain an effective Hamiltonian on a
coarse-grained lattice and then simulate that using the standard
techniques.

\section{Background}

Our presentation is intended to be self-contained and accessible
to a broad audience.\footnote{For a concise description of the main
ideas, see \cite{shortversion}.} 
Section \ref{sec:bg1} describes some basic concepts of quantum field theory
drawn upon in this work.  
For simplicity, the discussion is restricted to the scalar $\phi^4$ theory, 
which is the focus of this paper.
Section \ref{QC} gives an overview of the standard techniques in 
quantum computing that are used in our algorithm.
For reference, notation introduced in this and subsequent sections 
is tabulated in Appendix~A.

\subsection{Quantum Field Theory}
\label{sec:bg1}

The scalar $\phi^4$ quantum field theory in the continuum limit
is relevant to diverse physical phenomena.
In the Standard Model of particle physics, the Higgs boson endows
mass to other particles. Its self-interactions are described by the
(four-component) $\phi^4$ theory. Additionally, the Euclidean versions
of $\phi^4$ theories have been highly successful in describing
critical phenomena and universality in statistical-mechanical
systems.

The simplicity of $\phi^4$ theory makes it suitable for studying
challenging formal aspects of quantum field theory.
The existence of the $\phi^4$ continuum theory has been rigorously 
established in two
\cite{Glimm:1968kh,Glimm:1986ki,Glimm:1986kj,Glimm:1972kn,Glimm:1986kk} 
and three \cite{Glimm:1973kp,Feldman:1976im} (spacetime) dimensions.
The perturbative expansion for the scattering matrix has been shown
to be asymptotic in two \cite{Osterwalder:1975zn,Eckmann:1976xa}
and three \cite{Constantinescu:1977xr} dimensions.
Furthermore, in those dimensions the Schwinger functions (Euclidean 
correlation functions) can be reconstructed from their perturbative
expansions, that is, they are Borel summable  
\cite{Eckmann:1975,Magnen:1977ha,Eckmann:1979pr}. 
In contrast, in five or more dimensions, $\phi^4$ theory is trivial
\cite{Aizenman:1982ze,Frohlich:1982tw}, that is, the continuum limit
is equivalent to a free theory.
In the borderline case of four dimensions, obtaining a rigorous proof
is still an open question, but the theory is believed to be trivial.
Despite this, it is still an important theory: it is simply regarded
as an effective field theory valid below some high-energy cutoff.

\medskip

One possibility for attempting to construct a continuum theory in an 
infinite volume is to define a theory on a finite lattice and then 
study the limits as the lattice spacing is taken to zero and the volume 
to infinity. In the following, we present a lattice theory and its
continuum limit in a manner directly corresponding to the basis of our
algorithm. From a field theorist's perspective, the setup is therefore
unorthodox in several respects. Only the spatial dimensions of the
(Minkowski-space) theory are put on a lattice, as opposed to all
spacetime dimensions of a Euclidean theory. The Schr\"{o}dinger
picture is used rather than the interaction picture, and the
Hamiltonian formalism rather than the Lagrangian formalism. As is
conventional in quantum field theory, we use units where 
$\hbar = c = 1$, so that all quantities have units of some power of
mass. At the end, any necessary factors of $\hbar$ and $c$
can be reinserted by dimensional analysis.

\subsubsection{Lattice $\phi^4$ Theory}
\label{lattice}
Let $\Omega$ be a cubic lattice in $d$ spatial dimensions with length $L$,
volume $V = L^d$, lattice spacing $a$, and $\mathcal{V} = (L/a)^d$
lattice sites. Also let $\hat{L} = L/a$. 
For simplicity, we assume periodic boundary conditions, that is,
\begin{equation}
\Omega = a \mathbb{Z}_{\hat{L}}^d.
\end{equation}
For each $\mathbf{x} \in \Omega$, let $\phi(\mathbf{x})$ be a
continuous, real degree of freedom interpreted as the field at
$\mathbf{x}$, and $\pi(\mathbf{x})$ the canonically conjugate variable to 
$\phi(\mathbf{x})$. 
By canonical quantization, these
degrees  of freedom are promoted to Hermitian operators with the
following commutation relations:
\begin{eqnarray}
\left[ \phi(\mathbf{x}), \pi(\mathbf{y}) \right] 
& = & i a^{-d} \delta_{\mathbf{x},\mathbf{y}} \id \,, \nonumber \\
 \label{canonical}
[\phi(\mathbf{x}), \phi(\mathbf{y})] = [\pi(\mathbf{x}),
  \pi(\mathbf{y})] & = & 0 \,.
\end{eqnarray}
Here, the factor of $a^{-d}$ provides dimensions of
$(\mathrm{length})^{-d}$, matching a Dirac delta function in
the $d$-dimensional continuum. $\phi^4$ theory on the lattice $\Omega$
is defined by the following Hamiltonian:
\begin{equation}
\label{ham}
H = \sum_{\mathbf{x} \in \Omega} a^d \left[ \frac{1}{2} \pi(\mathbf{x})^2 
+ \frac{1}{2} \left( \nabla_a \phi \right)^2 (\mathbf{x}) + \frac{1}{2}
m_0^2 \phi(\mathbf{x})^2 + \frac{\lambda_0}{4!} \phi(\mathbf{x})^4 \right]
\,.
\end{equation}
$\nabla_a \phi$ denotes a discretized derivative, so that
\begin{equation}
\left( \nabla_a \phi \right)^2 (\mathbf{x}) = \sum_{j=1}^d \left(
\frac{\phi(\mathbf{x}+ a\hat{r}_j) - \phi(\mathbf{x})}{a} \right)^2,
\end{equation}
where $\hat{r}_1,\ldots,\hat{r}_d$ are unit vectors such that $a
\hat{r}_1,\ldots, a \hat{r}_d$ form the basis of the 
lattice $\Omega$. The ground state of $H$ is called the vacuum and denoted
$\ket{\mathrm{vac}}$. It can be interpreted as the state in which
there are no particles.  The covariance matrix
\begin{equation}
\label{propagator}
G(\mathbf{x}-\mathbf{y}) = \bra{\mathrm{vac}} \phi(\mathbf{x}) \phi(\mathbf{y})
\ket{\mathrm{vac}}
\end{equation}
is very useful in quantum field theory and is referred to as
the equal-time propagator. 
In the free theory ($\lambda_0 = 0$), the equal-time propagator, $G^{(0)}$, 
can be exactly evaluated.
Let $\Gamma$ be the momentum-space lattice
corresponding to $\Omega$:
\begin{equation}
\Gamma = \frac{2 \pi}{L} \mathbb{Z}_{\hat{L}}^d.
\end{equation}
Then,
\begin{equation}
\label{freepropagator}
G^{(0)}(\mathbf{x} - \mathbf{y}) = \sum_{\mathbf{p} \in \Gamma} \frac{1}{L^d}
\frac{1}{2 \omega(\mathbf{p})} e^{i \mathbf{p} \cdot (\mathbf{x} -
  \mathbf{y})},
\end{equation}
where
\begin{equation}
\omega(\mathbf{p}) = \sqrt{m_0^2 + \frac{4}{a^2} 
\sum_{j=1}^d \sin^2 \left( \frac{ a \mathbf{p} \cdot \hat{r}_j}{2} \right)}.
\end{equation}
Note that $\sum_{\mathbf{x} \in \Omega} a^d \to \int d^dx$ as $a \to 0$, and 
$\sum_{\mathbf{p} \in \Gamma} \frac{1}{L^d} \to \int \frac{d^d p}{(2 \pi)^d}$
as $L \to \infty$.

Let $H^{(0)}$ be the Hamiltonian of \eq{ham} with $\lambda_0 = 0$. One
can exactly solve for the spectrum of $H^{(0)}$ by rewriting it in
terms of creation and annihilation operators, $a_{\mathbf{p}}^\dag$ and
$a_{\mathbf{p}}$. 
For each $\mathbf{p} \in \Gamma$, let
\begin{equation}
a_{\mathbf{p}} = \sum_{\mathbf{x} \in \Omega} a^d 
e^{-i\mathbf{p} \cdot \mathbf{x}} 
\left[ \sqrt{\frac{\omega(\mathbf{p})}{2}} \phi(\mathbf{x}) 
+ i \sqrt{\frac{1}{2 \omega(\mathbf{p})}} \pi(\mathbf{x}) \right].
\label{a}
\end{equation}
One can verify that
\begin{eqnarray}
\phi(\mathbf{x}) & = & \sum_{\mathbf{p} \in \Gamma} \frac{1}{L^d} 
e^{i\mathbf{p} \cdot \mathbf{x}}
\sqrt{\frac{1}{2 \omega(\mathbf{p})}} \left( a_{\mathbf{p}} +
a^\dag_{-\mathbf{p}} \right) \,, 
 \\ \label{phidags}
\pi(\mathbf{x}) & = & -i \sum_{\mathbf{p}
  \in \Gamma} \frac{1}{L^d}  e^{i \mathbf{p} \cdot \mathbf{x}} 
\sqrt{ \frac{\omega(\mathbf{p})}{2}} \left( a_{\mathbf{p}} -
a^\dag_{-\mathbf{p}} \right) \,,
 \label{pidags} \\
H^{(0)} & = & \sum_{\mathbf{p} \in \Gamma} \frac{1}{L^d} \omega(\mathbf{p})
a_{\mathbf{p}}^\dag a_{\mathbf{p}} + E^{(0)} \id \,,
\label{h0dags} \\
\big[ a_{\mathbf{p}}, a_{\mathbf{q}}^\dag \big] & = & L^d
\delta_{\mathbf{p},\mathbf{q}} \id \,, \label{laddercom1} \\
\left[ a_{\mathbf{p}}, a_{\mathbf{q}} \right] & = & 0 \,, \label{laddercom2}
\end{eqnarray}
with $E^{(0)} = \sum_{\mathbf{p} \in \Gamma} \frac{1}{2} \omega(\mathbf{p})$. 
In the limit $L \to \infty$ we have $[ a_{\mathbf{p}}, a_{\mathbf{q}}^\dag ] \to
(2 \pi)^d \delta^{(d)}(\mathbf{p}-\mathbf{q})$, which is the standard
relation for bosonic creation and annihilation operators in quantum
field theory. 

Let $\ket{\mathrm{vac}(0)}$ denote the ground state of $H^{(0)}$
(``the free vacuum''). By \eq{h0dags}, \eq{laddercom1}, and
\eq{laddercom2}, one sees that $\ket{\mathbf{p}} \equiv 
L^{-d/2} a^\dag_{\mathbf{p}} \ket{\mathrm{vac}(0)}$ 
is an eigenstate of $H^{(0)}$ of unit
norm\footnote{
The normalization
$\braket{\mathbf{p}_1,\dotsc,\mathbf{p}_n}{\mathbf{p}_1,\dotsc\mathbf{p}_n}=1$
is more convenient here than the standard relativistic normalization.
} 
with energy
$\omega(\mathbf{p})$. This state can be interpreted as containing a
single particle of momentum $\mathbf{p}$.  Applying additional
creation operators with sharply defined momentum yields an eigenstate
of $H^{(0)}$ whose energy is simply the sum of the energies of the
particles created. In this way, one obtains the entire spectrum of
this non-interacting theory. A perfectly momentum-resolving particle
detector can be modeled by the observable $L^{-d} a_{\mathbf{p}}^\dag
a_{\mathbf{p}}$, whose eigenvalues are integers that count how many
particles are in momentum mode $\mathbf{p}$.

In accordance with the uncertainty principle, $L^{-d/2} a^\dag_{\mathbf{p}}
\ket{\mathrm{vac}(0)}$ represents a particle completely delocalized
in space.  The state $\phi(\mathbf{x}) \ket{\mathrm{vac}(0)}$ is
interpreted as a single particle localized at $\mathbf{x}$ (up to
normalization). Because $a_{\mathbf{p}} \ket{\mathrm{vac}(0)} = 0$,
$\phi(\mathbf{x}) \ket{\mathrm{vac}(0)} = a_{\mathbf{x}}^\dag
\ket{\mathrm{vac}(0)}$, where
\begin{equation}
\label{axdag}
a_{\mathbf{x}}^\dag = \sum_{\mathbf{p} \in \Gamma} \frac{1}{L^d} e^{-i
  \mathbf{p} \cdot \mathbf{x}} \sqrt{\frac{1}{2 \omega(\mathbf{p})}}
a_{\mathbf{p}}^\dag \,.
\end{equation}
In the nonrelativistic limit $m_0 \gg |\mathbf{p}|$, the factor of
$\frac{1}{\omega(\mathbf{p})}$ in \eq{axdag} becomes approximately
constant, and one recovers the familiar nonrelativistic Fourier
relation between position-space and momentum-space wavefunctions. 

We are able to solve for the spectrum of $H^{(0)}$ in terms of
non-interacting particles with sharply defined momentum. In an
interacting quantum field theory, the momentum of individual particles
is no longer conserved. Nevertheless, total momentum is conserved, and
thus the single-particle subspace can be analyzed similarly to the
non-interacting case. Specifically, by starting with the
single-particle momentum-$\mathbf{p}$ eigenstate of the non-interacting 
theory, $L^{-d/2} a_{\mathbf{p}}^\dag \ket{\mathrm{vac}(0)}$, and then 
adiabatically 
turning on $\lambda_0$, one obtains an eigenstate of $H$, which can be
interpreted as a single particle of momentum $\mathbf{p}$ of the
interacting theory.

\subsubsection{Continuum Limit}
\label{continuum}

In the preceding section we described a lattice field theory. The
quantum mechanics of such a system is mathematically well
defined. However, a lattice theory lacks translational, rotational,
and Lorentz invariance. Thus one is led to try to construct a
continuum limit of a sequence of lattice theories with successively
finer spacing, much as one obtains integrals by a sequence of more
finely discretized Riemann sums.

The naive attempt at constructing a continuum limit in which one simply 
takes $a \to 0$ in the definition of the lattice Hamiltonian fails to yield
convergent answers to physical questions. Instead, one should consider
a sequence of Hamiltonians of the form \eq{ham} on successively
finer lattices, where $m_0^2$ and $\lambda_0$ are functions of $a$.
One aims to show that, for a suitable choice of the $a$ dependence of
$m_0^2$ and $\lambda_0$, the entire theory converges to some 
meaningful limit as $a \to 0$. This is known as renormalization. 

Although the existence of a continuum limit of a sequence of
(Euclidean) $\phi^4_{2,3}$ lattice theories has been shown rigorously
in \cite{Sokal}, all that has been demonstrated for most physically
interesting quantum field theories is perturbative renormalizability,
namely, that physical quantities are finite to all orders in
perturbation theory.  Physical (renormalized) quantities are then
calculated as a perturbative series in the coupling.  In our analysis
of weak coupling, we use expressions obtained from perturbation
theory.  Specifically, in \sect{sec:renorm}, we calculate the physical
mass $m$ up to second order in $\lambda_0$. Inverting this
expression and the analogous one for $\lambda$ yields
prescriptions for  choosing the bare parameters $m_0$ and $\lambda_0$
as functions of $a$ to achieve given physical parameters $\lambda$
and $m$ (corresponding to what is measured by experiment).

For the purpose of simulation, in particular the analysis of
algorithmic complexity, one must ask not only whether the
sequence of Hamiltonians converges to a continuum limit, but also how
quickly it converges to that limit. 
Such questions are addressed in \sect{sec:a}.
Similarly, the rate of convergence to an infinite-volume limit is
studied in \sect{sec:vol}.

\subsubsection{Phase Transitions}
\label{QPT}

Consider the continuum $\phi^4$ theory in the infinite-volume limit
in $D=2$ or $3$ spacetime dimensions.
At certain values of the parameters $m_0^2$ and $\lambda_0$,
the particle mass $m$ vanishes. These points form a critical curve,
across which a quantum (zero-temperature) phase transition occurs,
corresponding to the spontaneous breaking of the $\phi \to -\phi$ symmetry 
of the theory.

The existence of such a phase transition was shown rigorously
in \cite{Glimm:1974tz,Guerra:1975ym,McBryan:1976ga}.
As the system approaches it, thermodynamic functions and correlation 
functions exhibit power-law behavior, as is characteristic of
a second-order phase transition. In particular, for constant $m_0^2$, 
\begin{equation}
\label{numass}
m \sim |\lambda_0 - \lambda_c|^\nu \,,
\end{equation}
where $\lambda_c$, the critical value of the coupling, depends on $m_0^2$. 
The critical exponent $\nu$ depends only on the dimension
and is bounded below: $\nu \geq 1/2$ \cite{Glimm:1974tz,McBryan:1976ga}. 

Empirically, it has been found that systems with second-order phase
transitions can be classified into universality classes.
Within each class, critical exponents are universal, taking the same
values for all systems, even though the systems may be vastly different
in their microscopic interactions. (That the critical behavior depends
only on the symmetries and the spatial dimensionality of the 
Hamiltonian is explained by the concept of the renormalization
group.) The $\phi^4$ theory is believed to be in the same
universality class as the Ising model, for which
\begin{equation}
\label{nu}
\nu = \left\{ \begin{array}{ll} 1 \,, & D=2 \,,
\\
0.63\ldots \,, & D=3 \,.
\end{array} \right.
\end{equation}
The value above for $D=3$ has also been obtained directly in the $\phi^4$
theory by Borel resummation \cite{LeGuillou:1977ju}.

In general, the non-perturbative regime, in which the coupling is 
sufficiently strong that perturbation theory fails (and hence also
called the strong-coupling regime), is expected to occur in the vicinity 
of the phase transition, with the inter-particle force maximum at the 
critical value of the coupling. Correspondingly, the dimensionless
ratio $\lambda/m^{4-D}$ ($D=2,\,3$) will become large: indeed, this
is consistent with the mass shrinking to zero and the coupling 
approaching a non-trivial infrared fixed point. 

In $D=4$ dimensions, in contrast, the believed triviality of the 
continuum $\phi^4$ theory implies that there is no non-trivial fixed point
of the renormalization group and hence no phase transition as one varies
($m_0^2$, $\lambda_0$). Moreover, triviality places bounds on the 
maximum value of the renormalized coupling \cite{Luscher:1987ay}.
In particular, strong coupling requires $p a$ to be  $O(1)$: 
in the continuum-like regime, renormalized perturbation theory
should be valid.

\subsection{Quantum Computing}
\label{QC}

The quantum circuit model is a convenient framework for describing
quantum computation.\footnote{An introduction to the basic notions of
  quantum circuits can be found in \cite{Nielsen_Chuang}.} Quantum
circuits are in many respects analogous to classical circuits. In
classical circuits, logic gates operate on bits, each in the state
$0$ or $1$. There exist universal sets of gates (on their own able to
implement any Boolean function): examples are  $\{\mathrm{AND},
\mathrm{OR}, \mathrm{NOT} \}$ and $\{\mathrm{NAND} \}$.  Similarly, in
quantum circuits, quantum gates operate unitarily on qubits,  each of
which is in a linear superposition of two basis states ($\ket{0}$ and
$\ket{1}$). Arbitrary single-qubit and controlled-NOT gates together
can perform any unitary operation. Furthermore, there exist
small\footnote{There exist two-qubit gates that are universal by
  themselves, but the standard universal gate set consists of three
  gates: controlled-NOT, Hadamard, and $\pi/8$ \cite{Nielsen_Chuang}.}
sets of gates that are universal, in the sense that they can
approximate any unitary operation to arbitrary accuracy.

In the following subsections we review the standard primitives from
the theory of quantum circuits that are used in our algorithm. The
number of gates needed to implement a unitary transformation is used
as a measure of its running time, although the actual running time might
be much shorter if the computation can be highly parallelized. 

\subsubsection{Quantum Fourier Transform}

For any $f:\{0,1,2,\ldots,N-1\} \to \mathbb{C}$, the function 
$\tilde{f}:\{0,1,\ldots,N-1\} \to \mathbb{C}$ given by
\begin{equation}
\tilde{f}(k) = \frac{1}{\sqrt{N}} \sum_{j=0}^{N-1} e^{-2 \pi i j k/N} f(j)
\end{equation}
is called the discrete Fourier transform of $f$. The linear
transformation $f \to \tilde{f}$ on the $N$-dimensional complex vector space
of all functions $f:\{0,1,2,\ldots,N-1\} \to \mathbb{C}$ is
unitary. It can be implemented on the amplitudes of an arbitrary state 
with $O(\log^2 N \log \log N \log \log \log N)$ 
quantum gates \cite{Shor_factoring}.

\subsubsection{Reversible Circuits and Phase Kickback}
\label{reversible}

Because the quantum time evolution of a closed system is unitary,
information cannot be erased in a quantum circuit. Suppose we are
given a classical circuit computing some function
$f:\{0,1\}^n \to \{0,1,2,\ldots,M-1\}$. (We think of the output as an
integer, although it is of course represented as a
string of $m = \lceil \log_2 M \rceil$ bits, called a ``register''.) 
Unless $f$ is injective,
there is no unitary operator $V_f$ such that $V_f \ket{x} =
\ket{f(x)}$. However, one can always define a unitary operator $U_f$
such that, for any $y \in \{0,1,2,\ldots,M-1\}$, 
\begin{equation}
U_f \ket{x} \ket{y} = \ket{x} \ket{(y + f(x)) \!\!\! \mod M} \,.
\end{equation}
Furthermore, the number of quantum gates needed to implement this unitary 
operator can never exceed the number of classical gates needed to compute $f$
by more than a constant factor \cite{FT82, Nielsen_Chuang}. If we
initialize the output qubits to the zero state, then applying $U_f$
yields $f(x)$ written into the values of the bits:
\begin{equation}
U_f \ket{x} \ket{0} = \ket{x} \ket{f(x)} \,.
\end{equation}
If we instead prepare the output qubits in the state
\begin{equation}
\ket{R_M} = \frac{1}{\sqrt{M}} \sum_{y = 0}^{M-1} e^{-i 2 \pi y/M} 
\ket{y} \,,
\end{equation}
then applying $U_f$ ``kicks back'' $f(x)$ into the
phase~\cite{Cleve_revisit}, that is,
\begin{equation}
U_f \ket{x} \ket{R_M} = e^{i 2 \pi f(x)/M} \ket{x} \ket{R_M} \,.
\end{equation}
The state $\ket{R_M}$ can be prepared efficiently by means of a
quantum Fourier transform.

\subsubsection{Phase Estimation}
\label{phase_estimation}

Consider a quantum circuit implementing a unitary transformation $U$ 
on $n$ qubits, and an eigenstate $\ket{\theta}$ such that $U \ket{\theta} =
e^{i \theta} \ket{\theta}$. One can use quantum Fourier transforms to measure
$\theta$ to $m$ bits of precision using $O(m^2 \log m \log \log m +
T)$ quantum gates, where $T$ is the number of gates needed to
implement $U^{2^m}$ \cite{Kitaev95}. 
Applying the phase estimation circuit 
to an arbitrary state yields an approximate measurement in the
eigenbasis of $U$. Eigenvalues closer together than $\sim 2^{-m}$ are not
distinguished. In many physical applications, one wishes to measure in
the eigenbasis of a Hermitian operator $Q$. The eigenbasis of $Q$ is
the same as the eigenbasis of the unitary operator $e^{iQt}$. Thus, in
cases where one can implement the unitary transformation $e^{iQt}$ by
an efficient quantum circuit, the problem of measuring $Q$ reduces to
the standard phase-estimation problem. More discussion of the
implementation of $e^{iQt}$ is given in the next section.

\subsubsection{Simulating Hamiltonian Time Evolution}
\label{hamiltonian_simulation}

Given a Hamiltonian, we wish to find a quantum circuit of
few gates implementing the corresponding unitary time evolution. The
details of how to do this depend on the specifics of the
Hamiltonian. Thus for concreteness we specifically consider 
the Hamiltonian $H$ defined in \eq{ham}. The simulation method
described here was introduced in \cite{Zalka, Wiesner}.

To simulate time evolution according to $H$, we decompose $H$ as
\begin{eqnarray}
H & = & H_{\pi} + H_{\phi} \,, \\
H_{\pi} & = & \frac{1}{2} \sum_{\mathbf{x} \in \Omega} a^d \pi(\mathbf{x})^2 
\,,\\
H_{\phi} & = & \sum_{\mathbf{x} \in \Omega} a^d \left[ \frac{1}{2} 
\left( \nabla_a \phi \right)^2 (\mathbf{x}) + \frac{m_0^2}{2} \phi(\mathbf{x})^2 
+ \frac{\lambda_0}{4 !} \phi(\mathbf{x})^4 \right] 
\,.
\end{eqnarray}
$e^{-iH_{\phi} \delta t}$ acts on the computational basis states simply by
inducing a phase. This phase (namely, $\sum_{\mathbf{x} \in \Omega} a^d
\left[ \frac{1}{2} \left( \nabla_a   \phi \right)^2 (\mathbf{x}) +
  \frac{m_0^2}{2} \phi(\mathbf{x})^2  + \lambda_0 \phi(\mathbf{x})^4
  \right]$) is not hard to compute; on a classical computer one could
compute it using $O(\mathcal{V})$ gates. Thus, using the method of
phase kickback, one can implement $e^{-iH_{\phi}\delta t}$ for any $\delta t$ 
using $O(\mathcal{V})$ quantum gates (\sect{reversible}).

Similarly, we can simulate $e^{-i
  H_{\pi} \delta t}$ for any $\delta t$ by first Fourier transforming
each of the $\mathcal{V}$ registers. (Because each register contains
only logarithmically many qubits (\sect{qubits}),
these Fourier transforms use logarithmically many quantum gates each,
and thus they make only a small contribution to the overall
complexity of the algorithm.) This brings the state into the
eigenbasis of the complete set of commuting observables $\{ \pi(\mathbf{x}) 
| \mathbf{x} \in \Omega \}$. In this basis, $e^{-i H_{\pi} \delta t}$ simply 
induces a phase 
$\sum_{\mathbf{x} \in \Omega} a^d \pi(\mathbf{x})^2 \delta t$. 
Computing this phase has complexity $O(\mathcal{V})$, and therefore
phase kickback implements $e^{-i H_{\pi} \delta t}$ using
$O(\mathcal{V})$ quantum gates. Afterwards, inverse Fourier
transforms bring the system back into the computational basis, which
is the eigenbasis of $\{ \phi(\mathbf{x}) | \mathbf{x} \in \Omega \}$.

Given the ability to simulate $e^{-i H_{\phi} \delta t}$ and $e^{-i
  H_{\pi} \delta t}$, one can simulate $e^{-i (H_{\phi} + H_{\pi}) t}$
using Trotter's formula \cite{Lloyd_science},
\begin{equation}
e^{-i (H_{\phi} + H_{\pi}) t} = \left( e^{-i H_{\phi} t/n} e^{-i H_{\pi}
  t/n} \right)^n + O(t^2 \| [H_{\pi}, H_{\phi} ] \| /n) \,.
\end{equation}
Thus, to achieve error $\epsilon_{\mathrm{ST}}$ one alternatingly performs time
evolutions of $H_{\phi}$ and $H_{\pi}$ with $n \sim t^2 \| [H_{\pi},
  H_{\phi}] \|/\epsilon_{\mathrm{ST}}$ steps. Each step requires $O(\mathcal{V})$
gates. $\| [H_{\pi}, H_{\phi}]\| = O(\mathcal{V})$ because of
locality. Thus, the total complexity is $O(t^2 \mathcal{V}^2)$.

Fortunately, one can systematically construct higher-order Suzuki-Trotter
formulae that yield better scaling in $t$.  
Using a $k\th$-order Suzuki-Trotter formula, one can simulate
time evolution for time $t$ on a lattice of $\mathcal{V}$ sites using
a quantum circuit of $O((t\mathcal{V})^{1+\frac{1}{2k}})$ gates
(\sect{Trotter}). Note
that our analysis of the $\mathcal{V}$ scaling of quantum circuits to
simulate spatially local Hamiltonians on large lattices is, to our
knowledge, novel. This may be of independent interest for the
simulation of lattice Hamiltonians arising, for example, in
condensed-matter physics.

\subsubsection{Adiabatic State Preparation}
\label{ASP}

Preparing an arbitrary state on $n$ qubits requires exponentially many
quantum gates \cite{Knill}.  Even the restricted problem of preparing 
ground states of local Hamiltonians is, for worst-case instances, 
intractable for quantum computers. (More precisely, this problem is
QMA-complete \cite{Kitaev_book, Kempe_Regev, Kempe}.) However, for certain 
Hamiltonians one can solve the problem of preparing the ground state 
using polynomially many gates by simulating adiabatic time evolution, 
that is, time evolution according to a Hamiltonian with slowly 
time-varying parameters. 

Let $H_{\mathrm{init}}$ be a Hamiltonian whose ground state is easy to
prepare. Suppose there exists some Hamiltonian $H(s)$ such that $H(0)
= H_{\mathrm{init}}$ and $H(1) = H_{\mathrm{final}}$, and let
$\ket{\psi(s)}$ be the ground state of $H(s)$. According to the
adiabatic theorem, by starting with $\ket{\psi(0)}$, and evolving
according to the time-dependent Hamiltonian $H(t/\tau)$ for
sufficiently large $\tau$, one obtains at time $\tau$ a state
approximately equal to $\ket{\psi(1)}$. 
Let $\gamma$ be the energy gap between the ground state and first excited
state.
Quantitative versions of the
adiabatic theorem have been proven that show it always suffices to
choose $\tau = O(1/\gamma^3)$ \cite{Ruskai, Goldstone}, and in certain
cases it suffices to choose $\tau = O(1/\gamma^2)$ \cite{Messiah}.

For efficient state preparation within the quantum circuit model, 
it suffices if $\gamma(s)$ is not too small and time evolution 
governed by $H(s)$ can be simulated efficiently. 

\section{Quantum Algorithm}
\label{algorithm}

The set of field operators $\{\phi(\mathbf{x})| \mathbf{x} \in \Omega\}$ 
forms a complete set of commuting observables. We represent the state 
of our lattice field theory by devoting one register of qubits to store the
value of the field at each lattice point. 
One can represent a quantum field at energy scale $E$ with fidelity
$1-\epsilon$ to the exact state, with each register consisting of 
$O\left( \log \left( \frac{E \mathcal{V}}{\epsilon} \right)
\right)$ qubits (\sect{qubits}).

Schematically, our quantum algorithm proceeds as follows.
\begin{enumerate}
\item Use the method of Kitaev and Webb \cite{Kitaev_Webb} to prepare
  the ground state of the free theory ($\lambda_0 = 0$).
\item Excite wavepackets of the free theory.
\item Evolve for a time $\tau$ during which the interaction is
  adiabatically turned on. This yields wavepackets of the interacting
  theory.
\item Evolve for a time $t$ in which scattering occurs.
\item Perform measurements. 
There are two possible methods, described below.\\
i) Evolve for a time $\tau$ during which the interaction is
  adiabatically turned off. This brings us back to the free theory,
  where interpreting field states in terms of particles is
  straightforward.
  Then, measure the number operators of the momentum modes of the free theory.
  This method is suitable if there are no outgoing bound states.\\
ii) Choose small regions corresponding to the positions of localized 
  detectors, and measure the total energy operator in each one via phase 
  estimation.

\end{enumerate}


\subsection{Description and Discussion}
\label{details}

We next discuss each step of the quantum algorithm in more detail. The
full analysis of their complexity scaling occupies \sect{analysis}.

Instead of $\{\phi(\mathbf{x})| \mathbf{x} \in \Omega\}$,
one could use $\{a^\dag_{\mathbf{p}} a_{\mathbf{p}} | \mathbf{p} 
\in \Gamma\}$ as the
complete set of commuting observables. In this case, the qubits store
the occupation numbers of momentum modes. However,  
simulations in the field representation seem to be more efficient than
simulations in the occupation-number representation. The primary
reason for this is that, upon expanding the $\phi^4$ operator in terms
of creation and annihilation operators, one obtains an expression that
is nonlocal in momentum space. This makes the simulation of $e^{-iHt}$
using Suzuki-Trotter formulae somewhat inefficient (although still
polynomial-time). Therefore, throughout this paper we consider only
simulations based on the field representation.

\begin{enumerate}

\item 
Improving upon the efficiency of earlier, more general, state-construction 
  methods \cite{Zalka, Grover_Rudolph}, Kitaev and Webb
  developed a quantum algorithm for constructing multivariate Gaussian
  superpositions \cite{Kitaev_Webb}.
  The main idea of this algorithm is to prepare a $\mathcal{V}$-dimensional
  multivariate Gaussian wavefunction with a diagonal covariance matrix, 
  and then reversibly change basis to obtain the desired covariance matrix.
  For large $\mathcal{V}$, the dominant cost in Kitaev and Webb's
  method 
  is the computation of the 
  $\mathbf{L}\mathbf{D}\mathbf{L}^T$ decomposition of the
  inverse covariance matrix, where $\mathbf{L}$ is a unit lower-triangular 
  matrix, and $\mathbf{D}$ is a diagonal matrix. This can be done in
  $\tilde{O}(\mathcal{V}^{2.376})$ time with established classical
  methods \cite{Bunch, Coppersmith}.\footnote{Recent work~\cite{Williams}
  on the complexity of matrix multiplication has lowered the bound
  on the exponent from 2.376 \cite{Coppersmith} to 2.373. In fact, a 
  long-standing conjecture is that this exponent is 2.} 
 (The notation $f(n) = \tilde{O}(g(n))$ means $f(n) = O(g(n) \log^c(n))$ 
  for some constant $c$.)
  The computation of the matrix elements of the
  covariance matrix itself is easy because, for large $V$, the sum
  \eq{freepropagator} is well approximated by an easily evaluated integral.

\item The expression $\phi(\mathbf{x}) \ket{\mathrm{vac}}$ for a
  single-particle state is not directly useful for quantum computing
  because $\phi(\mathbf{x})$ is not unitary. Instead, given a
  position-space wavefunction $\psi$ (with normalization
  $\sum_{\mathbf{x} \in \Omega} a^d |\psi(\mathbf{x})|^2 = 1$ so that it
  matches the dimensions of a continuum wavefunction) let
\begin{equation}
\ket{\psi} = a_{\psi}^\dag \ket{\mathrm{vac}(0)} \,,
\end{equation}
where
\begin{equation}
a_{\psi}^\dag = \eta(\psi) \sum_{\mathbf{x} \in \Omega} a^d \psi(\mathbf{x}) a_{\mathbf{x}}^\dag \,,
\end{equation}
$a_{\mathbf{x}}^\dag$ is defined in \eq{axdag}, and $\eta(\psi)$
is the normalization factor that ensures $a_{\psi} a_{\psi}^\dag
\ket{\mathrm{vac}(0)} = \ket{\mathrm{vac}(0)}$. Introduce one
ancillary qubit, and let 
\begin{equation}
\label{hpsi}
H_\psi = a_\psi^\dag \otimes \ket{1}\bra{0} + a_\psi \otimes \ket{0}
\bra{1}.
\end{equation}
One can easily verify that the span of $\ket{\mathrm{vac}(0)} \ket{0}$ and
$\ket{\psi} \ket{1}$ is an invariant subspace, on which $H_\psi$
acts as
\begin{eqnarray}
H_\psi \ket{\mathrm{vac}(0)} \ket{0} & = & \ket{\psi} \ket{1} \,, \\
H_\psi \ket{\psi} \ket{1} & = & \ket{\mathrm{vac}(0)} \ket{0}.
\end{eqnarray}
Thus
\begin{equation}
e^{-i H_\psi \pi/2} \ket{\mathrm{vac}(0)} \ket{0} = -i \ket{\psi} \ket{1}.
\end{equation}
Hence, by simulating a time evolution according to the Hamiltonian
$H_\psi$, we obtain the desired wavepacket state $\ket{\psi}$, up to an
irrelevant global phase and extra qubit, which can be discarded. After
rewriting $H_\psi$ in terms of the operators $\phi(\mathbf{x})$ and 
$\pi(\mathbf{x})$, one sees that simulating $H_\psi$ is a very similar task 
to simulating $H$ and can be done with the same techniques,
described in \sect{hamiltonian_simulation} and \sect{Trotter}. 

For a spatially localized wavepacket, that is, $\psi$ with bounded support,
$H_\psi$ is a quasilocal operator (\sect{locality}): 
it can be written in the form
$\sum_{\mathbf{x}} \left[ f_\psi(\mathbf{x}) \pi(\mathbf{x})  +
  g_\psi(\mathbf{x}) \phi(\mathbf{x}) \right]$, where $f$ and $g$ decay 
exponentially with characteristic decay length $1/m_0$ outside the support 
of $\psi$. One can thus choose some $c_1 \gg 1$ and set $f$ and $g$ to zero 
outside a distance $c_1/m_0$ from the support of $\psi$. The resulting
operator will be fully local and an exponentially good approximation
to $H_\psi$. The time evolution according to this local approximation
to $H_\psi$ can be simulated with complexity independent of
$V$. Furthermore, this quasilocality shows that, to create wavepackets
of additional particles, we can simply repeat this procedure with
different, well separated, choices of the position-space wavefunction
$\psi$. The only errors introduced at this step are due to the finite 
separation distance $\delta$ between wavepackets and are of order 
$\epsilon_{\mathrm{loc}} \sim e^{-\delta/m}$. 
(However, our wavepackets have a constant spread in momentum and thus
differ from the idealization of particles with precisely defined momenta.)
The wavepacket preparation thus has 
complexity scaling linearly with $n_{\mathrm{in}}$, the number of 
particles being prepared, and necessitates a dependence 
$V \sim n_{\mathrm{in}} \log(1/\epsilon_{\mathrm{loc}})$.

\item  To obtain a wavepacket of the interacting theory, we start with
  the wavepacket of the free theory, constructed in the previous step,
  and then simulate adiabatic turn-on of the interaction.\footnote{In
  this paper, we do not consider the simulation of processes involving
  incoming bound states.} 
  An adiabatic process of time $\tau$ can be simulated by a quantum
  circuit of $O\big((\tau \mathcal{V})^{1+\frac{1}{2k}}\big)$ gates 
  implementing a $k\th$-order Suzuki-Trotter formula (\sect{Trotter}).

For $0 \leq s \leq 1$, let 
\begin{equation}
\label{HS}
H(s) = \sum_{\mathbf{x} \in \Omega} a^d \left[ \frac{1}{2} \pi(\mathbf{x})^2 
+ \frac{1}{2} (\nabla_a \phi)^2 (\mathbf{x}) + \frac{1}{2} m_0^2(s)
  \phi(\mathbf{x})^2 + \frac{\lambda_0(s)}{4!} \phi(\mathbf{x})^4 \right]
\end{equation}
with $\lambda_0(0) = 0$. If we started with an eigenstate of the free
theory $H(0)$ and applied the time-dependent Hamiltonian $H(t/\tau)$ 
for time $\tau$ then, by the adiabatic theorem, we would obtain a
good approximation to the corresponding eigenstate of $H(1)$, provided
$\tau$ was sufficiently large. However, a wavepacket is a superposition
of eigenstates with different energies. These acquire different phases
during the adiabatic state preparation. Physically, this means the
wavepackets propagate and broaden. Propagation during adiabatic state
preparation is undesirable because the particles being prepared could
collide and scatter prematurely. Broadening of wavepackets decreases the
efficiency of the algorithm because very diffuse wavepackets largely
pass through each other without significant scattering. (For diffuse
wavepackets, the expected distance between particles is large even
when the wavepackets are on top of each other.) In this case, many
repetitions of the simulation are necessary before we observe
interesting scattering events. 

We can correct the dynamical phases by interspersing the adiabatic state 
preparation with backward time evolutions, thereby counteracting the
propagation and broadening of wavepackets. Specifically, we
follow an adiabatic path $H(s)$ from $s=0$ to $s=1$. To undo the
dynamical phase, we divide the total adiabatic evolution into $J$
steps. Before and after each step we apply time-independent
Hamiltonians as follows. For $j=0,1,\ldots,J-1$, let
\begin{eqnarray}
M_j & = & \exp \left[ i H \left( \frac{j+1}{J} \right) \frac{\tau}{2J}
  \right] U_j \exp \left[ i H \left( \frac{j}{J} \right)
  \frac{\tau}{2J} \right] \,, \label{Mj} \\
U_j & = & T \left\{ \exp \left[ -i \int_{j/J}^{(j+1)/J} H(s) \tau ds
  \right] \right\} \,, \label{Uj}
\end{eqnarray}
where $T\{ \cdot \}$ indicates the time-ordered product. ($U_j$ is the 
unitary time evolution induced by $H(t/\tau)$ from
$t=\frac{j\tau}{J}$ to $t=\frac{(j+1)\tau}{J}$.) The full
state-preparation process is given by the unitary operator
\begin{equation}
M = \prod_{j=0}^{J-1} M_j \,.
   \label{Mprod}
\end{equation}
We suppress the dynamical phases by choosing $J$ to be sufficiently 
large. The choice of a suitable ``path'' $\lambda_0(s),m_0^2(s)$ and
the complexity of this state-preparation process depend in a
complicated manner on the parameters in $H$ (\sect{preparing}). 

\item The time evolution $e^{-iH(1)t}$ can be implemented with
  $O\big((t\mathcal{V})^{1+\frac{1}{2k}}\big)$
  gates via a $k\th$-order Suzuki-Trotter formula (\sect{Trotter}).

\item In method i), the adiabatic turn-off of the coupling is simply the
  time-reversed version of the adiabatic turn-on. 
The complexity is no higher than that of the adiabatic state
preparation.\footnote{Readers of \sect{preparing} may notice that, in the
strongly coupled case, the matrix element for particle splitting,
which appears in the numerator of the diabatic error, gets multiplied
by $n_{\mathrm{out}}$. However, this is more than compensated by the
smaller momentum and hence larger energy gap against splitting.}

  By the method of phase estimation, 
  measurement of the number operator
  $L^{-d} a_{\mathbf{p}}^\dag a_{\mathbf{p}}$ reduces to the problem of
  simulating $e^{i L^{-d} a_{\mathbf{p}}^\dag a_{\mathbf{p}} t}$ for various 
  $t$. Thus,
  for a given $\mathbf{p}$, $L^{-d} a_{\mathbf{p}}^\dag a_{\mathbf{p}}$ 
  can be measured with $O \big( \mathcal{V}^{2+\frac{1}{2k}} \big)$ 
  quantum gates via a $k\th$-order Suzuki-Trotter formula
  (\sect{measurement}). Furthermore, if we instead simulate localized 
  detectors, the computational cost becomes independent of $V$ 
  (much as the computational cost of
  creating local wavepackets is independent of $V$ (\sect{locality})),
  but the momentum resolution becomes lower, as dictated by the 
  uncertainty principle. 

  Details of method ii) are given in \sect{detectors}.

\end{enumerate}

The allowable rate of adiabatic increase of the coupling constant
during state preparation is determined by the physical mass of the
theory. In the weakly coupled case, this can be calculated perturbatively, 
as is done in \sect{sec:mass}. In the strongly coupled case, it is not known
how to do the calculation. Thus one is left with the problem 
of determining how fast one can perform the adiabatic state
preparation without introducing errors. Fortunately, one can easily
calculate the mass using a quantum computer, as follows. First, one
adiabatically prepares the interacting vacuum state at some small
$\lambda_0$ and measures the energy of the vacuum using phase
estimation. 
The speed at which to increase $\lambda_0$ can be chosen
perturbatively for this small value of $\lambda_0$. Next, one
adiabatically prepares the state with a single zero-momentum particle
at the same value of $\lambda_0$ and measures its energy using phase
estimation. Taking the difference of these values yields the physical mass. 
This value of the physical mass provides guidance as to the speed of
adiabatic increase of the coupling to reach a slightly higher
$\lambda_0$. Repeating this process for successively higher
$\lambda_0$ allows one to reach strong coupling, while always having an
estimate of mass by which to choose a safe speed for adiabatic state
preparation. In addition, mapping out the physical mass as a function
of bare parameters (hence, for example, mapping out the phase diagram)
may be of independent interest.


\subsection{Efficiency}
\label{efficiency}

In complexity theory, the efficiency of an algorithm is judged by
how its computational demands scale with the problem size or some other
quantity associated with the problem's intrinsic difficulty.
An algorithm with polynomial-time asymptotic scaling is considered to
be feasible, whereas one with super-polynomial (typically, exponential)
scaling is considered infeasible. This classification has proved to
be a very useful guide in practice. The results stated below can
be roughly summarized as follows: the calculation of quantum
field-theoretical scattering amplitudes at high precision or strong
coupling is infeasible on classical computers but feasible on quantum
computers.

Traditional calculations of QFT scattering amplitudes rely upon 
perturbation theory, namely, a series expansion in powers of the coupling
(the coefficient of the interaction term), which is taken to be small. 
A powerful and intuitive way of organizing this perturbative expansion is 
through Feynman diagrams, in which the number of loops is associated with 
the power of the coupling. A reasonable measure of the computational
complexity of perturbative calculations is therefore the number of
Feynman diagrams involved. The number of diagrams is determined by
combinatorics and grows factorially with the number of loops and the
number of external particles. If the coupling constant is
insufficiently small, the perturbation series does not yield correct
results. There are then no feasible classical methods for calculating
scattering amplitudes, although lattice field theory can be used to
obtain static quantities, such as mass ratios. 
For computing the time evolution of the state vector on a classical 
computer non-perturbatively, no better method is known than discretizing the 
field and solving the Schr\"odinger equation numerically. This method is
infeasible because the dimension of the Hilbert space is exponentially large.

Even at weak coupling, the perturbation series is not
convergent, but only asymptotic: the error in the $N\th$-order sum
$\sum_{k=1}^N b_k g^k$ satisfies
\begin{equation}
\bigg| {\cal M} - \sum_{k=1}^N b_k g^k\bigg| = O(g^{N+1}) \,\,\,  
\mathrm{as} \,\, g \rightarrow 0 \,.
\end{equation}
Since the coefficients grow as \cite{Lipatov:1976ny,Brezin:1976vw} 
\begin{equation}
b_k \sim k ! c_1^k k^{c_2} \,,
\end{equation}
for some constants $c_1$ and $c_2$, 
there is a maximum possible precision, corresponding to truncation of
the series around the $(1/g)\th$ term. 

Consequently, our quantum algorithm should have an advantage in the
non-perturbative regime, or if very high precision is required. Thus,
we analyze the asymptotic scaling of our quantum algorithm when simulating
weakly coupled theories at arbitrarily high precision and 
strongly coupled theories arbitrarily close to the quantum phase
transition.  In the strongly coupled case, we also consider the scaling as
a function of the momenta of the incoming particles. As the energy of
the incoming particles becomes larger, the maximum number of kinematically
allowed outgoing particles correspondingly increases, thereby making the
problem potentially more computationally difficult. In the weakly
coupled case, processes producing large numbers of outgoing particles
are suppressed even at high energy, because they arise only at high
order in perturbation theory. 

In the weakly coupled case, we wish to determine the complexity of our
algorithm as a function of precision. We quantify this by demanding
that any standard physical quantity $\sigma$ extracted from the
simulation (for example, a scattering cross section) satisfy
\begin{equation}
\label{epsiloncrit}
(1-\epsilon) \sigma_{\mathrm{exact}} \leq \sigma \leq (1+\epsilon)
\sigma_{\mathrm{exact}} \,.
\end{equation}

To analyze the scaling of our algorithm with $\epsilon$, we first
consider errors due to spatial discretization. The effect of spatial
discretization is captured by (infinitely many) additional terms in
the effective Hamiltonian (\sect{sec:a}). Truncation of these terms
alters the calculated probability of scattering events. In particular,
the two dominant extra terms in the effective Hamiltonian are $\sum_i
\phi \partial_i^4 \phi \equiv \phi \partial_{\mathbf{x}}^4 \phi$ 
and $\phi^6$ terms, arising from
discretization of $(\nabla_a \phi)^2$ and quantum effects,
respectively. The coefficient of the 
$\phi \partial_{\mathbf{x}}^4 \phi$
term is $O(a^2)$, and the coefficient of the $\phi^6$ term is
$O(a^{5-d})$, so that the former dominates for $d=1,2$, whereas the
latter makes a comparable contribution for $d=3$. (To improve the
scaling, one can use better finite differences to approximate the
derivative or include the $\phi^6$ operator. However,
renormalization and mixing of the coefficients make this idea more
complicated than it is in standard numerical analysis.)

We now describe the errors induced by neglecting the effective
$\phi \partial_{\mathbf{x}}^4 \phi$ term. The analysis of $\phi^6$ errors is
analogous. Let $\sigma$ be a scattering cross section induced by the
Hamiltonian $H$ of \eq{ham}, and let $\sigma'$ be the corresponding
scattering cross section induced by
\begin{equation}
H' = H - \sum_{\mathbf{x} \in \Omega} a^d \frac{\tilde{c}}{2} 
\phi \partial_{\mathbf{x}}^4 \phi (\mathbf{x}) \,.
\end{equation}
By standard arguments from effective field theory 
(\sect{sec:eft}--\ref{sec:a}),
\begin{equation}
\sigma' = \sigma(1 + \tilde{c} f + O(\tilde{c}^2)),
\end{equation}
where $f$ is a function of the momenta and masses of the particles
involved in the scattering process. $\tilde{c} = O(a^2)$ and $f$ is
independent of $a$, so the errors induced by neglecting $\phi
\partial_\mathbf{x}^4 \phi$ are of order $a^2$. Similarly, the errors 
induced by neglecting the effective $\phi^6$ term are of order $a^2$ 
or smaller for $D=2,3,4$.

For the total error to satisfy \eq{epsiloncrit}, each individual
source of error must be at most $O(\epsilon)$. Thus, to ensure the
spatial discretization errors are sufficiently small, we set $a \sim
\sqrt{\epsilon}$. This choice of $a$ affects the complexity of the
preparation of the free vacuum and the complexity of the
Suzuki-Trotter time evolutions. Because $\mathcal{V} = \frac{V}{a^d}$,
the Kitaev-Webb state preparation uses
\begin{equation}
\label{Gprep}
G_{\mathrm{prep}} = \tilde{O}(\mathcal{V}^{2.376}) = \tilde{O}(a^{-2.376d}) =
\tilde{O}(\epsilon^{-1.188d}) \,
\end{equation}
quantum gates. Among the various Suzuki-Trotter time evolutions in
our algorithm, the most time-consuming is the adiabatic transition
from the free theory to the interacting theory. We thus substitute $a
\sim \sqrt{\epsilon}$ into \eq{Gstrict} and find that implementing
this process with a $k\th$-order Suzuki-Trotter formula
uses\footnote{Whether we use \eq{Gstrict} or \eq{Glenient} affects
  only the scaling with $V$.}
\begin{equation}
\label{gadiabatic1}
G_{\mathrm{adiabatic}} \sim 
\left\{ \begin{array}{ll}
\left( \epsilon^{-d/2}
\epsilon_{\mathrm{ad}}^{-1} \right)^{1+\frac{1}{2k}}, \quad d=1,2, \\ 
\left( \epsilon^{-4.5}
\epsilon_{\mathrm{ad}}^{-1} \right)^{1+\frac{1}{2k}}, \quad d=3 
\end{array} \right.
\end{equation}
quantum gates, where, by definition, the adiabatically produced state has an
inner product of at least $1-\epsilon_{\mathrm{ad}}$ with the 
exact state (\sect{preparing}).
It is efficient to use Suzuki-Trotter formulae with
large $k$, in which case the terms in the exponents such as
$\frac{1}{2k}$ become very small. To simplify our exposition, we
henceforth use the standard ``little-$o$'' notation, defined as
follows:
\begin{equation}
\label{littleo}
f(n) = o(g(n)) \textrm{ if and only if }  \lim_{n \to \infty} f(n)/g(n) = 0.
\end{equation}
In this language, \eq{gadiabatic1} becomes 
\begin{equation}
G_{\mathrm{adiabatic}} \sim 
\left\{ \begin{array}{ll}
\left( \epsilon^{-d/2}
\epsilon_{\mathrm{ad}}^{-1} \right)^{1+o(1)}, \quad d=1,2, \\
\left( \epsilon^{-4.5}
\epsilon_{\mathrm{ad}}^{-1} \right)^{1+o(1)}, \quad d=3. 
\end{array} \right.
\end{equation}

To obtain the full scaling with $\epsilon$ of the adiabatic state
preparation, we must next determine the relationship between $\epsilon$
and $\epsilon_{\mathrm{ad}}$. 
For small $\epsilon_{\mathrm{ad}}$, the adiabatically prepared
and exact states yield nearly identical probability distributions for all
possible measurements. (More precisely, the total variation distance
from the exact probability distribution is at most
$O(\epsilon_{\mathrm{ad}})$.) Thus, setting
$\epsilon_{\mathrm{ad}} = O(\epsilon)$ is certainly sufficient
to satisfy \eq{epsiloncrit}. We then have 
\begin{equation}
\label{Gadiabatic}
G_{\mathrm{adiabatic}} = 
\left\{ \begin{array}{ll}
\epsilon^{-1-d/2+o(1)}, \quad d=1,2, \\
\epsilon^{-5.5 +o(1)}, \quad d=3. 
\end{array} \right.
\end{equation}

Comparing \eq{Gprep} with \eq{Gadiabatic}, one sees that in
$d=1,3$ the adiabatic state preparation is the dominant cost, whereas in
$d=2$ the preparation of the free vacuum dominates. This leaves a
final asymptotic scaling of
\begin{equation} \label{Gtotal}
G_{\mathrm{total}} = O(G_{\mathrm{adiabatic}} + G_{\mathrm{prep}}) = 
\left\{ \begin{array}{ll}
\left( \frac{1}{\epsilon} \right)^{1.5+o(1)} \,, & d=1 \,,\\
\left( \frac{1}{\epsilon} \right)^{2.376+o(1)} \,, & d=2 \,,\\
\left( \frac{1}{\epsilon} \right)^{5.5+o(1)} \,, & d=3 \,.
\end{array} \right.
\end{equation}

So far, in determining the $\epsilon$ scaling of our algorithm, we
have considered only errors due to spatial discretization and errors
due to imperfect adiabaticity. We now
argue that the remaining sources of error make negligible
contributions to the overall $\epsilon$ scaling, which are already
captured by the $o(1)$ in \eq{Gtotal}.

In a theory with a non-zero mass, the error $\epsilon_{\mathrm{loc}}$
due to imperfect particle separation shrinks exponentially with the 
distance between the particles (\sect{sec:Veff}).
Therefore the total simulated volume $V$ should increase polylogarithmically
with $1/\epsilon_{\mathrm{loc}}$, and correspondingly the complexity
of the algorithm scales as $\mathrm{poly}(\log(1/\epsilon_{\mathrm{loc}}))$.
Similarly, by the analysis of
\sect{qubits}, the number of qubits per site scales only
logarithmically with $\epsilon_{\mathrm{trunc}}$, where
$1-\epsilon_{\mathrm{trunc}}$ is an inner product between the exact
quantum state and the achieved state. Thus, this source of error also
contributes only a polylogarithmic factor to the overall
$\epsilon$ scaling.

By \eq{hightrotter}, the errors resulting from use of a $k\th$-order
Suzuki-Trotter formula with $n$ timesteps are $\epsilon_{\mathrm{ST}}
\sim n^{-2k}$. $\epsilon_{\mathrm{ST}}$ is an operator norm of the
difference between the achieved and exact unitary transformations. It
thus induces a Euclidean distance of $O(\epsilon_{\mathrm{ST}})$
between quantum states. Hence, this error contributes only a factor of
$O(\epsilon^{-1/2k})$ to the overall $\epsilon$ scaling.

The number of quantum gates used to simulate the
strongly coupled theory has scaling in $1/(\lambda_c - \lambda_0)$ and
$p$ that is dominated by adiabatic state preparation (\sect{strong}).
We also estimate scaling with $n_{\mathrm{out}}$ as follows. For
two incoming particles with momenta $\mathbf{p}$ and $\mathbf{-p}$, the
maximum number of kinematically allowed outgoing particles is $n_{\mathrm{out}}
\sim p$. Furthermore, one needs $V \sim n_{\mathrm{out}}$ to obtain
good asymptotic out states separated by a distance of at least $\sim
1/m_0$. For continuum behavior,  $p = \eta/a$ for constant $\eta \ll
1$. Thus, $\mathcal{V} \sim n_{\mathrm{out}}^{d+1}$. Hence one needs 
$n_{\mathrm{out}}^{2.376 (d+1)}$ gates to prepare the free vacuum and, by
\eq{strongpscale}, $n_{\mathrm{out}}^{2d+3+o(1)}$ gates to
reach the interacting theory adiabatically. (The adiabatic turn-off
takes no longer than the adiabatic turn-on.) 
Thus the total scaling in $n_{\mathrm{out}}$ is dominated by
preparation of the free vacuum in three-dimensional spacetime,  but by
adiabatic turn-on in two-dimensional spacetime. The results of our resource
analysis are summarized in Table~\ref{strongtable}.

\begin{table}[hbt]
\begin{center}
\begin{tabular}{|c|c|c|c|}
\hline
    & $\lambda_c - \lambda_0$ & $p$ & $n_{\mathrm{out}}$ \\
\hline
$d=1$ & $\left( \frac{1}{\lambda_c - \lambda_0} \right)^{9+o(1)}$ & 
$p^{4+o(1)}$ & $\tilde{O}(n_{\mathrm{out}}^5)$ \\
\hline
$d=2$ & $\left( \frac{1}{\lambda_c-\lambda_0} \right)^{6.3+o(1)}$ &
$p^{6+o(1)}$ & $\tilde{O}(n_{\mathrm{out}}^{7.128})$ \\
\hline
\end{tabular}
\end{center}
\vspace{6pt}
\caption{\label{strongtable} The asymptotic scaling
  of the number of quantum gates needed to simulate scattering in the
  strong-coupling regime in one and two spatial dimensions is polynomial
  in $p$ (the momentum of the incoming pair of particles), 
  $\lambda_c - \lambda_0$ (the distance from the
  critical coupling), and $n_{\mathrm{out}}$ (the maximum kinematically
  allowed number of outgoing particles). Note that $V$ is kept fixed in
  the calculation of the scaling with $\mathrm{p}$. This is justified 
  when one is interested in scattering processes with a bounded number of
  outgoing particles.}
\end{table}


\section{Analysis of Algorithm}
\label{analysis}

Analysis of the algorithm requires quantifying various sources of error.
In broad terms, these fall into two categories: field-theoretical
cutoffs to render the problem finite and quantum computing primitives
upon which the algorithm is built. 

Cutoffs are imposed on both space and the field itself.
In \sect{qubits}, we analyze the effect of discretizing and imposing 
a cutoff on the field and thereby determine the number of qubits 
that is sufficient to represent the field. The effects of putting space 
on a lattice are considered later, in \sect{sec:eft}--\ref{sec:Veff}.

Sections \ref{preparing}--\ref{detectors} address quantum
computing primitives. In \sect{preparing}, we analyze the
adiabatic preparation of interacting wavepackets. This method induces
a phase shift, whose effect on a wavepacket is to cause undesirable
broadening and propagation.  The phase is shown to be proportional to
$\tau/J^2$, where $J$ is the number of adiabatic steps and $\tau$ is
the period of the turn-on.  The finite period for turn-on causes
errors through imperfect adiabaticity.  Such so-called `diabatic'
errors fall into two classes: particle creation from the vacuum and
the splitting of one particle into three. Using the adiabatic theorem,
we derive the probability of such events. Two criteria, keeping both
propagation and diabatic errors small, then determine satisfactory
choices of $J$ and $\tau$, and hence the gate complexity.  The
physical mass as a function of the coupling features prominently in
our calculations. For the weak-coupling regime, its form is obtained
by perturbation theory (\sect{sec:renorm}). For the strong-coupling
regime, we use its known behavior near the phase transition.

The time evolution during adiabatic turn-on and turn-off, and during 
the scattering, is implemented with a Suzuki-Trotter formula. In 
\sect{Trotter}, we show that a $k\th$-order Suzuki-Trotter
achieves linear scaling in the number of lattice sites, provided that
the Hamiltonian is local. This result is also used for the phase
estimation with which either occupation numbers or the energy and
momentum within regions are measured (see \sect{measurement}
and \sect{detectors}). 



\subsection{Representation by Qubits}
\label{qubits}

To represent the quantum state of the field, we assign a register of
qubits to store $\phi(\mathbf{x})$ at each lattice point 
$\mathbf{x} \in \Omega$. Each $\phi(\mathbf{x})$ is in principle an
unbounded continuous variable. Thus, to represent the field at a given
site with finitely many qubits, we cut off the field at a maximum
magnitude $\phi_{\max}$ and discretize it in increments of
$\delta_\phi$. This requires 
\begin{equation}
\label{placevalue}
n_b = \lceil \log_2 \left( 1+2 \phi_{\max}/\delta_{\phi} \right) \rceil
\end{equation}
qubits per site. In this section we show that one can simulate processes at
energy scale $E$, while maintaining $1-\epsilon_{\mathrm{trunc}}$
inner product with the exact state, with $n_b$ logarithmic in $1/a$,
$1/\epsilon_{\mathrm{trunc}}$, and $V$. Our analysis is
non-perturbative and thus applies equally to strongly and weakly
coupled $\phi^4$ theory.

Let $\ket{\psi}$ be the state, expressed in the field representation,
namely,
\begin{equation}
\ket{\psi} = \int_{-\infty}^\infty d \phi_1 \dotsi
\int_{-\infty}^\infty d \phi_{\mathcal{V}} \ 
\psi(\phi_1,\ldots,\phi_{\mathcal{V}}) \ket{\phi_1, \ldots, \phi_{\mathcal{V}}}
\, ,
\end{equation}
where $\{\phi_1, \dotsc, \phi_{\mathcal{V}}\} \equiv
\{\phi(\mathbf{x})| \mathbf{x} \in \Omega\}$,
and let
\begin{equation}
\ket{\psi_{\mathrm{cut}}} = \int_{-\phi_{\max}}^{\phi_{\max}} d \phi_1
\dotsi \int_{-\phi_{\max}}^{\phi_{\max}} d \phi_{\mathcal{V}} \ 
  \psi(\phi_1, \ldots, \phi_{\mathcal{V}}) \ket{\phi_1, \ldots
    \phi_{\mathcal{V}}} \,.
\end{equation}
Then
\begin{equation}
\braket{\psi}{\psi_{\mathrm{cut}}} =  \int_{-\phi_{\max}}^{\phi_{\max}} d \phi_1
\dotsi \int_{-\phi_{\max}}^{\phi_{\max}} d \phi_{\mathcal{V}}
\ \rho(\phi_1, \ldots, \phi_{\mathcal{V}}) \,,
\end{equation}
where $\rho$ is the probability distribution
\begin{equation}
\rho(\phi_1,\ldots, \phi_{\mathcal{V}}) = | \psi(\phi_1, \ldots, \phi_{\mathcal{V}})|^2 \,.
\end{equation}
In other words, $\braket{\psi}{\psi_{\mathrm{cut}}} = 1 -
p_{\mathrm{out}}$, where $p_{\mathrm{out}}$ is the probability that at
least one of $\phi_1,\ldots,\phi_{\mathcal{V}}$ is out of the range
$[-\phi_{\max},\phi_{\max}]$. By the union bound
($ \mathrm{Pr}(A \cup B) \leq \mathrm{Pr}(A) + \mathrm{Pr}(B)$),
\begin{equation}
\braket{\psi}{\psi_{\mathrm{cut}}} \geq 1 - \mathcal{V}
\max_{\mathbf{x} \in \Omega} p_{\mathrm{out}}(\mathbf{x}) \,,
\end{equation}
where $p_{\mathrm{out}}(\mathbf{x})$ is the probability that
$\phi(\mathbf{x})$ is out of the range $[-\phi_{\max},\phi_{\max}]$.

Let $\mu_{\phi(\mathbf{x})}$ and $\sigma_{\phi(\mathbf{x})}$ denote
the mean and standard deviation of $\phi(\mathbf{x})$ determined by
$\rho$. By Chebyshev's inequality, choosing $\phi_{\max} =
|\mu_{\phi(\mathbf{x})}| + c \sigma_{\phi(\mathbf{x})}$ (with $c>0$) 
ensures
\begin{equation}
p_{\mathrm{out}}(\mathbf{x}) \leq \frac{1}{c^2} \,.
\end{equation}
Thus, choosing
\begin{equation}
\label{phichoice}
\phi_{\max} =  \max_{\mathbf{x} \in \Omega} \left(
\left|\mu_{\phi(\mathbf{x})}\right| + 
\sqrt{\frac{\mathcal{V}}{\epsilon_{\mathrm{trunc}}}}
\sigma_{\phi(\mathbf{x})} \right)  
\end{equation}
ensures $\braket{\psi}{\psi_{\mathrm{cut}}} \geq 1-\epsilon_{\mathrm{trunc}}$.

Next, we observe the following.

\begin{proposition}
\label{canonfourier}
Let $\hat{p}$ and $\hat{q}$ be Hermitian operators on
$L^2(\mathbb{R})$ obeying the canonical commutation relation
$[\hat{q},\hat{p}]=i \id$. Then the eigenbasis of $\hat{p}$ is the
Fourier transform of the eigenbasis of $\hat{q}$.
\end{proposition}

\noindent
\begin{proof}
Let $\hat f(\delta) = e^{-i\hat q\delta}\hat p e^{i\hat q\delta}$. Then 
$\hat f(0) = \hat p$ and
\begin{equation}
\frac{d}{d\delta}\hat f(\delta) = e^{-i\hat q\delta}\left( -i [\hat q,\hat p] 
\right)e^{i\hat q\delta} = \id \,;
\end{equation}
therefore
\begin{equation}
\hat f(\delta) = \hat p + \delta \id \,,
\end{equation}
and thus
\begin{equation}
\hat p  e^{i\hat q\delta} = e^{i\hat q\delta}\left(\hat p + \delta \id\right)\,.
\end{equation}
If $|p\rangle$ denotes the $\hat p$ eigenstate with eigenvalue $p$, then
\begin{equation}
\hat p e^{i\hat q\delta}|p\rangle  = e^{i\hat q\delta}\left(\hat p + \delta
\id\right)|p\rangle = \left( p+\delta\right) e^{i\hat q\delta}|p\rangle \,,
\end{equation}
that is, $|p+\delta\rangle \equiv e^{i\hat q\delta}|p\rangle$ is the 
eigenstate of $\hat p$ with eigenvalue $p+\delta$. It follows that 
\begin{equation}
e^{i\hat q \delta} \int_{-\infty}^\infty dp |p\rangle e^{-iqp} 
= \int_{-\infty}^\infty dp |p+\delta\rangle e^{-iqp} 
= \int_{-\infty}^\infty dp |p\rangle e^{-iq(p-\delta)}
= e^{iq\delta}\int_{-\infty}^\infty dp |p\rangle e^{-iqp}.
\end{equation}
Expanding both sides to linear order in $\delta$, we conclude that 
$\int_{-\infty}^\infty dp |p\rangle e^{-iqp}$ is an eigenstate of $\hat q$ 
with eigenvalue $q$.
\end{proof}

By Proposition~\ref{canonfourier}, the eigenbasis of $a^d \pi(\mathbf{x})$ 
is the Fourier transform of the eigenbasis of $\phi(\mathbf{x})$. Thus,
discretizing $\phi(\mathbf{x})$ in increments of
$\delta_{\phi(\mathbf{x})}$ is roughly equivalent to the truncation $-
\pi_{\max} \leq \pi(\mathbf{x}) \leq \pi_{\max}$, where
\begin{equation}
\label{pimax}
\pi_{\max} = \frac{\pi}{a^d \delta_{\phi(\mathbf{x})}} \,.
\end{equation}
By the same argument used to choose $\phi_{\max}$, choosing
\begin{equation}
\label{pichoice}
\pi_{\max} =  \max_{\mathbf{x} \in \Omega} \left(
\left|\mu_{\pi(\mathbf{x})}\right| + 
\sqrt{\frac{\mathcal{V}}{\epsilon_{\mathrm{trunc}}}} 
\sigma_{\pi(\mathbf{x})} \right) 
\end{equation}
ensures fidelity $1-\epsilon_{\mathrm{trunc}}$ between $\ket{\psi}$
and its truncated and discretized version.

To obtain useful bounds on $\phi_{\max}$ and $\pi_{\max}$, we must bound
$\mu_{\phi(\mathbf{x})}$, $\sigma_{\phi(\mathbf{x})}$, 
$\mu_{\pi(\mathbf{x})}$, and $\sigma_{\pi(\mathbf{x})}$. 
To this end, we make the following straightforward observation.
\begin{proposition}
\label{moments}
Let $M$ be a Hermitian operator and let $\ket{\psi}$ be a quantum
state of unit norm. Then 
$|\bra{\psi} M \ket{\psi}| \leq \sqrt{ \bra{\psi} M^2 \ket{\psi} }$.
\end{proposition}

\noindent
\begin{proof}
For brevity, let $\langle Q \rangle = \bra{\psi} Q \ket{\psi}$ for any
observable $Q$. The operator $\left( M - \langle M \rangle \id
\right)^2$ is positive semidefinite. Thus,
\begin{eqnarray}
0 & \leq & \left\langle \left( M - \langle M \rangle \id \right)^2
\right\rangle \\
& = & \left\langle M^2 - 2 \langle M \rangle M + \langle M \rangle^2 \id
\right\rangle \\
& = & \langle M^2 \rangle - \langle M \rangle^2 \,.
\end{eqnarray}
\end{proof}

\noindent
Applied to the definitions
\begin{eqnarray}
\mu_{\phi(\mathbf{x})} & = & \bra{\psi} \phi(\mathbf{x}) \ket{\psi} \,,\\
\sigma_{\phi(\mathbf{x})} & = & \sqrt{ \bra{\psi} \phi(\mathbf{x})^2
    \ket{\psi} - \bra{\psi} \phi(\mathbf{x}) \ket{\psi}^2} \,,\\
\mu_{\pi(\mathbf{x})} & = & \bra{\psi} \pi(\mathbf{x}) \ket{\psi} \,, \\
\sigma_{\pi(\mathbf{x})} & = & \sqrt{ \bra{\psi} \pi(\mathbf{x})^2
  \ket{\psi} - \bra{\psi} \pi(\mathbf{x}) \ket{\psi}^2} \,,
\end{eqnarray}
Proposition~\ref{moments} implies that $\mu_{\phi(\mathbf{x})}$ and
$\sigma_{\phi(\mathbf{x})}$ are each at most 
$\sqrt{\bra{\psi} \phi(\mathbf{x})^2 \ket{\psi}}$, and
$\mu_{\pi(\mathbf{x})}$ and $\sigma_{\pi(\mathbf{x})}$ are each at
most $\sqrt{\bra{\psi} \pi(\mathbf{x})^2 \ket{\psi}}$. Thus, by
\eq{phichoice} and \eq{pichoice},
\begin{eqnarray}
\phi_{\max} & = & O \left( \max_{\mathbf{x} \in \Omega} \sqrt{
  \frac{\mathcal{V}}{\epsilon_{\mathrm{trunc}}} \bra{\psi} \phi(\mathbf{x})^2
  \ket{\psi}} \right) \,, \\
\pi_{\max} & = & O \left( \max_{\mathbf{x} \in \Omega} \sqrt{
  \frac{\mathcal{V}}{\epsilon_{\mathrm{trunc}}} \bra{\psi} \pi(\mathbf{x})^2
  \ket{\psi}} \right) \,,
\end{eqnarray}
so that, by \eq{placevalue} and \eq{pimax},
\begin{equation}
n_b = O \left( \log \left( \frac{V}{\epsilon_{\mathrm{trunc}}}
\max_{\mathbf{x},\mathbf{y} \in \Omega} \sqrt{ \bra{\psi} \pi(\mathbf{x})^2
\ket{\psi} \bra{\psi} \phi(\mathbf{y})^2 \ket{\psi}} \right) \right) \,.
\end{equation}

To establish logarithmic scaling of $n_b$, we need only prove
polynomial upper bounds on $\bra{\psi} \phi(\mathbf{x})^2 \ket{\psi}$
and $\bra{\psi} \pi(\mathbf{x})^2 \ket{\psi}$. Rather than making a physical
estimate of these expectation values, we prove simple
upper bounds that are probably quite loose. In the adiabatic state
preparation described in \sect{preparing}, the parameters
$m_0^2$ and $\lambda_0$ are varied. The following two propositions
cover all the combinations of parameters used in the adiabatic
preparation and subsequent scattering of both strongly and weakly
coupled wavepackets.

\begin{proposition}
\label{mpos}
Let $H$ be of the form shown in \eq{ham}. Suppose $m_0^2 > 0$ and
$\lambda_0 \geq 0$. Let $\ket{\psi}$ be any state of the field such
that $\bra{\psi} H \ket{\psi} \leq E$. Then $\forall \mathbf{x} \in
\Omega$,
\begin{eqnarray}
\label{phiboundpos}
\bra{\psi} \phi(\mathbf{x})^2 \ket{\psi} & \leq & \frac{2E}{a^d m_0^2} \,,\\
\label{piboundpos}
\bra{\psi} \pi(\mathbf{x})^2 \ket{\psi} & \leq & \frac{2E}{a^d} \,.
\end{eqnarray}
\end{proposition}

\noindent
\begin{proof}
\begin{eqnarray}
E & \geq & \bra{\psi} H \ket{\psi} \\
  & = & 
\label{almost}
  \bra{\psi} \sum_{\mathbf{y} \in \Omega} a^d \left[ \frac{1}{2}
  \pi(\mathbf{y})^2 + \frac{1}{2} (\nabla_a \phi)^2(\mathbf{y}) +
  \frac{m_0^2}{2} \phi(\mathbf{y})^2 + \frac{\lambda_0}{4!}
  \phi(\mathbf{y})^4 \right] \ket{\psi} \\
& \geq & \bra{\psi} a^d \frac{m_0^2}{2} \phi(\mathbf{x})^2 \ket{\psi},
\end{eqnarray}
where the last inequality follows because all of the operators we have
dropped are positive semidefinite. This establishes
\eq{phiboundpos}. Similarly, we can drop all but the
$\pi(\mathbf{x})$ term from the right-hand side of \eq{almost},
leaving
\begin{equation}
E \geq \bra{\psi} a^d \frac{1}{2} \pi(\mathbf{x})^2 \ket{\psi} \,,
\end{equation}
which establishes \eq{piboundpos}.
\end{proof}

\begin{proposition}
\label{mneg}
Let $H$ be of the form shown in \eq{ham}. Suppose $m_0^2 \leq 0$ and
$\lambda_0 > 0$. Let $\ket{\psi}$ be any state of the field such that
$\bra{\psi} H \ket{\psi} \leq E$. Then $\forall \mathbf{x} \in
\Omega$,
\begin{eqnarray}
\label{phiboundneg}
\bra{\psi} \phi(\mathbf{x})^2 \ket{\psi} & \leq &
-\frac{6m_0^2}{\lambda_0} + \sqrt{\frac{36
    m_0^4}{\lambda_0^2}+\frac{24}{\lambda_0 a^d} \left( E +
  \frac{3(V-a^d)m_0^4}{2\lambda_0} \right)} \,,\\
\label{piboundneg}
\bra{\psi} \pi(\mathbf{x})^2 \ket{\psi} & \leq & \frac{2}{a^d} \left(
E + \frac{3V m_0^4}{2 \lambda_0} \right) \,.
\end{eqnarray}
\end{proposition}

\noindent
\begin{proof}
The operator
\begin{equation}
U(\mathbf{x}) = \frac{m_0^2}{2} \phi(\mathbf{x})^2 +
\frac{\lambda_0}{4!} \phi(\mathbf{x})^4
\end{equation}
is sufficiently simple that we can directly calculate its minimal
eigenvalue  $U_{\min}$. If $m_0^2 \leq 0$ and $\lambda_0 > 0$ then
\begin{equation}
\label{Vmin}
U_{\min} = - \frac{3 m_0^4}{2 \lambda_0} \,.
\end{equation}
Thus, for \emph{any} state $\ket{\psi}$,
\begin{equation}
\label{anystate}
\bra{\psi} \sum_{\mathbf{y} \in \Omega} a^d U(\mathbf{y}) \ket{\psi}
\geq \frac{-3Vm_0^4}{2 \lambda_0} \,.
\end{equation}
Hence, recalling \eq{ham}, we obtain
\begin{eqnarray}
E & \geq & \bra{\psi} H \ket{\psi} \\
& = & 
\label{beginning}
\bra{\psi} \sum_{\mathbf{y} \in \Omega} a^d \left[ \frac{1}{2}
  \pi(\mathbf{y})^2 + \frac{1}{2} ( \nabla_a \phi)^2(\mathbf{y}) + 
  \frac{m_0^2}{2} \phi(\mathbf{y})^2 + \frac{\lambda_0}{4!} 
  \phi(\mathbf{y})^4 \right] \ket{\psi} \\
& \geq &
\label{secondtolast}
 \bra{\psi} \sum_{\mathbf{y} \in \Omega} a^d \left[
  \frac{1}{2} \pi(\mathbf{y})^2 + \frac{1}{2} (\nabla_a
  \phi)^2(\mathbf{y}) \right] \ket{\psi} - \frac{3V m_0^4}{2 \lambda_0}
\\
& \geq & 
\label{last}
\bra{\psi} \frac{a^d}{2} \pi(\mathbf{x})^2 \ket{\psi} -
\frac{3Vm_0^4}{2 \lambda_0} \,.
\end{eqnarray}
\eq{secondtolast} follows from \eq{anystate}. \eq{last} holds
(for any choice of $\mathbf{x}$) because all of the operators we have
dropped are positive semidefinite. This establishes
\eq{piboundneg}.

Similarly, dropping positive operators from \eq{beginning} and
using \eq{anystate} yield, for any $\mathbf{x}$,
\begin{equation}
a^d \bra{\psi} \left( \frac{m_0^2}{2} \phi(\mathbf{x})^2 +
\frac{\lambda_0}{4!} \phi(\mathbf{x})^4 \right) \ket{\psi} \leq
 E + \frac{3(V-a^d)m_0^4}{2\lambda_0}  \,.
\end{equation}
Applying Proposition \ref{moments} with $M = \phi(\mathbf{x})^2$ shows
that $\bra{\psi} \phi(\mathbf{x})^4 \ket{\psi} \geq \bra{\psi}
  \phi(\mathbf{x})^2 \ket{\psi}^2$. Thus,
\begin{equation}
a^d \left[ \frac{m_0^2}{2} \bra{\psi} \phi(\mathbf{x})^2 \ket{\psi} +
\frac{\lambda_0}{4!} \bra{\psi} \phi(\mathbf{x})^2 \ket{\psi}^2
\right] \leq  E + \frac{3(V-a^d)m_0^4}{2\lambda_0}  \,.
\end{equation}
Via the quadratic formula, this implies \eq{phiboundneg}.
\end{proof}


\subsection{Adiabatic Preparation of Interacting Wavepackets}
\label{preparing}

In this section, we analyze the adiabatic state-preparation procedure.
To analyze the error due to finite
$\tau$ and $J$, we consider the process of preparing a single-particle
wavepacket. By the analysis in \sect{sec:Veff}, the procedure
performs similarly in preparing wavepackets for multiple particles,
provided the particles are separated by more than the characteristic
length $1/m$ of the interaction.
In what follows we use $p = |\mathbf{p}|$, rather than the $D$-vector, 
as will be clear from the context.

The phase induced by $M_j$ on the momentum-$p$ eigenstate of $H(s)$
(with energy $E_p(s)$) is
\begin{equation}
\theta_j(p) = \left( E_p \left( \frac{j+1}{J} \right) + E_p \left(
\frac{j}{J} \right) \right) \frac{\tau}{2J} - \tau
\int_{j/J}^{(j+1)/J} ds E_p(s) \,.
\end{equation}
Taylor expanding $E_p$ about $s=(j+\frac{1}{2})/J$ yields
\begin{equation}
\label{thetaj}
\theta_j(p) = \frac{\tau}{12 J^3} \left. 
\frac{\partial^2 E_p}{\partial s^2} \right|_{s=(j+\frac{1}{2})/J}
+ O(J^{-5}) \,.
\end{equation}
Thus, the total phase induced is
\begin{eqnarray}
\theta(p) & = & \sum_{j=0}^{J-1} \theta_j(p) \\
& \simeq & \frac{\tau}{12 J^2} \int_0^1 ds \frac{\partial^2
  E_p}{\partial s^2} \\
& = & \frac{\tau}{12 J^2} \left. \frac{\partial E_p}{\partial s}
\right|_0^1 \,, \label{thetap1}
\end{eqnarray}
where the approximation holds for large $J$. For a Lorentz-invariant
theory, $E_p(s)$ must take the form
\begin{equation}
\label{LI}
E_p(s) = \sqrt{p^2+m^2(s)} \,.
\end{equation}
This should be a good approximation for the lattice theory provided
the particle momentum satisfies $p \ll 1/a$. Substituting \eq{LI}
into \eq{thetap1} yields
\begin{equation}
\label{thetap2}
\theta(p) \simeq \frac{\tau}{24 J^2} \left. \frac{\frac{\partial
    m^2(s)}{\partial{s}}}{\sqrt{p^2 + m^2(s)}} \right|_0^1 \,.
\end{equation}

Next, we consider the effect of this phase shift on a wavepacket
centered around momentum $\bar{p}$. If the wavepacket is narrowly
concentrated in momentum, then we can Taylor expand $\theta(p)$ to
second order about $\bar{p}$:
\begin{equation}
\theta(p) \simeq \theta(\bar{p}) + \mathcal{D} \cdot (p-\bar{p}) + \frac{1}{2}
\mathcal{B} \cdot (p - \bar{p})^2 \,,
\end{equation}
where
\begin{eqnarray}
\mathcal{D} & = & \left. \frac{\partial \theta}{\partial p}
\right|_{\bar{p}} \,, \label{D} \\
\mathcal{B} & = & \left. \frac{\partial^2 \theta}{\partial p^2}
\right|_{\bar{p}} \,. \label{B}
\end{eqnarray}
The phase shift $e^{i \mathcal{D} \cdot (p - \bar{p})}$ induces a 
translation (in position space) of any wavepacket by a distance $\mathcal{D}$. 
Similarly, the phase shift 
$e^{i \frac{1}{2} \mathcal{B} \cdot (p -  \bar{p})^2}$  governs broadening
of the wavepacket. 
From \eq{D} and \eq{thetap2}, we have
\begin{equation}
\label{phasecrit}
\mathcal{D} \simeq -\frac{\tau \bar{p} }{24 J^2}
\left. \frac{\frac{\partial m^2(s)} {\partial  s}}{\left( \bar{p}^2 +
  m^2(s) \right)^{3/2}} \right|_{s=0}^{s=1}  \,.
\end{equation}

We next determine the complexity by demanding that the propagation length
$\mathcal{D}$ be restricted to some small constant, and that the
probability of diabatic particle creation be small. Together, these
criteria determine $J$ and $\tau$. We can obtain a tighter bound in
the perturbative case than in the general case, so we treat these
separately.

\begin{figure}
\begin{center}
\includegraphics[width=0.3\textwidth]{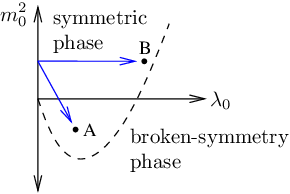}
\vspace{6pt}
\caption{\label{paths} The dashed line illustrates schematically the
  location of a quantum phase transition of $\phi^4$ theory in two and
  three spacetime dimensions. A and B denote weakly and strongly coupled
  continuum-like theories, respectively. We prepare them adiabatically
  by following the arrows starting from the massive free theory
  ($m_0^2 >0$, $\lambda_0 = 0$). To maintain adiabaticity the path
  must not cross the quantum phase transition.}
\end{center}
\end{figure}

\subsubsection{Weak Coupling}
\label{weak}

We wish to prepare the weakly coupled continuum-like theory by
adiabatically following a path that does not cross the quantum phase
transition, as illustrated in Fig. \ref{paths}. Note that, in the weakly
coupled continuum limit $a \to 0$, $m_0^2$ is negative. The path
illustrated in Fig. \ref{paths} can be described by the following
parametrization of $H(s)$ from \eq{HS}:
\begin{eqnarray}
m_0^2(s) & = & m^2 + s \delta_m \,, \label{weakm} \\
\lambda_0(s) & = & s \lambda_0 \,. \label{weaklambda}
\end{eqnarray}
Let $m(s)$ denote the physical mass of particles defined by
$H(s)$. By \eq{weakm} and \eq{weaklambda}, $m_0(0) = m(0) = m$. We
choose $\delta_m$ so that $m(1)$ is also equal to $m$. Thus, $H(s)$
linearly interpolates between a non-interacting theory with physical
mass $m$ and an interacting theory with physical mass $m$. 


We wish to find the asymptotic scaling of the run time of the
adiabatic state preparation. The question is, with which parameter
should we consider scaling? There are at least three regimes in which
classical methods for computing scattering amplitudes break down or
are inefficient: strong coupling, large numbers of external particles,
and high precision. In this section we are considering only weak
coupling (that is, $\lambda/m^{4-D} \ll 1$), leaving discussion of
strong coupling until the next section.
For an asymptotically large number of external particles, the efficiency
of our algorithm depends upon strong coupling, for the following reason.
A connected Feynman diagram involving $n$ external particles must have 
at least $v=O(n)$ vertices, so the amplitude for such a process
is suppressed by a factor of $\left( \frac{\lambda}{E^{4-D}}
\right)^v$, where $E$ is the energy scale of the process. Since 
$E \geq m$, many-particle scattering events are exponentially rare
at weak coupling and thus cannot be efficiently observed in
experiments or simulations. This leaves the high-precision
frontier. Recall that the perturbation series used in quantum field 
theory are asymptotic but not convergent. Thus, perturbative methods 
cannot be extended to arbitrarily high precision. 

Hence, in this section we consider the quantum gate complexity of
achieving arbitrarily high precision.  To do so, one chooses $a$ small
to obtain small discretization errors, $V$ large to obtain better
particle separation, $\tau$ long to improve adiabaticity, and $J$
large enough to limit particle propagation as the interaction is
turned on. Thus, we wish to know the scaling of the probability of
particle creation from the vacuum (denoted $P_{\mathrm{create}}$)
with $a$, $V$, $\tau$, and $J$. In this context, we consider $m$,
$\lambda$, and the incoming momentum to be constants.

We now analyze $m^2(s)$ (to second order in $\lambda_0$),
which determines the energy gap along the adiabatic path and the
propagation of wavepackets. 
By \sect{sec:mass},
\begin{equation}
\label{physmass}
m^2(s) = m_0^2(s) + s \lambda_0 \mu^{(1)} + s^2 \lambda_0^2 \mu^{(2)}
+ O(\lambda_0^3) \,,
\end{equation}
where
\begin{equation}
\label{mu1}
\mu^{(1)} = \left\{ \begin{array}{ll}
\frac{1}{8 \pi} \log \left( \frac{64}{m^2 a^2} \right) \,, & d = 1, \\
\frac{r_0^{(2)}}{16 \pi^2} \frac{1}{a} \,, & d = 2, \\
\frac{r_0^{(3)}}{32 \pi^3} \frac{1}{a^2} \,, & d = 3,
\end{array} \right.
\end{equation}
and
\begin{equation}
\label{mu2}
\mu^{(2)} = \left\{ \begin{array}{ll}
- \frac{1}{384 m^2} \,, & d = 1, \\
\frac{1}{96\pi^2} \log (ma) \,, & d = 2, \\
- \frac{r_1^{(3)}}{1536 \pi^7} \frac{1}{a^2} \,, & d = 3. 
\end{array} \right.
\end{equation} 
with $r_0^{(2)}$, $r_0^{(3)}$, and $r_1^{(3)}$ given in
\sect{sec:mass}. Thus,
\begin{equation}
\label{deltam}
\delta_m = - \lambda_0 \mu^{(1)} - \lambda_0^2 \mu^{(2)} + O(\lambda_0^3) \,.
\end{equation} 
Substituting \eq{deltam} and \eq{weakm} into \eq{physmass} yields
\begin{equation}
\label{pathmass}
m^2(s) = m^2 + s(s-1) \lambda_0^2 \mu^{(2)} + O(\lambda_0^3) \,.
\end{equation}
For the purpose of analyzing adiabaticity, we note that
\begin{equation}
\min_{0 \leq s \leq 1} m^2(s) = m^2 + O(\lambda_0^3) \,.
\end{equation}
This feature is helpful, because a small mass gap would necessitate slow 
adiabatic preparation. 
To analyze wavepacket propagation, we substitute \eq{pathmass} into
\eq{phasecrit} and obtain
\begin{eqnarray}
\mathcal{D} & = & \left. -\frac{\tau \bar{p}}{24 J^2} \frac{(2s-1)
  \lambda_0^2 \mu^{(2)}}{(\bar{p}^2 + m^2)^{3/2}} \right|_{s=0}^{s=1}
\nonumber \\
& = & -\frac{\tau \bar{p}}{12 J^2} \frac{\lambda_0^2
  \mu^{(2)}}{(\bar{p}^2+m^2)^{3/2}} \,. \label{Dweak}
\end{eqnarray}

We can solve \eq{Dweak} to determine the necessary scaling of $J$. We
are primarily interested in the scaling of $J$ with $a$, because we
wish to investigate the high-precision limit. Making the simplifying
assumption that the particles are highly relativistic ($\bar{p}^2 \gg
m^2$), we find 
\begin{equation}
J \simeq \sqrt{-\frac{\tau \lambda_0^2 \mu^{(2)}}{12 \mathcal{D}
    \bar{p}^2}} \,. \label{Jweak}
\end{equation}
Substituting \eq{mu2} into \eq{Jweak} yields
\begin{equation}
\label{Jscaleweak}
J = \left\{ \begin{array}{ll}
\tilde{O} \left( \sqrt{\frac{\tau \lambda_0^2}{m^2 \mathcal{D}
    \bar{p}^2}} \right) \,, & d = 1, \\
\tilde{O} \left( \sqrt{ \frac{\tau \lambda_0^2}{\mathcal{D}
    \bar{p}^2}} \right) \,, & d = 2, \\
\tilde{O} \left( \sqrt{\frac{\tau \lambda_0^2}{a^2 \mathcal{D}
    \bar{p}^2}} \right) \,, & d = 3.
\end{array} \right.
\end{equation}

To determine $\tau$, we next consider adiabaticity. Let $H(s)$ be any 
Hamiltonian differentiable with respect to $s$. Let $\ket{\phi_l(s)}$ 
be an eigenstate (with $H(s) \ket{\phi_l(s)} = E_l(s) \ket{\phi_l(s)}$)
separated by a non-zero energy gap for all $s$. Let $\ket{\psi_l(t)}$
be the state obtained by Schr\"odinger time evolution according to
$H(t/\tau)$ with initial condition $\ket{\psi_l(0)} =
\ket{\phi_l(0)}$. The diabatic transition amplitude to any other eigenstate 
$\ket{\phi_k(s)}$ such that $H(s) \ket{\phi_k(s)} = E_k(s)
\ket{\phi_k(s)}$ ($k \neq l$) is \cite{Messiah}
\begin{equation}
\label{traditional_integral}
\braket{\phi_k(s)}{\psi_l(\tau s)} \sim \int_{0}^{s} d \sigma
\frac{\bra{\phi_k(\sigma)} \frac{dH}{ds} \ket{\phi_l(\sigma)}}{E_k(\sigma)-E_l(\sigma)}
e^{i \tau (\varphi_k(\sigma)-\varphi_l(\sigma))} \left( 1 + O(1/\tau)
\right) \,.
\end{equation}
(The integrand is made well-defined by the phase convention
$\bra{\phi_k} \frac{d \ket{\phi_k}}{ds} = 0$.) Here,
\begin{equation}
\varphi_l(\sigma) = \int_{0}^{\sigma} d \sigma' E_l(\sigma') \,.
\end{equation}
In the case that $E_l$, $E_k$, and $\bra{\phi_k} \frac{dH}{ds}
\ket{\phi_l}$ are $s$-independent, this integral gives
\begin{equation}
\label{traditional}
\braket{\phi_k(s)}{\psi_l(\tau s)} \sim \left( 1 - e^{i \tau
  (E_k-E_l)s} \right) \frac{\bra{\phi_k}
  \frac{dH}{ds} \ket{\phi_l}}{-i \tau (E_k-E_l)^2} (1 + O(1/\tau)) \,.
\end{equation}
In the case that these quantities are approximately $s$-independent,
\eq{traditional} should hold as an approximation. 
The quadratic dependence on $E_k-E_l$ is an adiabatic approximation
traditionally used in physics. Motivated by applications in quantum
computation, mathematicians have developed bounds that hold rigorously 
even with strong $s$ dependence \cite{Ruskai}. These more general results
have a less favorable (cubic) dependence on $E_k-E_l$.
However, the traditional adiabatic approximation appears to be applicable 
to our analysis.

In reality, we wish to prepare a wavepacket state, not an
eigenstate. However, the wavepacket is well separated from other
particles and narrowly concentrated in momentum space. Thus, we shall
approximate it as an eigenstate $\ket{\phi_l(s)}$. Furthermore, by our
choice of path, the energy gap is kept constant to first order in the
coupling, and thus \eq{traditional} should be a good approximation
to \eq{traditional_integral}.

Summing the transition amplitudes to some state $\ket{\phi_k}$ from
the $J$ steps in our preparation process and applying the
triangle inequality\footnote{The $O(J)$ scaling obtained by the
  triangle inequality can be confirmed by a more detailed calculation
  taking into account the relative phases of the contributions to the
  total transition amplitude.} yield the following:
\begin{equation} 
\label{totaldiabatic}
\left| \braket{\phi_k}{\psi_l(\tau)} \right| = O \left(
\frac{1}{\tau} \sum_{j=0}^{J-1} \left| \frac{\bra{\phi_k(j/J)} \frac{dH}{ds}
  \ket{\phi_l(j/J)}}{(E_k(j/J) - E_l(j/J))^2} \right| \right) \,.
\end{equation}

The $j=0$ term in this sum can be evaluated exactly, because it arises
from the free theory. At $j \neq 0$ the theory is no longer exactly
solvable. However, one obtains the lowest-order contribution to the matrix
element $\bra{\mathbf{p}_1,\mathbf{p}_2,\mathbf{p}_3,\mathbf{p}_4;s=1} 
\phi^4 \ket{\mathrm{vac}(1)}$ in renormalized perturbation theory
simply by taking the $j=0$ expression and replacing $m_0$
with the physical mass and $\lambda_0$ with the physical coupling. Our
adiabatic path, \eq{weakm}, is designed so that the physical mass
at $s=1$ matches the bare mass at $j=0$.  Furthermore, the physical 
coupling differs from the bare coupling only by a logarithmically 
divergent (in $a$) correction for $d=3$ and non-divergent corrections 
for $d=1,2$.\footnote{
By calculations analogous to those in \sect{sec:renorm}, one obtains
\begin{eqnarray*} \label{lambda}
\lambda & = & 
\left\{
\begin{array}{ll}
\lambda_0 - \frac{3\lambda_0^2}{8\pi (m^{(1)})^2}  
+ \cdots \,, & \text{for $d=1$},\\[5pt]
\lambda_0 - \frac{3\lambda_0^2}{16\pi m^{(1)}} 
+ \cdots \,, & \text{for $d=2$},\\[5pt]
\lambda_0 + \frac{3\lambda_0^2}{16\pi^2} \log(m^{(1)} a)
+ \cdots \,, & \text{for $d=3$}.
\end{array}
\right.
\end{eqnarray*}
}
Thus, we can make the following approximation:
\begin{equation}
\label{summedadiabatic}
\left| \braket{\phi_k}{\psi_l(\tau)} \right| = \tilde{O} \left(
\frac{J}{\tau} \left| \frac{\bra{\phi_k(0)} \frac{dH}{ds}
  \ket{\phi_l(0)}}{(E_k(0) - E_l(0))^2} \right| \right) \,.
\end{equation}

Diabatic errors come in two types, creation of particles from the
vacuum and splitting of the incoming particles. The matrix element in the
numerator of \eq{summedadiabatic} can correspondingly be decomposed
as the sum of two contributions. We first consider particle creation
from the vacuum, taking $\ket{\phi_l(s)}$ to be
$\ket{\mathrm{vac}(s)}$. 

By \eq{weakm} and \eq{weaklambda},
\begin{equation}
\frac{dH}{ds} = \sum_{\mathbf{x} \in \Omega} a^d \left[
  \frac{\lambda_0}{4!} \phi^4(\mathbf{x}) + \frac{1}{2} \delta_m
  \phi^2(\mathbf{x}) \right] \,.
\end{equation}
Substituting this into the numerator of \eq{summedadiabatic},
setting $\ket{\phi_l(0)} = \ket{\mathrm{vac}(0)}$, and expanding
$\phi$ in terms of creation and annihilation operators show that the
only potentially non-zero transition amplitudes are to states
$\ket{\phi_k(0)}$ of two or four particles. The transition amplitude
to states of two particles has contributions from the $\phi^4$ term
and the $\phi^2$ term in $\frac{dH}{ds}$. The contribution from the
$\phi^2$ term is\footnote{To simplify the presentation, we assume
in \eq{alpha2}--\eq{alphatot}, \eq{treelevel} and \eq{asplit} that 
all momenta are distinct (for example, $\mathbf{p}_1\neq\mathbf{p}_2$). 
The cases of degenerate momenta differ only by numerical factors.}
\begin{eqnarray}
\label{alpha2}
\bra{\mathbf{p}_1, \mathbf{p}_2} \sum_{\mathbf{x} \in
  \Omega} a^d \frac{\delta_m}{2} \phi^2(\mathbf{x})
\ket{\mathrm{vac}(0)}
& = & \sum_{\mathbf{x} \in \Omega} a^d \frac{\delta_m}{L^{2d}} e^{-i
  (\mathbf{p}_1 + \mathbf{p}_2) \cdot \mathbf{x}} \frac{1}{2
  \sqrt{\omega(\mathbf{p}_1) \omega(\mathbf{p}_2)}} \bra{
  \mathbf{p}_1, \mathbf{p}_2} a_{\mathbf{p}_1}^\dag
a_{\mathbf{p}_2}^\dag \ket{\mathrm{vac}(0)} \nonumber \\
& = & \sum_{\mathbf{x} \in \Omega} \frac{a^d}{L^d} e^{-i (\mathbf{p}_1
  + \mathbf{p}_2) \cdot \mathbf{x}} \frac{\delta_m}{2
  \sqrt{\omega(\mathbf{p}_1) \omega(\mathbf{p}_2)}} \nonumber \\
& = & \frac{\delta_m}{2 \omega(\mathbf{p}_1)} \delta_{\mathbf{p}_1 +
  \mathbf{p}_2, 0} \nonumber \\
& = & (- \lambda_0 \mu^{(1)} - \lambda_0^2 \mu^{(2)}) \frac{\delta_{\mathbf{p}_1 +
  \mathbf{p}_2, 0}}{2 \omega(\mathbf{p}_1)} \,.
\end{eqnarray}
The contribution from the $\phi^4$ term is
\begin{equation}
\bra{\mathbf{p}_1, \mathbf{p}_2} \sum_{\mathbf{x}
  \in \Omega} a^d \frac{\lambda_0}{4!} \phi^4(\mathbf{x})
\ket{\mathrm{vac}(0)}
 = \lambda_0 \mu^{(1)} \frac{\delta_{\mathbf{p}_1 +
    \mathbf{p}_2,0}}{2 \omega(\mathbf{p}_1)} \,,
\end{equation}
by the definition of $\mu^{(1)}$ (that is, $\mu^{(1)}$ is the
first-order mass correction given by the first diagram in \eq{eq:diags}).

Thus, the total matrix element for creating two particles is
\begin{equation}
\label{alphatot}
\bra{\mathbf{p}_1, \mathbf{p}_2} \frac{dH}{ds} \ket{\mathrm{vac}(0)} =
- \lambda_0^2 \mu^{(2)} \frac{\delta_{\mathbf{p}_1 +
    \mathbf{p}_2,0}}{2 \omega(\mathbf{p}_1)} + O(\lambda_0^3) \,.
\end{equation}
Note that this requires tuning of $\delta_m$.
The total probability $P_{\mathrm{create}}^{(2)}$ of creating two
particles is obtained by summing the squared amplitudes for all
possible two-particle outgoing states. Thus, by \eq{alphatot} and
\eq{summedadiabatic},
\begin{eqnarray}
P_{\mathrm{create}}^{(2)} & \sim & \frac{J^2 \lambda_0^4}{\tau^2} 
( \mu^{(2)} )^2 \sum_{\mathbf{p} \in \Gamma}
\frac{1}{\omega(\mathbf{p})^6} \\
& \sim & \frac{J^2 \lambda_0^4 V}{\tau^2} ( \mu^{(2)} )^2
\int_\Gamma d^d p \frac{1}{(\mathbf{p}^2 + m^2)^3} \,,
\end{eqnarray}
where the notation $\int_\Gamma$ denotes 
$\int_{-\pi/a}^{\pi/a} \dotsi \int_{-\pi/a}^{\pi/a}$.
(Four powers of $1/\omega(\mathbf{p})$ come from the square of
$1/(E_k - E_l)^2$ and two powers come from the square of 
$\bra{\mathbf{p}_1, \mathbf{p}_2} \frac{dH}{ds} \ket{\mathrm{vac}(0)}$.)
For $d=1,2,3$, this integral is nondivergent as $a \to 0$. Thus,
\begin{equation}
\label{pcreate2}
P_{\mathrm{create}}^{(2)} = \tilde{O} \left( \frac{J^2 \lambda_0^4
  ( \mu^{(2)} )^2 V}{\tau^2 m^{6-d}} \right) \,.
\end{equation}
By \eq{Jscaleweak},
\begin{equation}
\label{jo}
J = \left\{ \begin{array}{ll}
\tilde{O} \left( \sqrt{\tau} \right) \,, & d=1,2 \,, \\
\tilde{O} \left( \frac{\sqrt{\tau}}{a} \right) \,, & d=3 \,.
\end{array} \right.
\end{equation}
By \eq{mu2},
\begin{equation}
\label{mu2o}
\mu^{(2)} = \left\{ \begin{array}{ll}
\tilde{O}(1) \,, & d=1,2 \,, \\
\tilde{O}(1/a^2) \,, & d=3 \,.
\end{array} \right.
\end{equation}
Thus, by \eq{pcreate2}, \eq{jo}, and \eq{mu2o},
\begin{equation}
P_{\mathrm{create}}^{(2)} = \left\{ \begin{array}{ll}
\tilde{O}\left( \frac{V}{\tau} \right) \,, & d=1,2 \,, \\
\tilde{O} \left( \frac{V}{\tau a^6} \right) \,, & d=3 \,.
\end{array} \right.
\end{equation}

Next, we consider the amplitude to create four particles.
At $s=0$, the corresponding term in the numerator of \eq{summedadiabatic} is
\begin{equation}
\label{treelevel}
\bra{\mathbf{p}_1,\mathbf{p}_2,\mathbf{p}_3,\mathbf{p}_4} 
\frac{\lambda_0}{4!}
\sum_{\mathbf{x} \in \Omega} a^d \phi^4(\mathbf{x})
\ket{\mathrm{vac}(0)} = \frac{ \lambda_0
  \delta_{\mathbf{p}_1+\mathbf{p}_2+\mathbf{p}_3+\mathbf{p}_4,0}}
{4 V \sqrt{\omega(\mathbf{p}_1) \omega(\mathbf{p}_2)
    \omega(\mathbf{p}_3) \omega(\mathbf{p}_4)}}
\,.
\end{equation}
We obtain the probability of excitation due to creation of four
particles from the vacuum by substituting the matrix element 
above into \eq{summedadiabatic}, squaring the resulting amplitude,
and summing over all allowed combinations of the four outgoing
momenta. Thus,
\begin{equation}
P_{\mathrm{create}}^{(4)} \sim 
\sum_{\mathbf{p}_1,\mathbf{p}_2,\mathbf{p}_3,\mathbf{p}_4 \in \Gamma}
\frac{J^2 \lambda_0^2
  \delta_{\mathbf{p}_1+\mathbf{p}_2+\mathbf{p}_3+\mathbf{p}_4,0}}
{V^2 \tau^2 (\omega(\mathbf{p}_1)+\omega(\mathbf{p}_2) +
  \omega(\mathbf{p}_3) + \omega(\mathbf{p}_4))^4 \omega(\mathbf{p}_1)
  \omega(\mathbf{p}_2) \omega(\mathbf{p}_3) \omega(\mathbf{p}_4)} 
\,.
\end{equation}
One can verify that this sum has the following scaling as $a\to 0$:
\begin{equation}
P_{\mathrm{create}}^{(4)} =
\label{pcreate}
\left\{ \begin{array}{ll} \tilde{O} \left(
  \frac{V J^2}{\tau^2} \right) \,, & d=1,2 \,, \vspace{5pt} \\
\tilde{O} \left( \frac{V J^2}{\tau^2 a} \right) \,, & d=3 \,.
\end{array} \right.
\end{equation}
By \eq{Jscaleweak} and \eq{pcreate}, 
\begin{equation}
P_{\mathrm{create}}^{(4)}= \left\{ \begin{array}{ll}
\tilde{O} \left( \frac{V}{\tau} \right) \,, & d=1,2 \,, \\
\tilde{O} \left( \frac{V}{\tau a^3} \right) \,, & d=3 \,.
\end{array} \right.
\end{equation}

Finally, we consider the process in which the time dependence of the
$\phi^4$ term causes a single particle to split into three. For this
process, the relevant matrix element is
\begin{equation}
\label{asplit}
\bra{\mathbf{p}_2,\mathbf{p}_3,\mathbf{p}_4} \frac{\lambda_0}{4 !} 
\sum_{\mathbf{x} \in \Omega} 
a^d \phi^4(\mathbf{x}) \ket{\mathbf{p}_1} = \frac{ \lambda_0
  \delta_{\mathbf{p}_2+\mathbf{p}_3+\mathbf{p}_4,\mathbf{p}_1}}
{4 V \sqrt{\omega(\mathbf{p}_1) \omega(\mathbf{p}_2)
    \omega(\mathbf{p}_3) \omega(\mathbf{p}_4)}} \,,
\end{equation}
where $\mathbf{p}_1$ is the momentum of the incoming particle. 
By our choice of path, the physical mass is $s$-independent to
first order in the coupling, and the $s$ dependence of the 
coupling is only logarithmically divergent as $a \to 0$. Thus,
by \eq{summedadiabatic},
\begin{equation}
\label{Psplitdef}
P_{\mathrm{split}} \sim \frac{J^2}{\tau^2 V^2}
\sum_{\mathbf{p}_2,\mathbf{p}_3,\mathbf{p}_4 \in \Gamma}
\frac{\lambda_0^2 \delta_{\mathbf{p}_2+\mathbf{p}_3+\mathbf{p}_4,\mathbf{p}_1}}
{(\omega(\mathbf{p}_2)+\omega(\mathbf{p}_3)+\omega(\mathbf{p}_4)-\omega(\mathbf{p}_1))^4 
\omega(\mathbf{p}_1) \omega(\mathbf{p}_2) \omega(\mathbf{p}_3) \omega(\mathbf{p}_4)} \,.
\end{equation}

Let us now examine the divergence structure of $P_{\mathrm{split}}$ as
$a \to 0$. In the limit of large volume, the sum converges to the
following integral:
%
\[
\frac{J^2}{\tau^2} \int_\Gamma \frac{d^d p_2}{(2\pi)^d}
\int_\Gamma \frac{d^d p_3}{(2\pi)^d} \frac{\lambda_0^2}{(\omega(\mathbf{p}_2)+\omega(\mathbf{p}_3) +
  \omega(\mathbf{p}_1-\mathbf{p}_2-\mathbf{p}_3) - \omega(\mathbf{p}_1))^4 \omega(\mathbf{p}_1) \omega(\mathbf{p}_2)
  \omega(\mathbf{p}_3) \omega(\mathbf{p}_1-\mathbf{p}_2-\mathbf{p}_3)}
\,.
\]
If this were divergent as $a \to 0$, then, by approximating the integrand 
with its value at large
$|\mathbf{p}_2|$ and $|\mathbf{p}_3|$, we would be able to isolate the
divergence:
\begin{equation}
P_{\mathrm{split}} \sim \frac{J^2 \lambda_0^2}{\tau^2
  \omega(\mathbf{p}_1)} \int_\Gamma d^d p_2 \int_\Gamma d^d p_3
\frac{1}{(|\mathbf{p}_2|+|\mathbf{p}_3|+|\mathbf{p}_2+\mathbf{p}_3|)^4
  |\mathbf{p}_2| |\mathbf{p}_3| |\mathbf{p}_2+\mathbf{p}_3|} \,.
\end{equation}
However, for $d=1,2,3$, this is convergent as $a \to 0$. Thus, recalling 
\eq{Jscaleweak}, we obtain
\begin{equation}
P_{\mathrm{split}} = O \left( \frac{J^2}{\tau^2} \right) =
\left\{ \begin{array}{ll} \tilde{O} \left( \frac{1}{\tau} \right) \,, 
& d=1,2 \,, \\
\tilde{O} \left( \frac{1}{\tau a^2} \right) \,, & d=3 \,.
\end{array} \right.
\end{equation}

We can consider two criteria regarding diabatic particle creation. If
our detectors are localized, we may be able to tolerate a low constant
density of stray particles created during state preparation. This
background is similar to that encountered in experiments and may not
invalidate conclusions from the simulation. Alternatively, one could adopt a
strict criterion by demanding that, with high probability, not even
one stray particle is created in the volume being simulated during
state preparation. This strict criterion can be quantified by
demanding that the adiabatically produced state has an inner product of at
least $1-\epsilon_{\mathrm{ad}}$ with the exact state. This
parameter $\epsilon_{\mathrm{ad}}$ is thus directly comparable
with $\epsilon_{\mathrm{trunc}}$, and the two sources of error can be
added. Applying the strict criterion, we demand that
$P_{\mathrm{create}}^{(2)}$, $P_{\mathrm{create}}^{(4)}$, and 
$P_{\mathrm{split}}$ each be of order
$\epsilon_{\mathrm{ad}}$ and obtain
\begin{equation}
\tau_{\mathrm{strict}} = \left\{ \begin{array}{ll} 
\tilde{O} \left(\frac{V}{\epsilon_{\mathrm{ad}}} \right) \,, &  d=1,2 \,,\\
\tilde{O} \left(\frac{V}{a^6 \epsilon_{\mathrm{ad}}} \right) \,, & d=3 \,.
\end{array} \right. 
\end{equation}
Applying the more lenient criterion that
$(P_{\mathrm{create}}^{(2)}+P_{\mathrm{create}}^{(4)})/V$ and
$P_{\mathrm{split}}$ each be of order $\epsilon_{\mathrm{ad}}$ yields
\begin{equation}
\tau_{\mathrm{lenient}} = \left\{ \begin{array}{ll} 
\tilde{O} \left( \frac{1}{\epsilon_{\mathrm{ad}}} \right) \,, & d=1,2 \,,\\
\tilde{O} \left( \frac{1}{a^6 \epsilon_{\mathrm{ad}}} \right) \,, & d=3 \,.
\end{array} \right.
\end{equation}
If a $k\th$-order Suzuki-Trotter formula is used, the asymptotic scaling of 
the total number of gates needed for adiabatic state preparation is
$O \big( (\mathcal{V} \tau)^{1+\frac{1}{2k}} \big) = O \big( (V
\tau/a^d)^{1+\frac{1}{2k}} \big)$. Thus,
\begin{eqnarray}
\label{Gstrict}
G_{\mathrm{adiabatic}}^{\mathrm{strict}} & = &
\left\{ \begin{array}{ll} \tilde{O} \left(
\left( \frac{V^2}{a^d \epsilon_{\mathrm{ad}}} \right)^{1+\frac{1}{2k}}
\right) \,, & d=1,2 \,, \\
\tilde{O} \left( \left( \frac{V^2}{a^9 \epsilon_{\mathrm{ad}}}
\right)^{1+ \frac{1}{2k}} \right) \,, & d=3 \,. \end{array} \right. \\
\label{Glenient}
G_{\mathrm{adiabatic}}^{\mathrm{lenient}} & = & 
\left\{ \begin{array}{ll} \tilde{O} \left(
\left( \frac{V}{a^d \epsilon_{\mathrm{ad}}} \right)^{1+\frac{1}{2k}}
\right) \,, & d=1,2 \,, \\
\tilde{O} \left( \left( \frac{V}{a^9 \epsilon_{\mathrm{ad}}}
\right)^{1+ \frac{1}{2k}} \right) \,, & d=3 \,. \end{array} \right.
\end{eqnarray}

\subsubsection{Strong Coupling}
\label{strong}

In two and three spacetime dimensions, we can obtain a strongly coupled
(that is, non-perturbative) field theory by approaching the 
phase transition (\sect{QPT}). As in the case of weak
coupling, the necessary time for adiabatic state preparation depends
on various physical parameters of the system being simulated,
including the momentum of the incoming particles, the volume, the
strength of the final coupling, the number of spatial dimensions, and
the physical mass. To keep the discussion concise, we restrict our
discussion to the case of ultrarelativistic incoming particles, with
coupling strength close to the critical value. Under these conditions,
the incoming particles can produce a shower of many 
$(n_{\mathrm{out}} \sim p/m)$ outgoing particles. Because of the strong
coupling, perturbation theory is inapplicable and, even
if it could be used, would take exponential computation in the number of
outgoing particles.

In the strongly coupled case, we vary the Hamiltonian \eq{HS}
with $s$ by keeping the bare mass constant at $m_0$ and setting the
bare coupling to $s \lambda_0$. We choose $\lambda_0$ only slightly
below the critical value $\lambda_c$, so that at $s=1$ the system
closely approaches the phase transition, as illustrated in 
Fig.~\ref{paths}. Examining \eq{thetap2} suggests that we can estimate
phase errors by understanding the behavior of $m^2(s)$ at $s=0$ and
$s=1$, without needing to know exactly what happens in between. 
From \eq{mu1}, 
\begin{equation}
\left. \frac{d m^2}{d s} \right|_{s=0} = \left\{ \begin{array}{ll} 
\frac{\lambda_0}{8 \pi} \log \left( \frac{64}{m_0^2 a^2} \right) \,, &
d=1 \,, \\
\frac{25.379}{16
  \pi^2} \frac{\lambda_0}{a} \,, & 
d=2 \,,
\end{array} \right.
\end{equation}
and, 
from \eq{numass} and \eq{nu}, 
\begin{equation}
\left. \frac{d m^2}{d s} \right|_{s=1} \sim \left\{
\begin{array}{ll}
-2 \lambda_0 (\lambda_c - \lambda_0) \,, & 
d=1 \,, \\
-1.26 \lambda_0 (\lambda_c - \lambda_0)^{0.26} \,, & 
d=2 \,.
\end{array} \right.
\end{equation}
Thus, \eq{phasecrit} yields
\begin{equation}
\label{unconcrete}
J = \tilde{O} \left( \sqrt{ \frac{\tau \lambda_0}{a^{d-1} p^2 \mathcal{D}}}
\right) \,, \quad d=1,2 \,,
\end{equation}
under the assumption that $(\lambda_c - \lambda_0)$ is very
small.

The result \eq{thetap2} rests on two approximations, a Taylor
expansion to second order in \eq{thetaj} and an
approximation of a sum by an integral in \eq{thetap1}. The validity
conditions for these approximations become most stringent at $s=1$, 
where the derivatives of $m^2$ with respect to $s$ become large. 
Working out the $O(J^{-4})$ term in \eq{thetap2} at $s=1$, one
finds that it will be much smaller than the $O(J^{-2})$ term at $s=1$
provided
\begin{equation}
\label{Taylorcrit}
J \gg \frac{1}{\lambda_c - \lambda_0} \,.
\end{equation}
Similarly, higher-order terms in the Taylor expansion are suppressed
by additional powers of $\frac{1}{J(\lambda_c-\lambda_0)}$. The criterion
\eq{Taylorcrit} also suffices to justify the approximation of the sum
by an integral in \eq{thetap1}. 

We must next consider adiabaticity to determine $\tau$. In
the ultrarelativistic limit, the relevant energy gap $\gamma$ is $\sim
\frac{m^2}{p}$.\footnote{This is the case unless there exists a bound
state of three particles whose binding energy is at least $2m$.} 
This takes its minimum value at $s=1$, namely,
\begin{equation}
\label{gammamin}
\gamma_{\min} \sim \left\{ \begin{array}{ll}
\frac{(\lambda_c - \lambda_0)^2}{p} \,, & d=1 \,, \\
\frac{(\lambda_c - \lambda_0)^{1.26}}{p} \,, & d=2 \,.
\end{array} \right.
\end{equation}
Unlike in the perturbative case, we cannot make a detailed
quantitative analysis. However, under the condition~(\ref{Taylorcrit}),
the energy gap $\gamma$ changes only slightly in any adiabatic step of
the process described by \eq{Mj}, \eq{Uj}, and \eq{Mprod}. Thus we
apply the traditional adiabatic approximation \eq{traditional} and
find that each adiabatic step contributes an excitation amplitude of order
$\frac{1}{\tau \gamma^2}$. Depending on the relative phases of the excitation
amplitudes arising from the $J$ steps, the total excitation amplitude
could be as high as $\frac{J}{\tau \gamma^2}$. Indeed, a detailed
analysis applying \eq{traditional} to \eq{Mj}, \eq{Uj}, and
\eq{Mprod} suggests that this bound is not overly pessimistic.
Thus, to keep the error probability at some small constant
$\epsilon_{\mathrm{ad}}$,
we choose 
\begin{equation}
\label{tauscale}
\tau \sim \frac{J}{\gamma^2 \sqrt{\epsilon_{\mathrm{ad}}}} \,.
\end{equation}

We now consider asymptotic scaling with $p$ for fixed
$\lambda_0$. To achieve continuum-like behavior we need $a \ll
\frac{1}{p}$. Thus \eq{unconcrete} yields
\begin{equation}
\label{pscaleJ}
J \sim \tau^{1/2} p^{(d-3)/2} \,, \quad d=1,2 \,.
\end{equation}
Substituting \eq{Taylorcrit} and \eq{gammamin} into
\eq{tauscale}, we see that we need
\begin{equation}
\label{cond1}
\tau \gtrsim p^2 \,, \quad d=1,2 \,.
\end{equation}
Substituting \eq{pscaleJ} and \eq{gammamin} into
\eq{tauscale}, we see that we also need 
\begin{equation}
\label{cond2}
\tau \gtrsim p^{d+1} \,, \quad d=1,2 \,.
\end{equation}
The scaling $\tau = O(p^{d+1})$ for $d=1,2$ suffices to satisfy both 
conditions~(\ref{cond1}) and (\ref{cond2}). Thus, by the results of
\sect{Trotter}, the total number of gates needed for adiabatic state
preparation scales as
\begin{eqnarray}
G_{\mathrm{strong}} & = & 
O \left( \left( V \tau \right)^{1+o(1)} p^{d+1+o(1)} \right) \\
& = & O \left( V^{1+o(1)} p^{2d+2+o(1)} \right),
\label{strongpscale}
\end{eqnarray}
for $d=1,2$. 

Next, we consider asymptotic scaling with $(\lambda_c - \lambda_0)$
for fixed $p$. 
We substitute \eq{unconcrete} into \eq{tauscale}, obtaining
\begin{equation}
\tau \sim \left\{ \begin{array}{ll} \left( \frac{1}{\lambda_c -
    \lambda_0} \right)^8 \,, & d=1 \,, \\
\left( \frac{1}{\lambda_c - \lambda_0} \right)^{5.04} \,, & d=2 \,.
\end{array} \right. 
\end{equation}
The scaling of $J$ as $\sqrt{\tau}$ in \eq{unconcrete} automatically 
satisfies the condition (\ref{Taylorcrit}). 
The spacing between particles in the in and out states must be of order 
$1/m$. Thus, with constant $a$, $\mathcal{V} \sim 1/(\lambda_c -
\lambda_0)^{\nu d}$. If a $k\th$-order Suzuki-Trotter formula is used, the
necessary number of quantum gates, $G_{\mathrm{strong}}$, scales as
$O\big( ( \mathcal{V} \tau )^{1+\frac{1}{2k}} \big)$. Thus,
\begin{equation}
\label{stronglambdascale}
G_{\mathrm{strong}} \sim \left\{ \begin{array}{ll} \left( \frac{1}{\lambda_c -
    \lambda_0} \right)^{9 \left( 1+\frac{1}{2k} \right)} \,, & d=1 \,, \\
\left( \frac{1}{\lambda_c - \lambda_0} \right)^{6.3
  \left(1+\frac{1}{2k} \right) } \,, & d=2 \,.
\end{array} \right. 
\end{equation}
Note that one could improve this scaling by choosing a more optimized
adiabatic state-preparation schedule, which slows down as the gap gets
smaller. 


\subsection{Suzuki-Trotter Formulae for Large Lattices}
\label{Trotter}

It appears that, while scaling with $t$ has been thoroughly studied,
little attention has been given to scaling of quantum simulation
algorithms with the number of lattice sites $\mathcal{V}$.
Using a result of Suzuki and elementary Lie algebra theory, we derive
linear scaling provided the Hamiltonian is local.

For any positive integer $k$ and any pair of Hamiltonians $A, B$,
\begin{equation}
\label{hightrotter}
\left( e^{i A \alpha_1 t/n} e^{i B \beta_1 t/n} e^{i A \alpha_2 t/n}
e^{i \beta_2 B t/n} \dotsm e^{i A \alpha_r t/n} \right)^n = e^{i (A + B) t} 
+ O(t^{2k+1}/n^{2k}) \,,
\end{equation}
where $r = 1+5^{k-1}$ and
$\alpha_1,\ldots,\alpha_r,\beta_1,\ldots,\beta_{r-1}$ are specially
chosen coefficients such that $\sum_{j=1}^r \alpha_j = 1$ and
$\sum_{j=1}^{r-1} \beta_j = 1$ \cite{Suzuki}. Thus, using the $k\th$-order 
Suzuki-Trotter formula \eq{hightrotter}, one can simulate evolution 
for time $t$ with $O\big( t^{\frac{2k+1}{2k}} \big)$ quantum gates
\cite{Cleve_sim}. To determine the $\mathcal{V}$ scaling, we use the
following standard theorem (cf. the Baker-Campbell-Hausdorff
formula).
\begin{theorem}\label{BCH}
Let $A$ and $B$ be elements of a Lie algebra defined over any field of
characteristic 0. Then $e^{A} e^{B} = e^{C}$, where $C$ is a formal
infinite sum of elements of the Lie algebra generated by $A$ and $B$.
\end{theorem}
$A$ and $B$ generate a Lie algebra by commutation and linear
combination. Thus, without requiring any explicit calculation,
Theorem~\ref{BCH} together with \eq{hightrotter} implies
\begin{eqnarray}
\left( e^{i A \alpha_1 t/n} e^{i B \beta_1 t/n} \dotsm e^{i A
  \alpha_r t/n} \right)^n & = & e^{i (A + B) t} + 
\bigg[\frac{1}{n}\sum_{j=1}^{n} 
e^{i (A + B) tj/n} \Delta_{2k+1} e^{i (A + B) t(n-j)/n} \bigg]
\frac{t^{2k+1}}{n^{2k}} \nn\\
& & + \,\, O(n^{-(2k+1)})  \,,
\end{eqnarray}
\begin{equation}
\Big\lVert 
\left( e^{i A \alpha_1 t/n} e^{i B \beta_1 t/n} \dotsm e^{i A
  \alpha_r t/n} \right)^n - e^{i (A + B) t} 
\Big\rVert
= \frac{t^{2k+1}}{n^{2k}}  \| \Delta_{2k+1} \| 
+  O(n^{-(2k+1)})\,,
\qquad\qquad\,\,
\end{equation}
where $\Delta_{2k+1}$ is a linear combination of nested
commutators. In general, $\| \Delta_{2k+1} \|$ could be as large as 
$\left( \max \left\{ \|A\|, \|B\| \right\} \right)^{2k+1}$. However, 
by \eq{canonical}, one sees that, for the pair of local 
Hamiltonians $H_{\phi}, H_{\pi}$, $\|\Delta_{2k+1}\| =
O(\mathcal{V})$, for any fixed $k$. Thus, one needs only
$n = O \big( t^{\frac{2k+1}{2k}} \mathcal{V}^{\frac{1}{2k}} 
\big)$. Recalling the $O(\mathcal{V})$ cost for simulating each 
$e^{i H_{\phi} \delta t}$ or $e^{i  H_{\pi} \delta t}$, one sees that 
the total number of gates scales as $O\big( ( t \mathcal{V} 
)^{1+\frac{1}{2k}} \big)$. Note that this conclusion may be of 
general interest, as it applies to any lattice Hamiltonian for which 
non-neighboring terms commute.
Our scaling analysis also applies to Suzuki-Trotter decompositions 
\cite{suzuki:1993,wiebe:2010} for time-dependent Hamiltonians.

In the case of strong coupling, we care not only about how the number
of gates scales with $\mathcal{V}$, but also about scaling with
$p$. In the presence of high-energy incoming particles, the field can
have large distortions from its vacuum state. For example, if
$\bra{\psi} \phi(\mathbf{x}) \ket{\psi}$ is large, then local terms in
$\Delta_{2k+1} \ket{\psi}$ such as $\pi(\mathbf{x}) \phi(\mathbf{x})^3
\ket{\psi}$ can become large. We can obtain a heuristic upper bound on
this effect by noting that, in the strongly coupled case, $m_0^2 > 0$,
so each local term in $H$ is a positive operator. Thus, if $\bra{\psi}
H \ket{\psi} \leq E$, then the expectation value of each of the local
terms is bounded above by $E$. Using $E$ as a simple estimate of the
maximum magnitude of a local term, we see that $\Delta_{2k+1}
\ket{\psi}$, which is a sum of $O(\mathcal{V})$ terms, each of which
is of degree $2k+1$ in the local terms of $H$, has magnitude at most
$O(\mathcal{V} E^{2k+1})$, or in other words $O(\mathcal{V}
p^{2k+1})$. Recalling that $a$ scales as a small multiple of $1/p$, we
see that $\Delta_{2k+1} \ket{\psi} = O(V p^{2k+1+d})$. Thus, $n =
O(p^{1+(1+d)/2k}t^{1+1/2k})$. Each timestep requires $O(\mathcal{V}) =
O(Vp^d)$ gates to implement. Thus, the overall scaling is
$O(p^{d+1+o(1)} (tV)^{1+o(1)})$ quantum gates to simulate the strongly
coupled theory at large $p$.


\subsection{Measurement of Occupation Numbers}
\label{measurement}

In \sect{lattice}, we saw that the free theory can be decomposed into
bosonic modes of definite momentum. In this section, we analyze the
complexity of measuring the occupation of these modes. 

By \eq{laddercom1}, the operator $L^{-d} a_{\mathbf{p}}^\dag
a_{\mathbf{p}}$ has eigenvalues $0,1,2,3,\ldots$, which indicate the
number of particles in momentum mode $\mathbf{p}$. By the method of
phase estimation \cite{Kitaev_book}, measuring $L^{-d}
a_{\mathbf{p}}^\dag a_{\mathbf{p}}$ reduces to simulating $e^{i
  L^{-d} a_{\mathbf{p}}^\dag a_{\mathbf{p}} t}$ for various $t$. By \eq{a},
\begin{equation}
L^{-d} a_{\mathbf{p}}^\dag a_{\mathbf{p}} = \Pi_{\mathbf{p}} +
\Phi_{\mathbf{p}} + \chi_{\mathbf{p}} \,,
\end{equation}
where
\begin{eqnarray}
\label{Pip}
\Pi_{\mathbf{p}} & = & \frac{a^d}{\mathcal{V}} \sum_{\mathbf{x},\mathbf{y}
  \in \Omega} e^{i \mathbf{p} \cdot (\mathbf{x}-\mathbf{y})} 
\frac{1}{2 \omega(\mathbf{p})} \pi(\mathbf{x}) \pi(\mathbf{y}) \,, \\
\label{Phip}
\Phi_{\mathbf{p}} & = & \frac{a^d}{\mathcal{V}} \sum_{\mathbf{x},\mathbf{y}
  \in \Omega} e^{i \mathbf{p} \cdot (\mathbf{x}-\mathbf{y})}
\frac{\omega(\mathbf{p})}{2} \phi(\mathbf{x}) \phi(\mathbf{y}) \,, \\
\chi_{\mathbf{p}} & = & \frac{i a^{d}}{2\mathcal{V}}
\sum_{\mathbf{x},\mathbf{y} \in \Omega} 
e^{i \mathbf{p} \cdot (\mathbf{x} - \mathbf{y})} \left[ 
\phi(\mathbf{x}) \pi(\mathbf{y}) - \pi(\mathbf{x}) \phi(\mathbf{y}) \right]
\,.
\end{eqnarray}
Simulating  $e^{i L^{-d} a_{\mathbf{p}}^\dag a_{\mathbf{p}} t}$ with standard
Suzuki-Trotter formulae is not very efficient because of the
$O(\mathcal{V}^2)$ mutually non-commuting terms in
$\chi_{\mathbf{p}}$. For example, the method of \cite{Cleve_sim}
requires $O(\mathcal{V}^4)$ quantum gates for this task. 

However, there is a simple remedy for this inefficiency. 
A short calculation shows that
\begin{equation}
2(\Phi_{\mathbf{p}} + \Pi_{\mathbf{p}}) = \frac{1}{L^d} \big(
a_{\mathbf{p}}^\dag a_{\mathbf{p}} + a_{-\mathbf{p}}^\dag
a_{-\mathbf{p}} \big) + \id .
\end{equation}
By \eq{hightrotter}, the problem of simulating time evolution induced
by $\frac{1}{L^d} \big( a_{\mathbf{p}}^\dag a_{\mathbf{p}} + a_{-\mathbf{p}}^\dag
a_{-\mathbf{p}} \big) + \id$ thus reduces to the problem of simulating each
of $e^{i \Pi_{\mathbf{p}} t}$ and  $e^{i \Phi_{\mathbf{p}} t}$. 
One can efficiently simulate each of 
$e^{i \Pi_{\mathbf{p}} t}$ and  $e^{i \Phi_{\mathbf{p}} t}$ 
by going to the basis in which it is diagonal,
reversibly computing the induced phases, and then applying them by
phase kickback (\sect{reversible}).

The nested commutators appearing in the leading correction to the 
$k\th$-order Suzuki-Trotter approximation have norm $O(\mathcal{V})$. To see
this, first note that by \eq{Pip} and \eq{Phip}, each power of
$\Pi_{\mathbf{p}}$ or $\Phi_{\mathbf{p}}$ contributes a
coefficient $\frac{a^d}{\mathcal{V}}$ and a summation over
$\mathcal{V}^2$ pairs of lattice points. Thus a nested commutator at
$(k+1)\th$ order in $\Pi_{\mathbf{p}}$ and $\Phi_{\mathbf{p}}$ consists of
a coefficient $\frac{a^{d(k+1)}}{\mathcal{V}^{k+1}}$ times a sum of
$\mathcal{V}^{2(k+1)}$ nested commutators of $\phi(\mathbf{x})
\phi(\mathbf{y})$ operators and $\pi(\mathbf{x}) \pi(\mathbf{y})$
operators. Out of these $\mathcal{V}^{2(k+1)}$ nested commutators,
only $O(\mathcal{V}^{k+2})$ are non-zero, because of the product of $k$ delta
functions arising from \eq{canonical}. The sum of
$O(\mathcal{V}^{k+2})$ non-zero terms all with coefficient
$\frac{a^{d(k+1)}}{\mathcal{V}^{k+1}}$ yields, by the triangle
inequality, a total operator norm of $O(\mathcal{V})$.

By the above analysis, with a $k\th$-order Suzuki-Trotter formula,
it suffices to use $O\big( t^{\frac{2k+1}{2k}}
\mathcal{V}^{\frac{1}{2k}} \big)$ timesteps. Each step requires
$\tilde{O}(\mathcal{V}^2)$ gates to simulate, owing to the cost of reversibly
computing the double sums in $\Pi_{\mathbf{p}}$ and
$\Phi_{\mathbf{p}}$, so the total cost of measuring the combined
occupation of modes $\mathbf{p}$ and $-\mathbf{p}$ is $\tilde{O}\big(
t^{\frac{2k+1}{2k}} \mathcal{V}^{2+\frac{1}{2k}} \big)$. To
distinguish eigenvalues separated by a gap $\gamma$ using phase
estimation, we can choose $t \sim 1/\gamma$. All of the
eigenvalues of $\Phi_{\mathbf{p}} + \Pi_{\mathbf{p}}$ are separated by
gaps of $1/2$. Thus, $t = O(1)$.

Subtracting $1$ from the eigenvalue produced by phase estimation 
yields the total occupation of modes $\mathbf{p}$ and $-\mathbf{p}$.
This method has the deficiency of failing to distinguish particles with
momentum $\mathbf{p}$ from those with momentum
$-\mathbf{p}$. Furthermore, the complexity of this measurement scales
quadratically with $V$. Both of these problems can be corrected by
simulating localized detectors. Let $\mathcal{V}_D \subset \mathcal{V}$ be
the region occupied by our simulated detector. As the simplest case, one can
consider $\mathcal{V}_D = a \mathbb{Z}_{\hat{L}_D}^d$; in general, 
$\mathcal{V}_D$ is any spatial translation of this. 
Then, the corresponding Fourier basis on this region is 
$\Gamma_D = \frac{2 \pi}{L_D} \mathbb{Z}_{\hat{L}_D}^d$,
where $L_D = a \hat{L}_D$ is the length of the detector. A set of operators
$\{a_{\mathbf{p},D}, a_{\mathbf{p},D}^\dag | \mathbf{p} \in \Gamma_D\}$
can be obtained by replacing $\Omega$, $\Gamma$, and 
$\mathcal{V}$ with their local counterparts, $\Omega_D$, $\Gamma_D$, and
$\mathcal{V}_D$, in \eq{a}. These operators obey all the commutation 
relations that one expects of creation and annihilation operators.

Physically, the operator $a_{\mathbf{p},D}^\dag a_{\mathbf{p},D}$ can be
interpreted as a number operator for the region $\mathcal{V}_D$. By the
results of \sect{locality}, $a_{\mathbf{p},D}^\dag a_{\mathbf{p},D}$ can
be exponentially well approximated by a linear combination of $\phi$ and
$\pi$ operators with support only within $\mathcal{V}_D$ or 
a distance $O(1/m_0)$ of $\mathcal{V}_D$. Thus, by phase
estimation, we can measure in the eigenbasis of
$\frac{1}{L_D^d} \big( a_{\mathbf{p}, D}^\dag a_{\mathbf{p}, D} +
a_{-\mathbf{p},D}^\dag a_{-\mathbf{p},D} \big) + \id$ using 
$O\big(\mathcal{V}_D^{2 + \frac{1}{2k}}\big)$ quantum gates, independently of
$V$.

The localized detector 
$\frac{1}{L_D^d} \big( a_{\mathbf{p}, D}^\dag a_{\mathbf{p}, D} +
a_{-\mathbf{p},D}^\dag a_{-\mathbf{p},D} \big) + \id$ detects only
particles in the region $\mathcal{V}_D$. Thus it is partially
momentum-resolving and partially position-resolving. In accordance
with the uncertainty principle, the momentum resolution of this
detector must be lower than that of 
$\frac{1}{L^d} \big( a_{\mathbf{p}}^\dag a_{\mathbf{p}} +
a_{-\mathbf{p}}^\dag a_{-\mathbf{p}} \big) + \id$. This is reflected
in the fact that the momentum lattice $\Gamma_D$ is more
coarse-grained than $\Gamma$. After a collision, the shower of
outgoing particles will be directed outward from the collision
region. 
Thus, to resolve the $\mathbf{p}$ vs $-\mathbf{p}$ ambiguity, it should 
suffice to surround the collision region with a small number of localized 
detectors. For example, one could simulate $2d$ detectors corresponding to 
the faces of a $d$-dimensional cube surrounding the collision region.

Number operators with different momenta or within different spatial regions
commute (\sect{locality}).
One can thus measure $q$ momentum modes within
each of $r$ spatial regions simply by repeating the phase estimation 
procedure for each number operator, with total complexity 
$O \big( r q \mathcal{V}_D^{2+\frac{1}{2k}} \big)$.


\subsection{Detection of Bound States}
\label{detectors}

In some scattering processes, especially at strong coupling, the
outgoing particles may be in bound states. In this case, it may not be
desirable to turn off the coupling adiabatically and measure the
occupation numbers of momentum modes of the free theory. Instead, we
can measure the total energy and momentum within each of a set of spatial
regions. If the regions are small compared with the separation between
particles (some of which may be composite), then each region contains
at most one particle, and we obtain its energy and spatial
momentum. There is some probability that multiple particles  are in a
single region, in which case our simulated ``detector'' fails to
resolve them. Although the analogy between our measurement and real
detectors is rather loose, both share this basic feature of
limited resolution.

To measure the energy in a region $R \subseteq \Omega$, we multiply the 
Hamiltonian density by an envelope function, $f_R(\mathbf{x})$, localized 
in $R$: 
\begin{equation}
H[R] = \sum_{\mathbf{x} \in \Omega} a^d f_R(\mathbf{x}) 
\left[ \frac{1}{2} \pi(\mathbf{x})^2
  + \frac{1}{2} \left( \nabla_a \phi \right)^2 (\mathbf{x}) +
  \frac{1}{2} m_0^2 \phi(\mathbf{x})^2 + \frac{\lambda_0}{4!}
  \phi(\mathbf{x})^4 \right] \,. 
\end{equation}
The operator $H[R]$ will then be sensitive only to particles in $R$. 

The operator $H[R]$ does not commute with the full Hamiltonian, $H$. 
Even in the vacuum, not only the value of $H[R]$ but also 
its uncertainty will be non-zero. The vacuum expectation value of
$H[R]$ can be measured and subtracted. The uncertainty in $H[R]$ is
undesirable because it reduces the signal-to-noise ratio of the
measurements. In $d=1,2$ (the cases of interest for strong coupling), 
choosing $f_R$ to be Gaussian suffices to keep this problem under control; 
the rapid decay of the Fourier transform of $f_R$ suppresses the 
contribution of short-wavelength modes to $H[R]$. Calculations show that, as $a \to 0$, the 
variance in $H[R]$ scales as
\begin{equation}
\sigma^2 \sim \left\{ \begin{array}{cl} \frac{1}{r^3 m}\,, & d=1, \\[3pt]
\frac{1}{r^2} \log \left( \frac{1}{ma} \right)\,, & d=2,
\end{array} \right.
\end{equation}
where $r$ is the spatial width of the Gaussian envelope $f_R$. Thus,
one can lower the noise by increasing the radius of the detector, and
this radius need only scale at most logarithmically in $1/a$. 

By the canonical commutation relations, $\left[ H[R], H[R'] \right] =
0$ whenever the supports of $f_R$ and $f_{R'}$ are
disjoint. Truncating the tails of the Gaussian envelopes to obtain
this property approximates the exactly Gaussian detectors exponentially well. 
The spacetime location of the particle detection relative to the scattering 
vertex indicates the particle's velocity. Together with the measured energy,
this allows an estimate of the particle's momentum and hence the
identification of bound states. Geometrically, one sees that a
detector of radius $r$ at a distance $\ell$ from the scattering point 
yields a measurement of particle velocity with uncertainties in angle
and magnitude each scaling as $r/\ell$. Thus, it suffices to make $\ell$ 
scale at most logarithmically in $1/a$ to achieve constant detector 
resolution in $d=1,2$.

One can use phase estimation to measure energy to precision
$\Delta_E$ in a region $R$ by implementing $e^{-iH[R]t}$ for $t \sim
\frac{1}{\Delta_E}$. $e^{-iH[R]/\Delta_E}$ can be implemented with $O
\big( \left( |R|/\Delta_E \right)^{1+o(1)} \big)$ quantum gates
(\sect{Trotter}). (Note that the number of gates needed 
scales as $p^{d+1}$ in the limit
of large particle momentum; see \sect{Trotter}.) If one wishes to
detect outgoing particles in the full solid angle around the
scattering region, then one needs to simulate a spherical shell of
detectors. Such a shell encompasses only a small fraction of
the total volume being simulated. Hence, the cost of the
Suzuki-Trotter simulation of these detectors is small compared with
that of simulating the time evolution according to the full
Hamiltonian. 


\subsection{Localized Creation Operators}
\label{locality}

In this section we show that the operator $a_{\mathbf{x}}^\dag$
defined in \eq{axdag} is quasilocal, namely, it can be expanded in
the form
\begin{equation}
a_{\mathbf{x}}^\dag = \sum_{\mathbf{y}} \left[ f(\mathbf{y}-\mathbf{x})
  \pi(\mathbf{y}) + g(\mathbf{y} - \mathbf{x}) \phi(\mathbf{y})
  \right] \,,
\end{equation}
where $f(\mathbf{y}-\mathbf{x})$ and $g(\mathbf{y} - \mathbf{x})$ are
each either zero or exponentially small for $| \mathbf{y} -
\mathbf{x}| \gg 1/m_0$. This implies that the Hamiltonian
$H_{\mathbf{\psi}}$ defined in \eq{hpsi} is quasilocal whenever
$\psi$ has local support. Therefore simulating time evolution
according to $H_{\psi}$ in order to create a wavepacket as in
\sect{details} does not disturb previously created wavepackets provided
they are well separated. Furthermore, the number of quantum gates
needed to simulate $H_{\psi}$ scales only with the number of lattice
sites in the support of $\psi$, rather than with $\mathcal{V}$. Similarly,
the quasilocality of $a_{\mathbf{x}}^\dag$ allows us to show in
\sect{measurement} that localized detectors can be simulated with
complexity scaling with their volume rather than with $V$.

By \eq{axdag} and \eq{a} one calculates
\begin{equation}
a_{\mathbf{x}}^\dag = \frac{1}{2} \phi(\mathbf{x}) - i
\sum_{\mathbf{y} \in \Omega} a^d f_d(\mathbf{y}-\mathbf{x}) \pi(\mathbf{y})
\,,
\end{equation}
where
\begin{equation}
f_d(\mathbf{y} - \mathbf{x}) = 
\sum_{\mathbf{p} \in \Gamma} L^{-d} e^{i \mathbf{p} \cdot (\mathbf{y}- \mathbf{x})}
\frac{1}{2 \omega(\mathbf{p})}
\,.
\end{equation}
Near the infinite-volume and continuum limits,
\begin{equation}
f_d(\mathbf{y} - \mathbf{x}) \simeq \int_\Gamma
\frac{d^d p}{(2 \pi)^d}
 \frac{e^{i \mathbf{p} \cdot (\mathbf{y} 
- \mathbf{x})}}{2 \sqrt{\mathbf{p}^2 + m_0^2}} \,.
\end{equation}
For $|\mathbf{y}-\mathbf{x}| \gg 1/m_0$, this function decays exponentially with
characteristic length $1/m_0$. For example, in $d=1$ with infinite
volume,
\begin{equation}
f_1(y - x) \simeq \int_{-\infty}^{\infty} \frac{dp}{4 \pi}
\frac{e^{ip(y-x)}}{\sqrt{p^2+m_0^2}} = \frac{1}{2\pi} K_0(m_0 |y-x|) \,,
\end{equation}
where $K_0$ is the modified Bessel function of the second kind, 
which satisfies
\begin{equation}
K_0(z) \sim \sqrt{\frac{\pi}{2z}} e^{-z} 
\,,\quad \textrm{$z \rightarrow \infty$.}
\end{equation}


\section{Some Field-Theoretical Aspects}
\label{fieldtheory}

This section describes some quantum field-theoretical details,
beginning with the perturbative renormalization of the mass
in \sect{sec:renorm}.

Section~\ref{sec:eft} gives a brief introduction to effective field 
theory, which is now a well-developed formalism underlying our modern 
understanding of quantum field theory. In its regime of validity,
typically below a particular energy scale, an effective field theory
reproduces the behavior of the full theory. Although it consists of
infinitely many terms, it can be truncated, with corresponding finite
and controllable errors.

In \sect{sec:a}, we analyse the effect of discretizing the spatial
dimensions of the continuum $\phi^4$ quantum field theory.
The discretized Lagrangian can be thought of as the leading 
contribution to an effective field theory. From the leading operators
left out we can thus infer the scaling of the error associated with  
a non-zero lattice spacing, $a$.
First, in \sect{ssec:eft}, we obtain the general form of the
effective field theory, including the scaling of the coefficients of
different operators. 
Next, in \sect{sec:weak}, we explicitly calculate the coefficients
of the leading operators in the complete effective field theory
at weak coupling, matching the full and effective theories at
an energy-momentum scale $p \sim 1/a$.
The effective field theory is demonstrated to consist of three different
classes of operators, shown with the scaling of their coefficients
in Table~\ref{table:eftshort}. At strong coupling, the operators and
their scaling remain the same at the matching scale, although the
explicit coefficients are no longer calculable.  However, the running
of the coefficients down to lower energies is determined by their
anomalous dimensions, which depend on the coupling strength. These
anomalous dimensions modify the scaling; at weak coupling the
modification is small, but at strong coupling it could be larger (though
the scaling will remain polynomial).

Section~\ref{sec:vol} addresses finite-volume effects, which should be
small, since the interactions in our field theory are short-range.
This expectation can be confirmed and quantified in the perturbative regime.
From a technical perspective, the finite volume means that, in the 
calculations of operator coefficients, integrals over loop momenta become 
(discrete) sums.
Consequently the coefficients of the operators are modified. The results
for the operator $\phi^6$ are shown in Table~\ref{table:coeff}. Note that 
magnitude of the coefficient is increased for finite length $L$.

In \sect{sec:Veff}, we assess another effect of a finite volume:
state preparation is affected, since asymptotic wavepackets are then 
only approximately free. As before, corrections should be small, since
the interactions are short-range, and we can confirm this claim in
the perturbative regime, this time by calculating the effective potential
in the Born approximation. The leading asymptotic behavior of the
effective potential is given in Table~\ref{table:Veff}.

\noindent
\medskip
\begin{table}[hbt]
\begin{center}
\begin{tabular}{|c|c|c|}
\hline \T\B
Class & Operators & Scaling of coupling 
\\
\hline
\hline
\T\B
I & $\phi^{2n}$ ($n \geq 3$) & $\lambda^n a^{2n-D}$ \\[3pt]
\hline
\T\B
II & $\phi \partial_{\bf x}^{2l} \phi$ ($l \geq 2$) & $a^{2l-2}$ 
\\
\hline
\T
III & $\phi^{2j+1} \partial_{\bf x}^{2l} \phi$ 
& $\lambda^{j+1} a^{2j+2l+2-D}$ \\
\B & ($j\geq 1$, $l \geq2$) 
&  \\
\hline
\end{tabular}
\vspace{6pt}
\caption{Effective field theory operators fall into three classes 
(\sect{sec:a}).
The general operator in each class is shown, with the canonical scaling 
of its coefficient in $D$ spacetime dimensions.
}
\label{table:eftshort}
\end{center}
\end{table}
\medskip

\noindent
\medskip
\begin{table}[hbt!]
\begin{center}
\begin{tabular}{|c|l|}
\hline \T\B
$D$ & 
\hspace{50pt} Coefficient of $\phi^6/6!$ \\
\hline 
\rule{0pt}{4.0ex}
2 \T & 
$ -\frac{45}{64\pi^5}\lambda^3 a^4 
\left[ 1 + \frac{20}{3} \frac{1}{\hat{L}^2} \right] $
\\[10pt]
3 & 
$  -\frac{5}{64\pi^5}\lambda^3 a^3 
\left[ 10\sqrt{2} + \frac{43\sqrt{2}}{\hat{L}^2} \right] $
\\[10pt]
4 & 
$ -\frac{15}{128\pi^5}\lambda^3 a^2 
\left[ 2(2\sqrt{3}+\pi) + \frac{4}{9}(26\sqrt{3}+9\pi)\frac{1}{\hat{L}^2}
\right] $
\\[10pt]
\hline
\end{tabular}
\vspace{6pt}
\caption{Wilson coefficient of operator $\phi^6/6!$ in effective field 
theory for $\phi^4$ theory in $D$ dimensions, with a finite number
$2\hat{L}$ of lattice sites in each dimension (\sect{sec:vol}).
Corrections to the square-bracketed expressions are of order 
$(m^2a^2,1/\hat{L}^3)$.
}
\label{table:coeff}
\end{center}
\end{table}

\noindent
\medskip
\begin{table}[hbt!]
\begin{center}
\begin{tabular}{|c|l|}
\hline \T\B
$D$ & 
\hspace{25pt}$V(r \rightarrow \infty)$ \\
\hline 
\rule{0pt}{4.0ex}
2 \T & 
$-\frac{\lambda^2}{32 m^3}\frac{1}{\sqrt{\pi mr}} e^{-2mr}
$
\\[10pt]
3 & 
$-\frac{\lambda^2}{64\pi^{3/2} m}\frac{1}{(mr)^{3/2}}
e^{-2mr} 
$
\\[10pt]
4 & 
$- \frac{\lambda^2}{128\pi^{5/2}m^{3/2}}\frac{1}{r^{5/2}} e^{-2mr}
$
\\[10pt]
\hline
\end{tabular}
\vspace{6pt}
\caption{Leading asymptotic behavior as $r \rightarrow \infty$ 
of effective potential for $\phi^4$ theory in $D$ dimensions 
(\sect{sec:Veff}).}
\label{table:Veff}
\end{center}
\end{table}

\subsection{Mass Renormalization} 
\label{sec:renorm}

In this section, we calculate the renormalized (or physical) mass
of the discretized theory in the perturbative regime.
First, in \sect{sec:discth}, we derive the action and the
propagator. Next, in \sect{sec:mass}, we use (a hybrid form of)
perturbation theory to obtain the mass to second order in the 
coupling.

\subsubsection{The Discretized Theory} \label{sec:discth}

In quantum field theory, the Lagrangian formulation is generally more 
convenient than the equivalent Hamiltonian formulation.
The Lagrangian (density) ${\cal L}(\phi,\partial_\mu\phi)$ defines the 
field theory and is related to the Hamiltonian density 
${\cal H}(\pi,\phi)$ by 
\begin{equation*}
{\cal H} = \pi \dot{\phi} - {\cal L} \,,
\end{equation*}
where $\pi = \partial{\cal L}/\partial\dot{\phi}$ is the conjugate
momentum density.
For the scalar $\phi^4$ quantum field theory,
\begin{equation} \label{eq:L}
{\cal L} = \frac{1}{2}\partial_\mu\phi\partial^\mu\phi
           - \frac{1}{2} m_0^2 \phi^2 - \frac{1}{4 !}\lambda_0\phi^4 \,,
\end{equation}
where $\mu = 0, 1, \ldots, d$ in $D=d+1$ spacetime dimensions.
Then the field and coupling have the following mass dimensions:
\begin{equation} \label{eq:dim}
 \left[ \phi \right] = \frac{D-2}{2} \,, \,\,\,
 \left[ \lambda_0 \right] = 4-D \,. 
\end{equation}
We shall discretize the spatial dimensions, that is, put them on
a lattice with spacing $a$,
\begin{equation}
a \mathbb{Z}^d = \big\{ \mathbf{x} \, \big| \, x_i/a \in \mathbb{Z} \big\} \,.
\end{equation}
The time dimension will be left continuous.
Spatial derivatives become (forward and backward) difference operators,	
\begin{eqnarray}
\Delta_i^f f(x) & = & 
\frac{1}{a} \left( f(x + a \hat{\imath}) - f(x) \right) \,, \\
\Delta_i^b f(x) & = & 
\frac{1}{a} \left( f(x) - f(x - a \hat{\imath})\right) \,, \nn
\end{eqnarray}
so that
\begin{eqnarray} \label{delsq}
-\nabla_a^2 f(x) & \equiv & - \Delta_i^b \Delta_i^f f(x) \nn \\
& = & \sum_{i=1}^d \frac{1}{a^2}
 \left( 2 f(x) - f(x + a\hat{\imath}) - f(x - a\hat{\imath}) \right)
\,.
\end{eqnarray}

The free-theory action becomes
\begin{eqnarray}
S_{\rm free} & = & -\frac{1}{2} \iint dx^0 dy^0 \sum_{\mathbf{x},\mathbf{y}}
a^{2d} \phi(x) \big(\partial_t^2 -\nabla_a^2 + m_0^2 \big)_{x,y} \phi(y)
\,,
\end{eqnarray}
where 
\begin{equation}
\big(\partial_t^2 -\nabla_a^2 + m_0^2 \big)_{x,y} 
= a^{-d} \big(\partial_t^2 -\nabla_a^2 + m_0^2 \big) \,
  \delta_{\mathbf{x},\mathbf{y}} \, \delta(x^0-y^0) \,.
\end{equation}
$G(x,y;a)$ is the inverse of  
$\left(\partial_t^2 -\nabla_a^2 + m_0^2 \right)_{x,y}$, that is,
\begin{equation}
\int dy^0 \sum_{\mathbf{y}}
a^{d} \big(\partial_t^2 -\nabla_a^2 + m_0^2 \big)_{x,y} G(y,z;a)
= a^{-d} \delta_{\mathbf{x},\mathbf{z}} \, (-i) \delta(x^0-z^0) \,.
\end{equation}
Using the Fourier transform
\begin{equation}
G(x,y;a) = \int \frac{dp^0}{2\pi} 
\int_{-\pi/a}^{\pi/a} \frac{d^dp}{(2\pi)^d}
e^{-i p\cdot(x-y)} \tilde{G}(p;a) \,,
\end{equation}
we obtain the propagator
\begin{equation} \label{propagatord}
\tilde{G}(p;a) = \frac{i}{(p^0)^2 
- \sum_{i=1}^d \frac{4}{a^2}\sin^2\left(\frac{a p^i}{2}\right)
- m_0^2} \,.
\end{equation}
In the limit $a \rightarrow 0$, we recover the familiar
propagator of the continuum theory,
\begin{equation} \label{propagatorc}
\tilde{G}(p) = \frac{i}{p^2-m_0^2} \,.
\end{equation}
The Lagrangian of the interacting theory with spatial dimensions 
put on a lattice is therefore
\begin{equation} \label{discL}
{\cal L}^{(0)} =  \frac{1}{2}(\partial_t \phi)^2
 + \frac{1}{2} \phi \nabla_a^2 \phi - \frac{1}{2} m_0^2\phi^2 
 - \frac{\lambda_0}{4 !}\phi^4
 \,.
\end{equation}

\subsubsection{The Physical Mass} \label{sec:mass}

The analysis of the adiabatic turn-on procedure involves the physical mass. 
A suitable expression for this can be obtained from a hybrid form of
perturbation theory, namely partially renormalized perturbation theory,
in which we use the bare coupling and field but the renormalized mass.
From \eq{discL}, we have
\begin{eqnarray} 
{\cal L}^{(0)} 
& = & 
 \frac{1}{2}(\partial_t \phi)^2
 + \frac{1}{2} \phi \nabla_a^2 \phi - \frac{1}{2} m^2\phi^2 
 - \frac{\lambda_0}{4 !}\phi^4 - \frac{1}{2} \delta_m\phi^2
\,,
\end{eqnarray}
where the mass counterterm is
\begin{equation}
\delta_m \equiv \lambda_0 \mu = m_0^2 - m^2 \,.
\end{equation}

The shift in the mass is determined by one-particle irreducible (1PI)
diagrams. These are diagrams that remain connected after any single line
is cut.
At first order in $\lambda_0$, the 1PI 
insertions into the propagator give 
\begin{eqnarray}
-i M(p)
& = & 
\begin{array}{l} \includegraphics[width=0.6in]{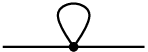} 
\end{array} 
+
\begin{array}{l} \includegraphics[width=0.6in]{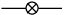} 
\end{array} 
 + O(\lambda_0^2)
\,,
  \label{eq:diags}
\end{eqnarray}
where the second diagram is the counterterm.
The calculation of the one-loop diagram, given in Appendix~B, 
implies that
\begin{eqnarray}
m_0^2 & = & m^2 + \lambda_0 \mu \,,
\nonumber \\ 
\mu & = & 
\left\{
\begin{array}{ll}
-\frac{1}{8\pi} \log\Big(\frac{64}{m^2a^2}\Big) 
+ \cdots \,,
 & \text{for $D=2$},\\
[5pt]
-\frac{r_0^{(2)}}{16\pi^2}\frac{1}{a} 
+ \cdots \,, 
& \text{for $D=3$},\\
[5pt]
-\frac{r_0^{(3)}}{32\pi^3}\frac{1}{a^2}
+ \cdots \,, 
& \text{for $D=4$},
\end{array}
\right.
   \label{eq:path}
\end{eqnarray}
where 
\begin{equation}
r_0^{(2)} = 25.379\ldots \,, \qquad
r_0^{(3)} = 112.948\ldots \,.
\end{equation}
Equation~(\ref{eq:path}) determines how the bare coupling $\lambda_0$
and bare mass $m_0$ are related if the physical mass $m$ has a specified
value to one-loop order.

At order $\lambda_0^2$, the 1PI amplitude has the additional
contributions
\begin{eqnarray}
\begin{array}{l} \includegraphics[width=0.6in]{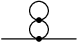} 
\end{array} 
+
\begin{array}{l} \includegraphics[width=0.6in]{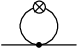} 
\end{array} 
+
\begin{array}{l} \includegraphics[width=0.6in]{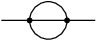} 
\end{array} 
\,.
\end{eqnarray}
The renormalization condition satisfied at first order in $\lambda_0$ 
implies that the first two diagrams cancel. The calculation of the 
remaining two-loop diagram (see Appendix~B) 
implies that 
\begin{eqnarray}
m^2 & = & 
\left\{
\begin{array}{ll}
(m^{(1)})^2 - \frac{\lambda_0^2}{384 (m^{(1)})^2} 
+ \cdots \,,
 & \text{for $D=2$},\\
[5pt]
(m^{(1)})^2 + \frac{\lambda_0^2}{96\pi^2} \log(m^{(1)} a)  
+ \cdots \,, 
& \text{for $D=3$},\\
[5pt]
(m^{(1)})^2 - \frac{r_1^{(3)}}{1536\pi^7} \frac{\lambda_0^2}{a^2} 
+ \cdots \,, 
& \text{for $D=4$},
\end{array}
\right.
\end{eqnarray}
where $m^{(1)}$ denotes the renormalized (physical) mass at one-loop 
order, namely, the quantity that is kept constant when one follows the
path specified by \eq{eq:path}. (The constant $r_1^{(3)}$ is defined in
Appendix~B.) 
If, instead, one keeps $m_0$ constant
while turning on $\lambda_0$, then the expression in terms of $m_0$,
that is, the result from bare perturbation theory equivalent to the above, 
is more relevant: 
\begin{eqnarray}
m^2 & = & 
\left\{
\begin{array}{ll}
m_0^2 + \frac{\lambda_0}{8\pi} \log\Big(\frac{64}{m_0^2a^2}\Big) 
-  \frac{\lambda_0^2}{64\pi^2 m_0^2} \log\Big(\frac{64}{m_0^2a^2}\Big)
- \frac{\lambda_0^2}{384 m_0^2} 
+ \cdots \,,
 & \text{for $D=2$},\\
[5pt]
m_0^2 + \frac{r_0^{(2)}}{16\pi^2}\frac{\lambda_0}{a}
- \frac{r_0^{(2)}}{256\pi^3}\frac{\lambda_0^2}{m_0a}
+ \frac{\lambda_0^2}{96\pi^2} \log(m_0 a)  
+ \cdots \,, 
& \text{for $D=3$},\\
[5pt]
m_0^2 + \frac{r_0^{(3)}}{32\pi^3}\frac{\lambda_0}{a^2}
+  \frac{r_0^{(3)}}{512\pi^5}\frac{\lambda_0^2}{a^2}\log(m_0 a)
- \frac{r_1^{(3)}}{1536\pi^7} \frac{\lambda_0^2}{a^2} 
+ \cdots \,, 
& \text{for $D=4$}.
\end{array}
\right.
\end{eqnarray}

\medskip

\subsection{Effective Field Theory} \label{sec:eft}

The formalism of effective field theories (EFTs) is typically used to
calculate observables in physically relevant theories and hence 
make predictions.
However, the influence of the EFT approach extends beyond just 
providing a tool for tackling otherwise intractable problems: indeed,
it has profoundly changed our understanding of renormalizability.  
In this work, the EFT framework is applied to determining the scaling
of lattice errors. Somewhat similar ideas were employed by Symanzik in 
the construction of improved actions in Euclidean lattice theories
\cite{Symanzik:1979ph,Symanzik:1983dc,Symanzik:1983gh}.

An effective field theory can be regarded as the low-energy 
limit of the fundamental theory under consideration. 
An EFT for a full theory is thus somewhat analogous to a Taylor series for
a function. 
The canonical example is 
Fermi theory, in which the four-fermion Hamiltonian is
\begin{eqnarray}
{\cal H}_{\rm eff} & = & \frac{4 G_F}{\sqrt{2}} 
      \left( \bar l_L \gamma^\mu \nu_L \right)
      \left( \bar u_L \gamma_\mu d_L \right) + \mbox{h.c.}
\end{eqnarray}
The modern interpretation is that this is an effective low-energy 
theory, in which the $W$ boson has been removed as an explicit,
dynamical degree of freedom. 
Pictorially, this corresponds to 
\medskip
\begin{center}
\includegraphics[width=3.5in]{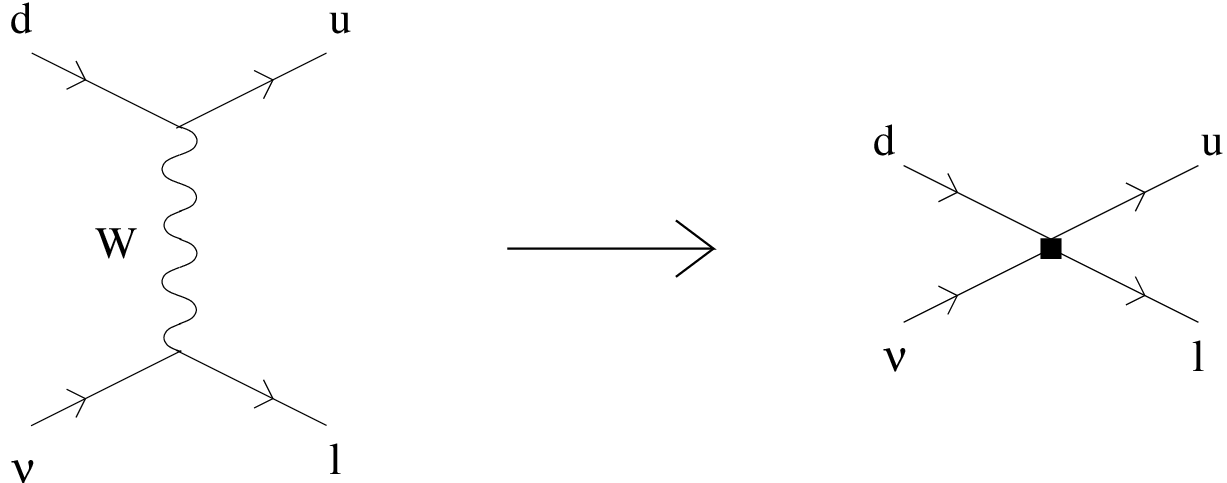}
\end{center}
\medskip
In other words, terms of order $k^2/M_W^2$ have 
been neglected in the propagator:
\begin{eqnarray} 
-\frac{i}{k^2-M_W^2} & = & \frac{i}{M_W^2} 
+ O \Big(\frac{k^2}{M_W^2}\Big)
\,.
\end{eqnarray}

An EFT is constructed from only the relevant infrared degrees of freedom 
and involves an expansion in some suitable small parameter.
The resulting effective Lagrangian will typically take the form
\begin{eqnarray}
  {\cal L}_{\text{eff}} & = & \sum_{n\geq 0} {\cal L}^{(n)} 
           \, = \, {\cal L}^{(0)} 
              + \sum_{n\geq 1}\sum_{i_n} \frac{c_{i_n}}{\Lambda^n}
                                         {\cal O}_{i_n}^{(n)}\,, 
   \label{eq:EFT}
\end{eqnarray}
where ${c_{i_n}}$ are dimensionless coefficients, $\Lambda$ is the
fundamental mass scale below which the EFT is valid, and the local operators 
${\cal O}_{i_n}^{(n)}$ have the same symmetries as the underlying theory.
This is an infinite series, but the higher the dimension of the operator 
the more powers of $\Lambda$ by which it is suppressed. In other words, the 
lowest-dimensional operators will be the most important ones.
Thus, in practice, one can truncate the series at some order dictated 
by the desired accuracy, so that one is left with a finite number of 
operators and hence a finite number of parameters ${c_{i_n}}$ to determine. 
If the underlying theory is known and weakly coupled, one may be able 
to compute the parameters. Otherwise, one can take them to be 
experimental inputs.

\subsection{Effects of Non-zero Lattice Spacing} \label{sec:a}

In this section, we determine the infinite series of operators
comprising the (complete) effective field theory whose leading terms
are ${\cal L}^{(0)}$.
The operators not included in ${\cal L}^{(0)}$ fall into three classes: 
operators of the form $\phi^{2n}$,
Lorentz-violating operators arising solely from discretization effects,
and Lorentz-violating operators due to discretization and quantum effects.
The Wilson coefficients of the ignored operators give the error 
associated with using ${\cal L}^{(0)}$ on a spatial lattice to approximate 
the continuum theory.

\subsubsection{The General Effective Theory} \label{ssec:eft}

The full (untruncated) effective Lagrangian will have every coupling respecting 
the $\phi \rightarrow -\phi$ symmetry and so will take the form
\begin{equation}
{\cal L}_{\rm eff} = {\cal L}^{(0)} + \frac{c}{6!}\phi^6 
+ c'\phi^3\partial^2\phi + \frac{c''}{8!}\phi^8 + \cdots \,.
\end{equation}
This can be simplified. First, the chain rule and integration by parts
(with boundary terms dropped) can be used to write any operator with two 
derivatives acting on different fields in the form $\phi^n \partial^2 \phi$. 
For example,
\begin{eqnarray}
\phi^2 \partial_\mu \phi \partial^\mu \phi
& = & \frac{1}{3} \partial_\mu(\phi^3) \partial^\mu \phi
\,\,\, \rightarrow \,\,\, - \frac{1}{3} \phi^3 \partial^2 \phi \,.
\end{eqnarray}
Such an operator can then be simplified via the 
equation of motion \cite{Arzt:1993gz,Georgi:1991ch}:
if this were $\partial^2 \phi + m^2 \phi = 0$, it would imply that the 
operator $\phi^3\partial^2\phi$ was redundant and could be eliminated
entirely.%
\footnote{The situation is modified by certain discretization effects,
a subtlety we shall describe later.} 
Equivalently, one can think of this reduction as making a field
redefinition on $\phi$ to eliminate the operator:
if we write 
\begin{equation}
{\cal L}_{\text{eff}} = \frac{1}{2}(\partial_\mu \phi)^2 
- \frac{1}{2} m^2 \phi^2 - \frac{\lambda}{4 !} \phi^4 
+ \eta g_1 \phi^6 +  \eta g_2 \phi^3 \partial^2 \phi \,,
\end{equation}
where $\eta$ is a small parameter, then the shift
of variables $\phi \rightarrow \phi + \eta g_2 \phi^3$
induces
\begin{equation}
{\cal L}_{\text{eff}}  \rightarrow \frac{1}{2}(\partial_\mu \phi)^2 
- \frac{1}{2} m^2 \phi^2 - \frac{\lambda'}{4 !} \phi^4 
+ \eta g_1' \phi^6 + O(\eta^2) \,.
\end{equation}
One can then iterate the process, repeatedly shifting variables to 
remove redundant terms order by order in $\eta$.

Equation~(\ref{eq:dim}) implies that
\begin{equation}
 \left[ c \right] = 6-2D \,, \,\,\,
 \left[ c'' \right] = 8-3D \,.
\end{equation}
In $D=4$ dimensions, $\left[ c \right] = -2$ and 
$\left[ c'' \right] = -4$. Since the only pertinent dimensionful
parameter is the lattice spacing, that is, $\Lambda \sim \pi/a$,
this means that $c \sim a^2$ and $c'' \sim a^4$. We see then
that, of the operators not included in the Lagrangian
${\cal L}^{(0)}$, $\phi^6$ is more significant than 
$\phi^{2n}, \, n > 3$, that is, its effects are suppressed by fewer powers
of $pa \ll 1$. 

In $D=2,3$, the scaling of the coefficients with $a$ is somewhat 
less obvious, because now the coupling $\lambda$ provides another 
dimensionful parameter (recall that $\left[\lambda\right] = 4-D$). 
To obtain the scaling of $c$, one should consider the Feynman diagram 
that generates the corresponding operator. This involves three $\phi^4$ 
vertices, so
\begin{eqnarray}
\begin{array}{l} \includegraphics[width=0.6in]{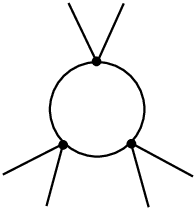} 
\end{array} 
& \sim & \lambda^3 a^{6-D} \,.
   \label{roundabout}
\end{eqnarray}
Equation~(\ref{roundabout}) refers not to the diagram's whole amplitude 
but only to its contribution to the coefficient $c$.
(Other diagrams involve higher powers of $\lambda$ and hence their
contributions are suppressed by higher powers of $a$.)
%
Likewise, the coefficient of $\phi^8$ will scale as $\lambda^4 a^{8-D}$, 
which means that it is suppressed by $a^2$ relative to the coefficient of 
$\phi^6$. 

In the following subsection, we verify these scalings in the perturbative
regime by explicit calculation.

\subsubsection{Matching} \label{sec:weak}

We must {\em match} the full theory on to the effective theory at a suitable 
energy scale. What this means is that we calculate matrix
elements in the full theory and in the effective theory. Comparing these
gives the Wilson coefficients (that is, the coefficients of the terms
in the effective Lagrangian). 
At weak coupling, that is, for sufficiently small values of the coupling, 
this matching can be done in ordinary perturbation theory.
In our case, the continuum theory and discretized theory correspond to the
full and effective theories, respectively, and the scale for matching is 
determined by the lattice spacing, $a$. 

\medskip

\paragraph{Operators induced purely by Discretization}

First, consider the matching of the two-point function.
Taylor expansion of \eq{delsq} gives
\begin{eqnarray}
-\nabla_a^2 f(x) & = &  - \nabla^2 f(x) 
- \sum_{i=1}^{d} \frac{1}{12} \partial_i^4 f(x) a^2
+ \cdots \,,
\end{eqnarray}
where $\nabla_a^2$ denotes the discrete Laplacian and
$\nabla^2$ denotes the continuum Laplacian. Thus, there
are Lorentz-violating operators induced purely by the
discretization. The leading operator of this kind is $\sum_{i=1}^{d} \phi
\partial_i^4 \phi$, which we shall denote by a box on a line.
Diagrammatically, the matching corresponds to
\begin{equation}
\includegraphics[width=2.4in]{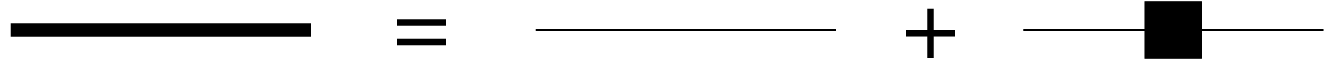}\,\,,
\end{equation}
with the full (effective) theory on the left-hand (right-hand) side
of the equation.
Letting $\tilde{c}/2$ be the coefficient of the operator
$\sum_{i=1}^{d} \phi \partial_i^4 \phi \equiv \phi \partial_{\mathbf{x}}^4 
\phi$, 
and expanding the denominator of \eq{propagatord} to order 
$p_{\mathbf{x}}^4$, we then have 
\begin{eqnarray}
\frac{i}{p^2-m^2} & = & \frac{i}{p^2+\frac{a^2}{12}p_{\mathbf{x}}^4-m^2}
+ i \tilde{c} p_{\mathbf{x}}^4 
\frac{i^2}{(p^2+\frac{a^2}{12}p_{\mathbf{x}}^4-m^2)^2}
\\
& = & \frac{i}{p^2-m^2} - i \frac{a^2 p_{\mathbf{x}}^4}{12(p^2-m^2)^2}
+ i \tilde{c} p_{\mathbf{x}}^4 \frac{i^2}{(p^2-m^2)^2} + \cdots
\\
\Rightarrow \qquad \tilde{c} & = & - \frac{a^2}{12} \,.
\end{eqnarray}
We can write down any operator of this type and calculate its Wilson 
coefficient (at the matching scale) in a similar manner. This matching
calculation is independent of the value of the coupling. Indeed, it has 
no direct relation to quantum mechanics: one can think of such operators 
as arising simply because the difference operators in the discretized theory 
are only approximately equal to the derivatives in the continuum theory.%
\footnote{However, their Wilson coefficients will depend on the scale and
mix (couple with other coefficients), in accordance with renormalization 
group equations.}

For convenience, we shall use the notation
\begin{equation}
\includegraphics[width=3in]{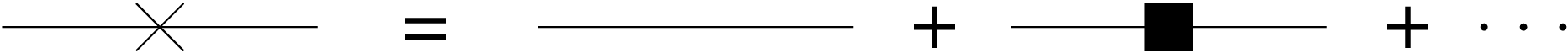} \,\,.
\end{equation}

\medskip

\paragraph{Operators induced by Quantum Effects}

As mentioned in \sect{sec:discth}, quantum effects (which correspond
to loop diagrams) induce operators of the form $\phi^{2n}$, $n \geq 3$ 
in the EFT, but not operators of the form $\phi^{2n-1}$, these being
``protected against'' by the $\phi \rightarrow -\phi$ symmetry.
Furthermore, the larger $n$ is, the greater the suppression in powers
of $a$. 

To obtain the coefficient of the operator $\phi^6$, we must
calculate the Feynman diagram
\begin{eqnarray}
\begin{array}{l} \includegraphics[width=0.6in]{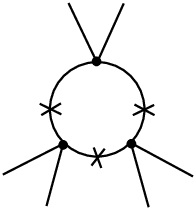} 
\end{array} 
+ \hbox{permutations} 
\,, 
\end{eqnarray}
but it suffices to consider zero external momentum, since this is a 
non-derivative operator.

Diagrammatically, matching of the six-point function is equivalent to
\begin{eqnarray}
\begin{array}{l} \includegraphics[width=0.6in]{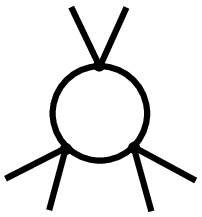}
\end{array}
& = &
\begin{array}{l} \includegraphics[width=0.6in]{diagrams/sixcross.eps}
\end{array}
+
\begin{array}{l} \includegraphics[width=0.6in]{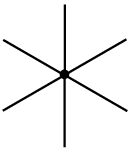}
\end{array}
\,,
\end{eqnarray}
where the diagram on the left-hand side refers to the full theory and
those on the right-hand side to the effective theory.\footnote{
The tree-level diagram with six external legs and one internal line
has the same value in the full and effective theories and thus does
not contribute to the matching.
}
(Permutations of the loop diagrams are implicitly included.)
The result to leading order in $a$ (see Appendix~C) 
is that the coefficient $c$ of the term $(c/6!)\phi^6$ is
\begin{eqnarray}
c & = & 
\left\{
\begin{array}{ll}
-\frac{45}{64\pi^5}\lambda^3 a^4\,, & \text{for $D=2$},\\[3pt]
-\frac{25\sqrt{2}}{32\pi^5}\lambda^3 a^3\,, & \text{for $D=3$},\\[3pt]
-\frac{15}{64\pi^5}(2\sqrt{3}+\pi)\lambda^3 a^2\,, & \text{for $D=4$}.
\end{array}
\right.
\end{eqnarray}

\medskip

\paragraph{Lorentz-Violating Operators induced by Quantum Effects}

Recall that we can use the equation of motion to eliminate operators
of the form $\phi^n \partial^2 \phi$. Note, however, that discretization
modifies the equation of motion from $\partial^2 \phi + m^2 \phi = 0$
to $\partial^2 \phi + m^2 \phi 
- \frac{1}{12} a^2 \phi^3 \partial_{\mathbf{x}}^4 \phi + \cdots = 0$.
Thus,
in addition to the Lorentz-violating operators induced purely by the 
discretization, there are Lorentz-violating operators induced by
loop effects (combined with the discretization). 
The leading operator of this kind is 
$\sum_i \phi^3 \partial_i^4 \phi \equiv \phi^3 \partial_{\mathbf{x}}^4 \phi$.

The precise coefficients of such operators are technically difficult to 
calculate even perturbatively (since the corresponding external momenta 
are non-zero).
However, one can determine their scalings in $a$, in the same manner
as for the other operators: the coefficient of 
$\phi^3 \partial_{\mathbf{x}}^4 \phi$ scales as $\lambda^2 a^{8-D}$.
Hence, this operator is suppressed relative to the leading operators 
in the other two classes.

%
%

\subsubsection{Strong Coupling} \label{sec:strong}

At strong coupling, perturbation theory is inapplicable and we can 
no longer calculate the Wilson coefficients explicitly. 
The $a$ dependence of the coefficients at the matching (energy) scale 
will be unchanged. However, the evolution of the coefficients to lower
scales affects the scaling, as we now describe.

For simplicity, consider first the perturbative regime in $D=4$.
When physical processes involve disparate energy scales, logarithms
of the ratio of those scales generically appear in calculations.
If the scales are widely separated, such that the logarithm is of the order
of the inverse of the expansion parameter, the perturbative expansion
will no longer be valid. 
(In our case, the magnitude of $\log(ma)$ should be compared with
$\lambda^2/(4\pi)$.)
This problem is dealt with as follows.
After one matches on to the appropriate EFT at the first energy scale, 
the (scale-dependent) Wilson coefficients are {\em run} down (that is, 
evolved) to the next energy scale by means of the renormalization group 
equations (RGEs), the solution of which resums the large logarithms.
This process is completely analogous to the case of gauge coupling 
constants, which run and obey renormalization group equations.

The RGEs are characterized by anomalous dimensions, whose effect is to 
add to (or subtract from) the power of the $a$ dependence of the coefficients. 
In the perturbative regime, the anomalous dimensions are small,
suppressed by the coupling.
At strong coupling, however, the anomalous dimensions are incalculable 
by known methods and potentially significant.  
(Nevertheless, the known existence of continuum limits in $D=2,3$
indicates that the anomalous dimensions of suppressed operators will
not override their canonical dimensions.)

\subsection{Effects of Finite Volume on EFT} \label{sec:vol}

Consider now a finite length $L$ of each dimension of the spatial lattice, 
with $\hat{L}$ lattice sites, so that $L = \hat{L}a$ (since periodic
boundary conditions are used).
Now the $d$-dimensional integral over loop momenta 
(in weak-coupling calculations)
becomes a 
$d$-dimensional sum:
\begin{eqnarray}
\int_{-\pi/a}^{\pi/a}\cdots\int_{-\pi/a}^{\pi/a}
\frac{d^dq}{(2\pi)^d}
 & \rightarrow & \frac{1}{(2\pi)^d}\frac{2^d\pi^d}{a^d \hat{L}^d}
 \sum_{q_1} \cdots \sum_{q_d}\,,
\end{eqnarray}
where each momentum component takes a finite number of values,
given by
\begin{eqnarray}
 q_i = \frac{2\pi}{a\hat{L}}n_i \,, \qquad
 n_i = -\frac{\hat{L}}{2}, -\frac{\hat{L}}{2}+1, \ldots,0,\ldots,
\frac{\hat{L}}{2}-1\,.
\end{eqnarray}

The Feynman-diagram calculation proceeds analogously to that in
the previous section; the only difference is that we now have 
a Riemann sum, which we know converges to the corresponding Riemann 
integral as $\hat{L} \rightarrow \infty$.
The difference between the Riemann sum and corresponding integral
is given by the Euler-Maclaurin summation formula 
\cite{Euler:1732,Poisson:1823},
which we therefore use to obtain (see Appendix~D) 
\begin{eqnarray}
c & = & -\frac{45}{64\pi^5}\lambda^3 a^4 
\left[ 1 + \frac{20}{3} \frac{1}{\hat{L}^2} + 
O\left(m^2a^2,\frac{m^2a^2}{\hat{L}^2},\frac{1}{\hat{L}^3}
\right)\right] \,, \qquad \text{for} \,\,\, D=2 \,.
\end{eqnarray}
Using the Euler-Maclaurin formula iteratively for two- and three-dimensional
sums, we also obtain 
\begin{eqnarray}
c & = & -\frac{5}{64\pi^5}\lambda^3 a^3 
\left[ 10\sqrt{2} + \frac{43\sqrt{2}}{\hat{L}^2} + 
O\left(m^2a^2,\frac{m^2a^2}{\hat{L}^2},\frac{1}{\hat{L}^3}
\right)\right] \,, \qquad \text{for} \,\,\, D=3 \,,
\end{eqnarray}
and
\begin{eqnarray}
c & = & -\frac{15}{128\pi^5}\lambda^3 a^2 
\left[ 2(2\sqrt{3}+\pi) + \frac{4}{9}(26\sqrt{3}+9\pi)\frac{1}{\hat{L}^2} + 
O\left(m^2a^2,\frac{m^2a^2}{\hat{L}^2},\frac{1}{\hat{L}^3}
\right)\right] \,, \qquad \text{for} \,\,\, D=4 \,.
\nonumber \\
\end{eqnarray}
Note that, in all dimensions, a finite $\hat{L}$ increases the
magnitude of the Wilson coefficient.


\subsection{Effect of Finite Volume on State Preparation} \label{sec:Veff}

The procedure for state preparation uses the fact that each particle
can be regarded as isolated and free when asymptotically separated. 
For a finite volume, this is an approximation, but the corrections
should be insignificant, since the interactions are short-range.

This claim is quantified by the effective potential $V(r)$ created 
by the interaction. Using the Born approximation from nonrelativistic
quantum mechanics, one can equate the scattering amplitude 
$i{\cal M}(2 \rightarrow 2)$ in the nonrelativistic limit
with ($-i$ times) the Fourier transform of $V(r)$.
For identical scalar particles, using the symmetric wavefunction of 
the two-particle system gives two terms, namely, 
\begin{eqnarray} \label{BornId}
\tilde{V}({\bf p}_i - {\bf p}_f) 
+ \tilde{V}({\bf p}_i + {\bf p}_f) 
\,,
\end{eqnarray}
where, in the center-of-momentum frame, the initial momenta are 
$\pm {\bf p}_i$, and the final momenta are $\pm {\bf p}_f$. 

Consider, then, the QFT amplitude. To lowest order in the coupling
$\lambda$, one simply obtains a repulsive contact term, that is, a term
proportional to the delta function. Since our interest is in the 
long-distance behavior, we must go to order $\lambda^2$.
If the incoming (outgoing) momenta for 
$2\text{-body} \rightarrow 2\text{-body}$ scattering are $p,p'$ ($k,k'$), 
then the amplitude can be expressed in terms of the Mandelstam variables 
$s = (p + p')^2$, $t = (k-p)^2$, $u = (k'-p)^2$.
The 
amplitude
\begin{eqnarray}
i{\cal M}
& = &
\begin{array}{l} \includegraphics[width=0.4in]{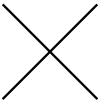}
\end{array}
+
\begin{array}{l} \includegraphics[width=0.4in]{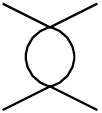}
\end{array}
+
\begin{array}{l} \includegraphics[width=0.5in]{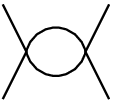}
\end{array}
+
\begin{array}{l} \includegraphics[width=0.4in]{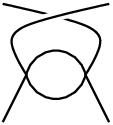}
\end{array}
+
\begin{array}{l} \includegraphics[width=0.4in]{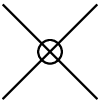}
\end{array}
+
\cdots
\end{eqnarray}
can be obtained by a straightforward perturbative calculation.
In the nonrelativistic limit and center-of-momentum frame, the $s$-channel 
contribution vanishes, and the $t$- and $u$-channels correspond to the 
two terms in \eq{BornId}.
After taking the inverse Fourier transform, we obtain 
(see Appendix~E) 
\begin{eqnarray}
  \label{eq:Veff}
V^{(2)}(r \rightarrow \infty) & = & 
\left\{
\begin{array}{ll}
-\frac{\lambda^2}{32 m^3}\frac{1}{\sqrt{\pi mr}} e^{-2mr}
+ \cdots \,, & \text{for $D=2$} \,,
\\[5pt]
-\frac{\lambda^2}{64\pi^{3/2} m}\frac{1}{(mr)^{3/2}}
e^{-2mr} + \cdots \,,  & \text{for $D=3$} \,,
\\[5pt]
- \frac{\lambda^2}{128\pi^{5/2}m^{3/2}}\frac{1}{r^{5/2}} e^{-2mr}
+ \cdots \,, &  \text{for $D=4$} \,.
\end{array}
\right.
\end{eqnarray}


\section{Conclusions}

In this paper, we have established an efficient quantum algorithm for 
determining scattering amplitudes in a scalar quantum field theory. 
In particular, the algorithm uses adiabatic turn-on to perform the 
crucial step of preparing interacting wavepacket states.  
It also discretizes space; a detailed analysis of complexity has 
addressed, among other issues, discretization errors and the continuum 
limit, a fundamentally important aspect of quantum field theory.

Our quantum algorithm provides exponential speedups over
the fastest known classical algorithms.
Specifically, it applies to both weakly and strongly interacting theories,
with a run-time that is polynomial in the desired precision, 
as well as the number of particles and their energy
(see \eq{Gtotal} and Table~\ref{strongtable}).
In contrast, standard methods in quantum field theory cannot
generally be used at strong coupling, or beyond a certain precision.

We have focussed upon massive scalar $\phi^4$ theory in 
spacetime of four and fewer dimensions.
In future work, we shall extend our results by considering such
problems as fermions, gauge symmetries, and massless particles. 
Our studies pave the way to a quantum algorithm for simulating the
Standard Model of particle physics. 
Such an algorithm would demonstrate that, except for quantum-gravity effects, 
the standard quantum circuit model suffices to capture completely the 
computational power of our universe.

\bigskip
\bigskip
\bigskip

\noindent \textbf{Acknowledgments:} We thank Alexey Gorshkov for
helpful discussions. 
This work was supported by NSF grant PHY-0803371, DOE grant 
  DE-FG03-92-ER40701, and NSA/ARO grant W911NF-09-1-0442.
Much of this work was done while S.J. was at the Institute for Quantum 
Information (IQI), Caltech, supported by the Sherman Fairchild Foundation. 
K.L. was supported in part by NSF grant PHY-0854782. 
He is grateful for the hospitality of the IQI, Caltech, 
during parts of this work.

\clearpage

\section*{Appendices}
\addcontentsline{toc}{section}{\hspace{0.4cm} Appendices}

\subsection*{A. Notation}
\addcontentsline{toc}{subsection}{A. \quad Notation}

\noindent
\medskip
\begin{table}[hbt!]
\begin{center}
\begin{tabular}{|c|l|}
\hline \T\B
Notation & 
Meaning \\
\hline 
\rule{0pt}{4.0ex}
$\Omega$ \T & Set of all spatial lattice points
\\[10pt]
$\Gamma$ & Set of all momentum-space lattice points
\\[10pt]
$D=d+1$ & Number of spacetime dimensions
\\[10pt]
$\mathbf{x}, \mathbf{p}$ & $d$-dimensional spatial and momentum vectors 
\\[10pt]
$x, p$ & $D$-dimensional spacetime and energy-momentum vectors
\\[10pt]
$\phi(\mathbf{x})$ & The field operator at $\mathbf{x}$
\\[10pt]
$\pi(\mathbf{x})$ & The operator canonically conjugate to
$\phi(\mathbf{x})$
\\[10pt]
$H$ & Lattice $\phi^4$ Hamiltonian 
\\[10pt]
${\cal L}$ & Lagrangian density
\\[10pt]
$a_{\mathbf{p}}^\dag, a_{\mathbf{p}}$ & Creation and annihilation
operators for momentum mode $\mathbf{p}$ of the free theory
\\[10pt]
$m_0$ & Bare mass, not to be confused with $m$, the physical mass
of a particle
\\[10pt]
$\lambda_0$ & Bare coupling, not to be confused with $\lambda$,
the physical coupling
\\[10pt]
$a$ & Lattice spacing
\\[10pt]
$L$ & Length of the lattice 
\\[10pt]
$V$ & Total volume of the lattice ($V = L^d$)
\\[10pt]
$\hat{L}$ & Number of lattice sites in one spatial dimension ($\hat{L}=L/a$)
\\[10pt]
${\cal V}$ & Total number of lattice sites ($\mathcal{V} = V/a^d$)
\\[10pt]
$\tau$ & Duration of the simulated adiabatic state preparation
\\[10pt]
$\nabla_a^2$ & Discrete Laplacian 
\\[10pt]
\hline
\end{tabular}
\vspace{6pt}
\caption{Notation}
\label{table:notation}
\end{center}
\end{table}

\clearpage


\subsection*{B. Loop Integrals for Mass Renormalization}
\addcontentsline{toc}{subsection}{B. \quad  Loop Integrals for Mass 
Renormalization}
   \label{app:mass}

At first order in $\lambda_0$, the 1PI 
insertions into the propagator give 
\begin{eqnarray}
-i M(p)
& = & 
\begin{array}{l} \includegraphics[width=0.6in]{diagrams/lineloop.eps} 
\end{array} 
+
\begin{array}{l} \includegraphics[width=0.6in]{diagrams/countercircle.eps} 
\end{array} 
 + O(\lambda_0^2)
\,.
\end{eqnarray}
The one-loop diagram gives
\begin{eqnarray}
\begin{array}{l} \includegraphics[width=0.6in]{diagrams/lineloop.eps} 
\end{array} 
& = & -i\frac{\lambda_0}{2} \int_{-\infty}^{\infty}
\int_{-\pi/a}^{\pi/a}\cdots\int_{-\pi/a}^{\pi/a} 
\frac{d^Dq}{(2\pi)^D} \frac{i}{(q^0)^2 
- \sum_{i=1}^d \frac{4}{a^2}\sin^2\big(\frac{a q^i}{2}\big)-m^2} 
\\
& = & -i\frac{\lambda_0}{4}\frac{a^{2-D}}{(2\pi)^{d}}
\int_{-\pi}^{\pi}\cdots\int_{-\pi}^{\pi} 
d^dq \frac{1}{\sqrt{\sum_{i=1}^d 4\sin^2\big(\frac{q^i}{2}\big)
+m^2 a^2}}
\,.
\end{eqnarray}
In $D=2$ dimensions, the integral is
\begin{eqnarray}
\int_{-\pi}^{\pi} dq \frac{1}{\sqrt{4\sin^2\big(\frac{q}{2}\big)
+y^2}}
& = & \frac{4}{\sqrt{4+y^2}} K\Big(\frac{4}{4+y^2}\Big)
\\
& = & \log\Big(\frac{64}{y^2}\Big) 
+ \frac{1}{8}\bigg(1+\frac{1}{2}\log\Big(\frac{y^2}{64}\Big) \bigg) y^2
+ \cdots \,,
\end{eqnarray}
where $K(x)$ is the complete elliptic integral of the first kind.
For $D=3$ and $D=4$, the integrals converge at $y=0$, so
\begin{eqnarray}
\int_{-\pi}^{\pi}\cdots\int_{-\pi}^{\pi} 
d^dq \frac{1}{\sqrt{\sum_{i=1}^d 4\sin^2\big(\frac{q^i}{2}\big)
+y^2}}
& = & r_0^{(d)} + \cdots \,,
\end{eqnarray}
where the ellipsis on the right-hand side denotes higher-order terms in $y$ 
and
\begin{equation}
r_0^{(2)} = 25.379\ldots \,, \qquad
r_0^{(3)} = 112.948\ldots \,.
\end{equation}
The renormalization condition $M(p=m) = 0$ implies that $i \delta_m$
equals the value of the one-loop diagram. Thus,
\begin{eqnarray}
m_0^2 & = & m^2 + \lambda_0 \mu \,,
\nonumber \\ 
\mu & = & 
\left\{
\begin{array}{ll}
-\frac{1}{8\pi} \log\Big(\frac{64}{m^2a^2}\Big) 
+ \cdots \,,
 & \text{for $D=2$},\\
[5pt]
-\frac{r_0^{(2)}}{16\pi^2}\frac{1}{a} 
+ \cdots \,, 
& \text{for $D=3$},\\
[5pt]
-\frac{r_0^{(3)}}{32\pi^3}\frac{1}{a^2}
+ \cdots \,, 
& \text{for $D=4$}.
\end{array}
\right.
\end{eqnarray}

At order $\lambda_0^2$, the 1PI amplitude has the additional
contributions
\begin{eqnarray}
\begin{array}{l} \includegraphics[width=0.6in]{diagrams/figure8.eps} 
\end{array} 
+
\begin{array}{l} \includegraphics[width=0.6in]{diagrams/counterloopcircle.eps} 
\end{array} 
+
\begin{array}{l} \includegraphics[width=0.6in]{diagrams/sunset.eps} 
\end{array} 
\,.
\end{eqnarray}
The renormalization condition satisfied at first order in $\lambda_0$ 
implies that the first two diagrams cancel. The remaining two-loop
diagram (with external momentum $p$) gives 
\begin{eqnarray}
\begin{array}{l} \includegraphics[width=0.6in]{diagrams/sunset.eps} 
\end{array}
& = & 
\frac{(-i\lambda_0)^2}{6} \iint 
\frac{d^Dk}{(2\pi)^D} \frac{d^Dq}{(2\pi)^D} 
\frac{i}{(k^0)^2 
- \sum_{i} \frac{4}{a^2}\sin^2\big(\frac{a k^i}{2}\big)-m^2} 
\,\,
\frac{i}{(q^0)^2 
- \sum_{i} \frac{4}{a^2}\sin^2\big(\frac{a q^i}{2}\big)-m^2} 
\nonumber \\
& & \qquad\qquad\qquad\qquad\qquad\qquad
\times
\frac{i}{(p^0+k^0+q^0)^2 
- \sum_{i} \frac{4}{a^2}\sin^2\big(\frac{a(p^i+k^i+q^i)}{2}\big)-m^2} 
\\
& & =  \frac{i\lambda_0^2}{3}
\int_0^1\!\!\!\int_0^1\!\!\!\int_0^1 dx\,dy\,dz\, \delta(x+y+z-1)
\iint \frac{d^Dk}{(2\pi)^D} \frac{d^Dq}{(2\pi)^D} 
\frac{1}{{\mathsf D}^3} \,,
\end{eqnarray}
where a Feynman-parameter integral has been introduced, with
\begin{eqnarray}
{\mathsf D} & = & 
x \left[ (k^0)^2 - \sum_{i} \frac{4}{a^2}\sin^2\big(\frac{a k^i}{2}\big) 
\right]
+ y \left[(q^0)^2 - \sum_{i} \frac{4}{a^2}\sin^2\big(\frac{a q^i}{2}\big) 
\right]
\nonumber \\
& & 
+ z \left[ (p^0+k^0+q^0)^2 
- \sum_{i} \frac{4}{a^2}\sin^2\big(\frac{a(p^i+k^i+q^i)}{2}\big) \right]
-m^2
\,.
\end{eqnarray}
To evaluate the $k^0$ and $q^0$ integrals, one can change variables:
\begin{eqnarray}
{\mathsf D} & = & \beta l_1^2 + \xi l_2^2 + \zeta \, (p^0)^2
- m^2 - x \sum_{i} \frac{4}{a^2}\sin^2\big(\frac{a k^i}{2}\big)
-y \sum_{i} \frac{4}{a^2}\sin^2\big(\frac{a q^i}{2}\big)
\nonumber \\
& &
-z \sum_{i} \frac{4}{a^2}\sin^2\big(\frac{a(p^i+k^i+q^i)}{2}\big)
\,,
\end{eqnarray}
where 
\begin{eqnarray}
l_1 & = & k^0 + \frac{z}{x+z}(q^0 + p^0) \,,  \\
l_2 & = & q^0 + \frac{xz}{xy+xz+yz}p^0 \,,  \\
\beta & = & x+z \,,  \\
\xi & = & \frac{xy+xz+yz}{x+z}\,, \\
\zeta & = & \frac{xyz}{xy+xz+yz}\,. 
\end{eqnarray}
Now,
\begin{eqnarray}
\int_{-\infty}^{\infty} \frac{dl_1}{2\pi}
\int_{-\infty}^{\infty} \frac{dl_2}{2\pi} \,
\frac{2}{(\beta l_1^2 + \xi l_2^2 - A^2)^3}
& = & 
\int_{-\infty}^{\infty} \frac{dl_{1E}}{2\pi}
\int_{-\infty}^{\infty} \frac{dl_{2E}}{2\pi} \,
\frac{2}{(\beta l_{1E}^2 + \xi l_{2E}^2 + A^2)^3}
\\
& = & 
\int_{0}^{\infty} d\rho\, \rho^2
\int_{-\infty}^{\infty} \frac{dl_{1E}}{2\pi}
\int_{-\infty}^{\infty} \frac{dl_{2E}}{2\pi} \,
e^{-\rho(\beta l_{1E}^2 + \xi l_{2E}^2 + A^2)}
\quad \\
& = & 
\int_{0}^{\infty} d\rho\, \rho^2 e^{-\rho A^2}
(4\pi\rho\beta)^{-1/2}(4\pi\rho\xi)^{-1/2}
\\
& = &  
\frac{1}{4\pi\sqrt{\beta\xi}(A^2)^2}
\,.
\end{eqnarray}
Thus, we obtain
\begin{eqnarray}
\begin{array}{l} \includegraphics[width=0.6in]{diagrams/sunset.eps} 
\end{array}
& = & 
\frac{i\lambda_0^2}{24\pi} \frac{a^{4-2d}}{(2\pi)^{2d}}
\iiint_0^1 dx\,dy\,dz\, 
\frac{\delta(x+y+z-1)}{\sqrt{xy+xz+yz}}
\int_{-\pi}^{\pi} d^dk
\int_{-\pi}^{\pi} d^dq \,\frac{1}{\Delta^2}
\,,
\end{eqnarray}
where
\begin{equation}
\Delta =
m^2 a^2 - \zeta \, (p^0)^2 a^2 
+ x \sum_{i=1}^d 4\sin^2\big(\frac{k^i}{2}\big) 
+ y \sum_{i=1}^d 4\sin^2\big(\frac{q^i}{2}\big)
+ z \sum_{i=1}^d 4\sin^2\big(\frac{ap^i+k^i+q^i}{2}\big)
\,.
\end{equation}
We shall consider this result at the point $p=(m,\mathbf{0})$, that is, 
take the renormalization condition to be
$M(p=(m,\mathbf{0})) = 0$.

As $a \rightarrow 0$, the momentum integral is convergent in $D=4$,
but becomes singular in $D=2,3$. The singular part can be extracted,
and the final result is then
\begin{eqnarray}
\left.
\begin{array}{l} \includegraphics[width=0.6in]{diagrams/sunset.eps} 
\end{array}\right|_{p=(m,\mathbf{0})}
& = &
\left\{
\begin{array}{ll}
\frac{i \lambda_0^2}{384 m^2}
+ \cdots \,,
 & \text{for $D=2$},\\
[5pt]
-\frac{i\lambda_0^2}{96\pi^2} \log(m a)
+ \cdots \,, 
& \text{for $D=3$},\\
[5pt]
\frac{i r_1^{(3)}}{1536\pi^7} \frac{\lambda_0^2}{a^2}
+ \cdots \,, 
& \text{for $D=4$},
\end{array}
\right.
\end{eqnarray}
where
\begin{eqnarray}
r_1^{(3)}
& = & 
\iiint_0^1 dx\,dy\,dz\, 
\frac{\delta(x+y+z-1)}{\sqrt{xy+xz+yz}}
\iiint_{-\pi}^{\pi} d^3k
\iiint_{-\pi}^{\pi} d^3q \,\frac{1}{[\Delta(a=0)]^2}
\\[5pt]
& \simeq & 3040 \,. \nonumber
\end{eqnarray}
Hence,
\begin{eqnarray}
m^2 & = & 
\left\{
\begin{array}{ll}
(m^{(1)})^2 - \frac{\lambda_0^2}{384 (m^{(1)})^2} 
+ \cdots \,,
 & \text{for $D=2$},\\
[5pt]
(m^{(1)})^2 + \frac{\lambda_0^2}{96\pi^2} \log(m^{(1)} a)  
+ \cdots \,, 
& \text{for $D=3$},\\
[5pt]
(m^{(1)})^2 - \frac{r_1^{(3)}}{1536\pi^7} \frac{\lambda_0^2}{a^2} 
+ \cdots \,, 
& \text{for $D=4$},
\end{array}
\right.
\end{eqnarray}
where $m^{(1)}$ denotes the renormalized (physical) mass at one-loop 
order, namely, the quantity that is kept constant when one follows the
path specified by \eq{eq:path}. 


\subsection*{C. Loop Integrals for Matching}
\addcontentsline{toc}{subsection}{C. \quad Loop Integrals for Matching}
   \label{app:match}

To obtain the coefficient of the operator $\phi^6$, we must
calculate (at zero external momentum --- since this is a non-derivative
operator) the Feynman diagram
\begin{eqnarray} \label{6ptloop}
\begin{array}{l} \includegraphics[width=0.6in]{diagrams/sixcross.eps} 
\end{array} 
+ \hbox{perms.} & =  &  
15 (-i\lambda)^3 \int_{-\infty}^{\infty}
\int_{-\pi/a}^{\pi/a}\cdots\int_{-\pi/a}^{\pi/a} 
\frac{d^Dq}{(2\pi)^D} \left(\frac{i}{q^2-m^2} \right)^3
\\
& = & 
-15 i\lambda^3 a^{6-D} \int_{-\infty}^{\infty}
\int_{-\pi}^{\pi}\cdots\int_{-\pi}^{\pi}
\frac{d^Dq}{(2\pi)^D} \frac{i}{(q^2-m^2 a^2)^3}
\\
& = & - \frac{1}{(2\pi)^d}\frac{45}{16} i\lambda^3 a^{6-D}
\int_{-\pi}^{\pi}\cdots\int_{-\pi}^{\pi} 
d^dq \frac{1}{(\mathbf{q}^2+m^2 a^2)^{5/2}} \,,
\end{eqnarray}
since
\begin{eqnarray}
\int\frac{dp^0}{2\pi} \frac{i}{((p^0)^2-A^2)^3}
& = & \frac{3}{16 A^5} \,.
\end{eqnarray}
Now, evaluation of the remaining $d$-dimensional integral
gives
\begin{eqnarray}
I_d(y) & \equiv &
\int_{-\pi}^{\pi}\cdots\int_{-\pi}^{\pi}
d^dq \frac{1}{(\mathbf{q}^2+y^2)^{5/2}}
\\[10pt]
& = & 
\left\{
\begin{array}{ll}
\frac{2(2\pi^3 + 3 \pi y^2)}{3 y^4(\pi^2+y^2)^{3/2}} \,,\,\, 
& \text{for $d=1$}\,,
\\[10pt]
\frac{4\left(2\pi^2 y 
+ (\pi^2+y^2)\sqrt{2\pi^2+y^2}\arccot\left[y\sqrt{2\pi^2+y^2}/\pi^2\right] 
\right)}
{3y^3(\pi^2+y^2)\sqrt{2\pi^2+y^2}} \,,\,\, 
& \text{for $d=2$}\,,
\\[10pt]
\frac{8\pi \arctan\left[\pi^2/\sqrt{3\pi^4 + 4\pi^2y^2 + y^4}\right]}
{y^2\sqrt{\pi^2+y^2}} \,,\,\, 
& \text{for $d=3$} \,.
\end{array}
\right.
\end{eqnarray}
The power series of these functions around $y=0$ are
\begin{eqnarray}
I_d(y) & = &
\left\{
\begin{array}{ll}
\frac{4}{3y^4} - \frac{1}{2\pi^4} + O(y^2) \,, & \text{for $d=1$},\\
[3pt]
\frac{2\pi}{3y^3} - \frac{10\sqrt{2}}{9\pi^3} + O(y^2) \,,
& \text{for $d=2$},\\
[3pt]
\frac{4\pi}{3 y^2} - \frac{2(2\sqrt{3}+\pi)}{3\pi^2} + O(y^2) \,,
& \text{for $d=3$}.
\end{array}
\right.
\end{eqnarray}
Hence,
\begin{eqnarray}
\begin{array}{l} \includegraphics[width=0.6in]{diagrams/sixcross.eps}
\end{array}
+ \hbox{perms.} & =  &
- \frac{1}{(2\pi)^{D-1}}\frac{45}{16} i\lambda^3 a^{6-D} I_{D-1}(ma)
\\ 
& = & 
\left\{
\begin{array}{ll}
-\frac{i}{2\pi}\frac{45}{16}\lambda^3 a^4
\left[ \frac{4}{3(ma)^4} - \frac{1}{2\pi^4} + O(m^2a^2)  \right] \,,
& \text{for $D=2$},
\\[6pt]
-\frac{i}{(2\pi)^2}\frac{45}{16}\lambda^3 a^3
\left[ \frac{2\pi}{3(ma)^3} - \frac{10\sqrt{2}}{9\pi^3} + O(m^2a^2) 
\right] \,, 
& \text{for $D=3$}, $\quad$
\\[6pt]
-\frac{i}{(2\pi)^3}\frac{45}{16}\lambda^3 a^2
\left[ \frac{4\pi}{3 (ma)^2} - \frac{2(2\sqrt{3}+\pi)}{3\pi^2} + 
O(m^2a^2)  \right] \,,
& \text{for $D=4$}.
\end{array}
\right.
   \label{eq:6pt}
\end{eqnarray}

Diagrammatically, matching of the six-point function is equivalent to
\begin{eqnarray}
\begin{array}{l} \includegraphics[width=0.6in]{diagrams/sixthick.eps}
\end{array}
& = &
\begin{array}{l} \includegraphics[width=0.6in]{diagrams/sixcross.eps}
\end{array}
+
\begin{array}{l} \includegraphics[width=0.6in]{diagrams/sixvertex.eps}
\end{array}
\,,
\end{eqnarray}
where the diagram on the left-hand side refers to the full theory and
those on the right-hand side to the effective theory.
(Permutations of the loop diagrams are implicitly included.)
The coefficient $c$ of the term $(c/6!)\phi^6$ is then
\begin{eqnarray}
c & = & 
 \frac{1}{(2\pi)^{D-1}}\frac{45}{16} \lambda^3 
\big( a^{6-D} I_{D-1}(ma)
 - \left[ a^{6-D} I_{D-1}(ma) \right]_{a=0} \big) \,,
\end{eqnarray}
since the expression with $a=0$ corresponds to the full theory.
Thus, to leading order in $a$,
\begin{eqnarray}
c & = & 
\left\{
\begin{array}{ll}
-\frac{45}{64\pi^5}\lambda^3 a^4\,, & \text{for $D=2$},\\[3pt]
-\frac{25\sqrt{2}}{32\pi^5}\lambda^3 a^3\,, & \text{for $D=3$},\\[3pt]
-\frac{15}{64\pi^5}(2\sqrt{3}+\pi)\lambda^3 a^2\,, & \text{for $D=4$}.
\end{array}
\right.
\end{eqnarray}

\medskip

Note that the sum of the propagator \eq{propagatord} and the series of
Lorentz-violating operators that begins with 
$\phi \partial_{\mathbf{x}}^4 \phi$ is equal to \eq{propagatorc},
which thus appears in \eq{6ptloop}.
Without these operators, \eq{propagatord} would appear instead,
and the resulting six-point function would take the form
\begin{eqnarray}
a^4\Big[\frac{c_1}{(ma)^4} + \frac{c_2}{(ma)^2} + c_3 + \cdots\Big] \,,
\,\,\, \text{for $D=2$}\,. 
\end{eqnarray} 
In particular, it would have a term of the form $a^2/m^2$, which
is divergent in the infrared. ($a$ can be regarded as an ultraviolet (UV)
regulator and $m$ as an infrared (IR) regulator.)
Since the full (continuum) theory corresponds to $a=0$, this term is 
absent from the full theory, as \eq{eq:6pt} confirms. 
However, the EFT must reproduce the IR behavior of the full theory, 
and so {\em must} have the same IR divergences. Therefore, one would 
not have the correct effective field theory. We see, then, that obtaining
the correct IR structure depends upon including the Lorentz-violating 
operators. 


\subsection*{D. Loop Sums for Matching}
\addcontentsline{toc}{subsection}{D. \quad Loop Sums for Matching}
   \label{app:vol}

The Feynman-diagram calculation proceeds analogously to that in
the previous section, with the replacement of $I_d(y)$ by 
\begin{eqnarray}
I_d(y,\hat{L}) & \equiv & \frac{2^d\pi^d}{\hat{L}^d}
 \sum_{n_1=-\frac{\hat{L}}{2}}^{\frac{\hat{L}}{2}-1} \cdots 
 \sum_{n_d=-\frac{\hat{L}}{2}}^{\frac{\hat{L}}{2}-1}
\frac{1}{\left[ \frac{4\pi^2}{\hat{L}^2}(n_1^2 + n_2^2 + \cdots + n_d^2) 
+ y^2 \right]^{5/2}}
\,,
\end{eqnarray}
where $I_d(y) = I_d(y,\infty)$.
This is simply a Riemann sum, which we know converges to the 
corresponding Riemann integral as $\hat{L} \rightarrow \infty$.

The difference between the Riemann sum and corresponding integral
is given by the Euler-Maclaurin summation formula 
\cite{Euler:1732,Poisson:1823},
which we can write in the following form.

\bigskip
\noindent
\textbf{Euler-Maclaurin summation formula}
\begin{eqnarray} \label{EulerMac2}
\frac{c}{N}\sum_{i=-N}^{N-1} f\left(\frac{ci}{N}\right)
& = & \int_{-c}^c f(x) dx 
- \frac{c}{2N}\left(f(c) - f(-c)\right)
+ \sum_{k=1}^{m} \frac{B_{2k}}{(2k)!}\left(\frac{c}{N}\right)^{2k}
                 f^{(2k-1)}(x)\Big|_{-c}^{c}
\\ 
& & - \int_{-c}^c \frac{1}{(2m+1)!} 
  B_{2m+1}\big(\left\{Nx/c\right\}\big)
  \left(\frac{c}{N}\right)^{2m+1} f^{(2m+1)}(x) dx \,,
\nonumber
\end{eqnarray}
for $m \geq 1$. Here, $B_k$ are the Bernoulli numbers,
$B_m(x)$ are the Bernoulli polynomials, and 
$\{ x \} = x - \lfloor x \rfloor$ denotes the fractional
part of $x$.

\bigskip
\bigskip

Let 
\begin{eqnarray}
f(q_1,\ldots,q_d) & = & \frac{1}{(\mathbf{q}^2 + m^2 a^2)^{5/2}} \,.
\end{eqnarray}
Then, the one-dimensional sum is given by
\begin{eqnarray}
I_1(ma,\hat{L}) & = & \int_{-\pi}^{\pi} f(q) dq  
+ \frac{1}{12} \Big(\frac{2\pi}{\hat{L}}\Big)^2 f'(q)\Big|_{-\pi}^{\pi}
+ O\Big(\Big(\frac{2\pi}{\hat{L}}\Big)^3\Big)
\\
& = & I_1(ma) - \frac{1}{\hat{L}^2}\frac{10\pi^3}{3(\pi^2+m^2a^2)^{7/2}}
+ O\Big(\Big(\frac{2\pi}{\hat{L}}\Big)^3\Big) \,.
\end{eqnarray}
Hence,
\begin{eqnarray}
\begin{array}{l} \includegraphics[width=0.6in]{diagrams/sixcross.eps}
\end{array}
+ \hbox{perms.} & =  &
-i\frac{45}{32\pi}\lambda^3 a^4
\left[ \frac{4}{3(ma)^4} - \frac{1}{2\pi^4} + O(m^2a^2)  
- \frac{10}{3\pi^4}\frac{1}{\hat{L}^2}(1+O(m^2a^2)) + 
   O\left(\frac{1}{\hat{L}^3}\right)\right] \,,
\nonumber \\
\end{eqnarray}
which implies
\begin{eqnarray}
c & = & -\frac{45}{64\pi^5}\lambda^3 a^4 
\left[ 1 + \frac{20}{3} \frac{1}{\hat{L}^2} + 
O\left(m^2a^2,\frac{m^2a^2}{\hat{L}^2},\frac{1}{\hat{L}^3}
\right)\right] \,, \qquad \text{for} \,\,\, D=2 \,.
\end{eqnarray}

\medskip
To calculate the two-dimensional sum, we use \eq{EulerMac2}
twice. 
\begin{eqnarray}
I_2(ma,\hat{L}) & = & \frac{2\pi}{\hat{L}} 
\sum_{i=-\frac{\hat{L}}{2}}^{\frac{\hat{L}}{2}-1} 
\left[ \int\limits_{-\pi}^{\pi} f\left(\frac{2\pi i}{\hat{L}}, q_2\right) dq_2
+ \frac{1}{6} \left(\frac{2\pi}{\hat{L}}\right)^2 
  f_{q_2}\left(\frac{2\pi i}{\hat{L}}, \pi\right) 
+ O\left(\frac{1}{\hat{L}^3} \right) \right]
\\
& = & 
\int\limits_{-\pi}^{\pi}\int\limits_{-\pi}^{\pi} f(q_1, q_2) d^2 q
+ \frac{1}{3} \left(\frac{2\pi}{\hat{L}}\right)^2
    \int\limits_{-\pi}^{\pi} f_{q_2}(q_1, \pi) dq_1
+ O\left(\frac{1}{\hat{L}^3} \right)
\\
& = & 
I_2 (ma) - \frac{20\pi^3}{3\hat{L}^2} \int\limits_{-\pi}^{\pi}
\frac{1}{(q_1^2 + \pi^2 + m^2 a^2)^{7/2}} dq_1
+ O\left(\frac{1}{\hat{L}^3} \right)
\,.
\end{eqnarray}
Evaluating the second term (exactly) and expanding the result
around $ma=0$, we obtain
\begin{eqnarray}
\begin{array}{l} \includegraphics[width=0.6in]{diagrams/sixcross.eps}
\end{array}
+ \hbox{perms.} & =  &
-i\frac{45}{64\pi^2}\lambda^3 a^3
\left[ \frac{2\pi}{3(ma)^3} - \frac{10\sqrt{2}}{9\pi^3} + O(m^2a^2) 
- \frac{43\sqrt{2}}{9\pi^3}\frac{1}{\hat{L}^2}(1+O(m^2a^2)) 
\right. \nonumber \\
& & \qquad \qquad \qquad \left. 
+ O\left(\frac{1}{\hat{L}^3}\right)\right] \,,
\end{eqnarray}
which implies
\begin{eqnarray}
c & = & -\frac{5}{64\pi^5}\lambda^3 a^3 
\left[ 10\sqrt{2} + \frac{43\sqrt{2}}{\hat{L}^2} + 
O\left(m^2a^2,\frac{m^2a^2}{\hat{L}^2},\frac{1}{\hat{L}^3}
\right)\right] \,, \qquad \text{for} \,\,\, D=3 \,.
\end{eqnarray}

\medskip
Finally,
\begin{eqnarray}
I_3(ma,\hat{L}) & = & \frac{2\pi}{\hat{L}} 
\sum_{i=-\frac{\hat{L}}{2}}^{\frac{\hat{L}}{2}-1}
\left[ \int\limits_{-\pi}^{\pi}\int\limits_{-\pi}^{\pi} 
f\left(\frac{2\pi i}{\hat{L}},q_2,q_3\right) dq_2 dq_3 
+ \frac{1}{3}\left(\frac{2\pi}{\hat{L}}\right)^2 \int\limits_{-\pi}^{\pi}
  f_{q_3}\left(\frac{2\pi i}{\hat{L}},q_2,\pi\right) dq_2
+ O\left(\frac{1}{\hat{L}^3}\right) \right]
\nonumber \\
\\
& = &
\int\limits_{-\pi}^{\pi}\int\limits_{-\pi}^{\pi}\int\limits_{-\pi}^{\pi}
f(q_1,q_2,q_3) d^3q 
+ \frac{1}{2} \left(\frac{2\pi}{\hat{L}}\right)^2
\int\limits_{-\pi}^{\pi}\int\limits_{-\pi}^{\pi}
f_{q_3}(q_1,q_2,\pi) dq_1 dq_2
+ O\left(\frac{1}{\hat{L}^3}\right)
\\
& = &
I_3(ma) - \frac{10\pi^3}{\hat{L}^2}
\int\limits_{-\pi}^{\pi}\int\limits_{-\pi}^{\pi}
\frac{1}{(q_1^2 + q_2^2 + \pi^2 + m^2 a^2)^{7/2}} dq_1 dq_2
+ O\left(\frac{1}{\hat{L}^3}\right)
\,.
\end{eqnarray}
Evaluating the second term (exactly) and expanding the result
around $ma=0$, we obtain
\begin{eqnarray}
\begin{array}{l} \includegraphics[width=0.6in]{diagrams/sixcross.eps}
\end{array}
+ \hbox{perms.} & =  &
-i\frac{45}{128\pi^3}\lambda^3 a^2
\left[ \frac{4\pi}{3 (ma)^2} - \frac{2(2\sqrt{3}+\pi)}{3\pi^2} + 
O(m^2a^2) \right.
\nonumber \\
& & \qquad\qquad\qquad \left.
- \frac{4(26\sqrt{3}+9\pi)}{27\pi^2}\frac{1}{\hat{L}^2}(1+O(m^2a^2)) + 
   O\left(\frac{1}{\hat{L}^3}\right)\right] \,, \qquad\quad
\end{eqnarray}
which implies
\begin{eqnarray}
c & = & -\frac{15}{128\pi^5}\lambda^3 a^2 
\left[ 2(2\sqrt{3}+\pi) + \frac{4}{9}(26\sqrt{3}+9\pi)\frac{1}{\hat{L}^2} + 
O\left(m^2a^2,\frac{m^2a^2}{\hat{L}^2},\frac{1}{\hat{L}^3}
\right)\right] \,, \qquad \text{for} \,\,\, D=4 \,.
\nonumber \\
\end{eqnarray}

Note that, in all dimensions, a finite $\hat{L}$ increases the
magnitude of the Wilson coefficient.


\subsection*{E. Integrals for Effective Potential}
\addcontentsline{toc}{subsection}{E. \quad Integrals for Effective Potential}
   \label{app:Veff}

If the incoming (outgoing) momenta for 
$2\text{-body} \rightarrow 2\text{-body}$ scattering are $p,p'$ ($k,k'$), 
then the amplitude can be expressed in terms of the Mandelstam variables 
$s = (p + p')^2$, $t = (k-p)^2$, $u = (k'-p)^2$.
A straightforward textbook calculation (in renormalized perturbation theory) 
of the one-loop amplitude
\begin{eqnarray}
i{\cal M}_2
& = &
\begin{array}{l} \includegraphics[width=0.4in]{diagrams/four.eps}
\end{array}
+
\begin{array}{l} \includegraphics[width=0.4in]{diagrams/fourloop1.eps}
\end{array}
+
\begin{array}{l} \includegraphics[width=0.5in]{diagrams/fourloop2.eps}
\end{array}
+
\begin{array}{l} \includegraphics[width=0.4in]{diagrams/fourloop3.eps}
\end{array}
+
\begin{array}{l} \includegraphics[width=0.4in]{diagrams/fourcirc.eps}
\end{array}
\end{eqnarray}
gives
\begin{eqnarray} \label{eq:2to2}
i{\cal M}_2 & = & -i\lambda 
+ i\frac{\lambda^2}{2}\frac{\Gamma(2-\frac{D}{2})}{(4\pi)^{D/2}}
\int_0^1 dx \left[ 
\frac{1}{(m^2 -x(1-x)s)^{2-D/2}}  
- \frac{1}{(m^2 -x(1-x)4m^2)^{2-D/2}}  
\right.
\nonumber
\\ 
& & \qquad\qquad\qquad\qquad\qquad\qquad
+ \frac{1}{(m^2 -x(1-x)t)^{2-D/2}}  
- \frac{1}{(m^2)^{2-D/2}}  
\nonumber \\ 
& & \qquad\qquad\qquad\qquad\qquad\qquad
\left.
+ \frac{1}{(m^2 -x(1-x)u)^{2-D/2}}  
- \frac{1}{(m^2)^{2-D/2}}  
\right] \,.
\end{eqnarray}
(Here, the following renormalization condition has been used:
$i{\cal M} = -i\lambda$ at  $s=4m^2$, $t=u=0$. This corresponds to
defining $\lambda$ as the magnitude of the amplitude at zero
$d$-momentum.)
For example, in $D=4$ dimensions,
\begin{eqnarray}
i{\cal M}_2 & = & -i\lambda -i\frac{\lambda^2}{32\pi^2}\int_0^1 dx
\left[
\log\left(\frac{m^2-x(1-x)s}{m^2-x(1-x)4m^2}\right)
+\log\left(\frac{m^2-x(1-x)t}{m^2}\right)
+\log\left(\frac{m^2-x(1-x)u}{m^2}\right)
\right] \,. 
\nonumber \\
\end{eqnarray}

In the nonrelativistic limit, $p = (m,{\bf p})$, etc., so that
$s = 4m^2 - |{\bf p} + {\bf p}'|^2$, 
$t = - |{\bf k} - {\bf p}|^2$,
$u = - |{\bf k}' - {\bf p}|^2$.
Thus, in the center-of-momentum frame, the $s$-channel contribution 
(namely, the first line of the integrand of \eq{eq:2to2}) vanishes, 
and the $t$- and $u$-channels correspond to the two terms in 
\eq{BornId}.
One must also divide by $(\sqrt{2m})^4$ to account for the difference
in the normalization of states.

Hence, the potential is $V({\bf x}) = V^{(1)}({\bf x}) + V^{(2)}({\bf x})$,
where
\begin{eqnarray}
V^{(1)}({\bf x}) & = & \frac{\lambda}{4m^2} \delta^d({\bf x}) \,
\end{eqnarray}
and
\begin{eqnarray}
V^{(2)}({\bf x}) & = & -\frac{\lambda^2}{4m^2}
\frac{\Gamma(2-\frac{D}{2})}{(4\pi)^{D/2}}
\int_0^1 dx \int_{-\infty}^{\infty} \frac{d^{D-1} q}{(2\pi)^{D-1}}
e^{i{\bf q}\cdot{\bf x}}
\left[  \frac{1}{(m^2 +x(1-x){\bf q}^2)^{2-D/2}}  
- \frac{1}{(m^2)^{2-D/2}}  
 \right] \,.
\nonumber \\
\end{eqnarray}

In $D=2$ dimensions,
\begin{eqnarray}
V^{(2)}(r>0) & = & - \frac{\lambda^2}{16\pi m^2}
\int_0^1 dx \int_{-\infty}^{\infty} \frac{d q}{2\pi}
\frac{e^{iqr}}{(m^2 +x(1-x) q^2)}
\\
& = & -\frac{\lambda^2}{32\pi m^3}
\int_0^1 dx \frac{e^{-mr/\sqrt{x(1-x)}}}{\sqrt{x(1-x)}}
\\
& = & 
-\frac{\lambda^2}{32 m^3}\frac{1}{\sqrt{\pi mr}} e^{-2mr}
+ \cdots, \qquad \text{as $r \rightarrow \infty$} \,.
\end{eqnarray}
In the last line, the asymptotic evaluation of the integral
was obtained by Laplace's method.

In $D=3$ dimensions,
\begin{eqnarray}
V^{(2)}(r>0) & = & 
-\frac{\lambda^2}{32\pi m^2}\frac{1}{(2\pi)^2}
\int_0^1 dx \int_{0}^{\infty} dq \, q \int_{0}^{2\pi} d\theta
\frac{e^{i q r \cos\theta}}{\sqrt{m^2 +x(1-x) q^2}}
\\
& = &
-\frac{\lambda^2}{64\pi^2 m^2 r}
\int_0^1 dx \frac{e^{-mr/\sqrt{x(1-x)}}}{\sqrt{x(1-x)}}
\\
& = & 
-\frac{\lambda^2}{64\pi^{3/2} m}\frac{1}{(mr)^{3/2}}
e^{-2mr} + \cdots, \qquad \text{as $r \rightarrow \infty$} \,.
\end{eqnarray}

In $D=4$ dimensions,
\begin{eqnarray}
V^{(2)}(r>0) & = & 
\frac{\lambda^2}{64\pi^2 m^2}\frac{1}{(2\pi)^3}
\int_{0}^{\infty} dq \, q^2 
\int_{-1}^{1} du \, e^{iqru} \int_{0}^{2\pi} d\varphi
\int_0^1 dx
\log\left(1 + \frac{x(1-x)}{m^2} q^2\right) \qquad
\\
& = & 
\frac{\lambda^2}{64\pi^2 m^2}\frac{1}{(2\pi)^2 ir}
\int_{-\infty}^{\infty} dq \, q e^{iqr} 
\int_0^1 dx
\log\left(1 + \frac{x(1-x)}{m^2} q^2\right)
\,.
   \label{eq:V4D}
\end{eqnarray}
To evaluate the integral with respect to $q$, complete the contour 
in the upper half-plane. There is a branch cut on the imaginary axis
from $q=2im$ to $q = i \infty$ (since the logarithm has a branch cut
where its argument is negative).

Let 
\begin{eqnarray}
f(p^2) &  \equiv & \int_0^1 dx
\log\left(1 - \frac{x(1-x)}{m^2} p^2\right) \,.
\end{eqnarray}
Then
\begin{eqnarray}
\int_{-\infty}^{\infty} dq \, q e^{iqr} 
\int_0^1 dx
\log\left(1 + \frac{x(1-x)}{m^2} q^2\right)
& = & \int_{-\infty}^{\infty} dq \, q e^{iqr} f(-q^2)
\\
& = & -2i \int_{2m}^{\infty} d\tilde{q} \, \tilde{q} e^{-\tilde{q} r}
\text{Im}[f(\tilde{q}^2-i\epsilon)] \,,
\end{eqnarray}
where $\tilde{q} = -iq$.

For a fixed $\tilde{q}^2$, there are contributions to $\text{Im} f$ when
$x$ is in the range $x_{-} < x < x_{+}$, where 
\begin{eqnarray}
x_{\pm} & = & \frac{1}{2} \pm \frac{1}{2}\sqrt{1-\frac{4m^2}{\tilde{q}^2}}
\,.
\end{eqnarray}
Hence, 
\begin{eqnarray}
\text{Im}[f(\tilde{q}^2-i\epsilon)]  & = & \pi \int_{x_{-}}^{x_{+}} dx
= \pi \sqrt{1-\frac{4m^2}{\tilde{q}^2}}
\,,
\end{eqnarray}
and
\begin{eqnarray}
\int_{-\infty}^{\infty} dq \, q e^{iqr} 
\int_0^1 dx
\log\left(1 + \frac{x(1-x)}{m^2} q^2\right)
& = &
-2\pi i \int_{2m}^{\infty} d\tilde{q} \, \tilde{q} e^{-\tilde{q} r}
\sqrt{1-\frac{4m^2}{\tilde{q}^2}}
\,.
\end{eqnarray}
Substituting this result into \eq{eq:V4D}, we obtain
\begin{eqnarray}
V^{(2)}(r>0) & = & 
-\frac{\lambda^2}{128\pi^3 m^2 r}
\int_{2m}^{\infty} d\tilde{q} \, \tilde{q} e^{-\tilde{q} r}
\sqrt{1-\frac{4m^2}{\tilde{q}^2}}
\\
& = &
- \frac{\lambda^2}{128\pi^{5/2}m^{3/2}}\frac{1}{r^{5/2}} e^{-2mr}
+ \cdots, \qquad \text{as $r \rightarrow \infty$} \,.
\end{eqnarray}



\subsection*{F. Minimal Qubit Requirement }
\addcontentsline{toc}{subsection}{F. \quad Minimal Qubit Requirement }
   \label{app:qubit}

In this appendix, we estimate the number of qubits needed for a minimal
non-trivial demonstration of our algorithm. Specifically, to simulate
a $2 \to 4$ scattering process in $1+1$ dimensions, on the order of a 
thousand to ten thousand qubits should suffice, depending on
the desired level of precision (see Fig.~\ref{fancybits}). Note that
we assume the qubits and quantum gates are noiseless. A large number
of noisy physical qubits can substitute for a smaller number of
perfect ``logical'' qubits through the use of quantum error correction
(see, for example, \cite{Gottesman}). The ratio of physical to
logical qubits depends not only on the quantum error-correction scheme
but also on the particular implementation chosen (for example, trapped
ions versus superconducting qubits) and the experimental techniques
for reducing sources of noise, which are beyond the scope of this
paper.

We must choose the energy, $E$, to be at least $4m$ so that $2 \to 4$ 
scattering is kinematically allowed. (Actually, one should choose $E$ 
slightly larger than $4m$ so that the process is not suppressed by the 
lack of phase space to scatter into.) Thus, we choose $E=5m$, which 
for ingoing particles of momenta $\pm p$ implies
\begin{equation}
\label{pm}
p/m \simeq 2.
\end{equation}
By our EFT analysis, discretization errors in scattering events 
are of order $(pa)^2$. Setting this to $\epsilon$ and using \eq{pm}, 
we obtain 
\begin{equation}
\label{ma}
ma \simeq \frac{\sqrt{\epsilon}}{2} \,.
\end{equation}

To estimate the total number of qubits, we must determine what
are sufficient numbers of lattice sites and qubits per site. 
Without a sufficient number of lattice sites, the incoming and outgoing
particles cannot be well separated. Thus, the interparticle force will
be non-negligible, and the in and out states created will not be a good
approximation to the asymptotic in and out states that define the $S$-matrix.

%
The interparticle potential at large $r$ is given by \eq{eq:Veff}.
We want substantial scattering to occur 
when the particles most closely approach one another, but
scattering not to occur when the particles are separated in their
in and out states. The expectation value of the distance of closest approach
is on the order of the wavepacket width. The wavepacket width should
not be chosen much larger than the range of the interaction ($\sim
1/m$), or scattering will be unlikely. We thus demand that $F(r) \ll
F(1/m)$, where $F$ is the (magnitude of) the interparticle
force. Quantitatively, we can demand 
\begin{equation}
\label{fcond}
\frac{F(r)}{F(1/m)} \leq \epsilon \,.
\end{equation}
By numerically solving \eq{fcond} and using \eq{ma}, one obtains the
left inset in Fig.~\ref{fancybits}. 

\begin{figure}
\begin{center}
\includegraphics[width=0.6\textwidth]{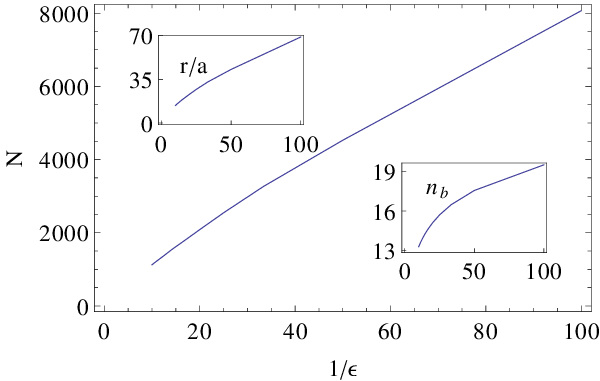}
\vspace{6pt}
\caption{\label{fancybits} The required number of qubits is shown for
  $2 \to 4$ scattering as a function of $1/\epsilon$. The insets
  display the interparticle separation in lattice units ($r/a$) and
  the number of qubits per site ($n_b$), each as a function of
  $1/\epsilon$. Our estimate $N$ for the total number of qubits is $6
  \times (r/a) \times n_b$. The prefactor six is chosen (somewhat
  arbitrarily) to provide enough space for four outgoing particles to
  be well separated, with an extra factor of $1.5$ to allow for the
  possibility that they are not evenly spaced.}
\end{center}
\end{figure}

Next, we estimate the necessary number of qubits per site. In
\sect{qubits}, the constant factors hidden by the big-O notation 
can easily be restored. Specifically, a sufficient choice is
\begin{equation}
\label{exactplacevalue}
n_b = \left\lceil \log_2 \left( 1+ \frac{2 a^d}{\pi} \left(1+
\sqrt{\frac{\mathcal{V}}{\epsilon}} \right)^2 \sqrt{ \langle \phi^2
  \rangle \langle \pi^2 \rangle} \right) \right\rceil.
\end{equation}
To bound $\langle \phi^2 \rangle$ and $\langle \pi^2 \rangle$, we must
estimate the bare quantities $\lambda_0$ and $m_0^2$ and then apply
Proposition \ref{mpos} or \ref{mneg}. In $D=2$, one finds
perturbatively that 
\begin{eqnarray}
\label{mrenorm}
m_0^2 & = & m^2 - \frac{\lambda}{8 \pi} \log \left( \frac{64}{m^2
  a^2} \right) + \cdots \,,\\
\lambda_0 & = & \lambda + \frac{3 \lambda^2}{8 \pi m^2} + \cdots
\,.
\end{eqnarray}
We must now choose a value of $\lambda$ that is small enough to
justify the use of $V^{(2)}$, but large enough to make classical
calculations at high precision difficult. As concrete
examples, we consider the values $\frac{\lambda}{2 \pi m^2} =
\frac{1}{3}$ and $\frac{\lambda}{2 \pi m^2} =
\frac{1}{10}$. The results show that $n_b$ is not very sensitive to
the choice of $\lambda$, as discussed below.

For $0.01 \leq \epsilon \leq 0.1$, our choices of $\lambda$ and $ma$
imply that $m_0^2 \geq 0$, and thus the applicable bound on $\langle
\phi^2 \rangle$ and $\langle \pi^2 \rangle$ is Proposition
\ref{mpos} throughout the entire adiabatic state preparation.  By
\eq{exactplacevalue} and Proposition \ref{mpos}, one obtains $n_b$ as
a function of $\epsilon$, as shown in the right inset in
Fig.~\ref{fancybits}. Replacing $\frac{\lambda}{2 \pi m^2}
= \frac{1}{3}$ with $\frac{\lambda}{2 \pi m^2} = \frac{1}{10}$ changes
$n_b$ from 20 to 19 at $\epsilon = 0.01$ and leaves $n_b$ unchanged at
13 for $\epsilon = 0.1$.
Our estimate of the total number of qubits $N=n_b \mathcal{V}$ needed to 
achieve a given precision $\epsilon$ is plotted in Fig.~\ref{fancybits}.

Finally, let us comment on the asymptotic scaling of the number of
qubits required by our algorithm. As discussed above (see \eq{ma}), 
to restrict discretization errors to order $\epsilon$ we need $ma \sim
\sqrt{\epsilon}$. Furthermore, to obtain good in and out states we
need particles to be separated by a distance $r \sim \frac{1}{m} \log
(1/\epsilon)$. Thus, the total number of lattice sites is
$\mathcal{V} \sim \left(\frac{r}{a}\right)^d = O \left(
\epsilon^{-d/2} \log^d (1/\epsilon) \right)$. The number of qubits per
site is $n_b \sim \log (1/\epsilon)$. Thus, the total number of
qubits is $N = O \left( \epsilon^{-d/2} \log^{d+1} (1/\epsilon)
\right)$. In contrast to the estimate shown in Fig.~\ref{fancybits}, 
which relies on perturbative calculations, this asymptotic scaling should 
hold at both strong and weak coupling. Note that, as mentioned in 
\sect{details}, we consider spatially localized wavepackets and the attendant
uncertainty in momentum to be physically realistic and not a source of
error, although this notion differs from the idealization used to define an
$S$-matrix.


\newpage

\bibliography{qft}

\end{document}